\pdfoutput=1

\documentclass{tlp}
\usepackage{mathptmx}

\usepackage{amsfonts,amsmath}

\usepackage{graphicx}
\usepackage{url}
\usepackage{verbatim}
\usepackage{color}
\definecolor{lgray}{gray}{0.92}
\definecolor{lblue}{rgb}{0.90,0.90,1.00}
\definecolor{lyellow}{rgb}{1.00,1.00,0.70}

\usepackage{listings}
\newtheorem{prop}{Proposition}

\newtheorem{df}{Definition}

\newtheorem{ex}{Example}

\newenvironment{codex}{\small\verbatim}{\endverbatim\normalsize}

\lstnewenvironment{code}
    {\lstset{}%
      \csname lst@SetFirstLabel\endcsname}
    {\csname lst@SaveFirstLabel\endcsname}
\newcommand{\BI}[0]{\begin{itemize}}
\newcommand{\EI}[0]{\end{itemize}}
\newcommand{\I}[0]{\item}
\newcommand{\BE}[0]{\begin{enumerate}}
\newcommand{\EE}[0]{\end{enumerate}}

\newcommand{\BX}[0]{\begin{ex}}
\newcommand{\EX}[0]{\end{ex}}
\newcommand{\BV}[0]{\begin{verbatim}}
\newcommand{\EV}[0]{\end{verbatim}}

\def \bscale1 {0.25}
\def \bscale {0.25}
\def \N {\mathbb{N}}

\newcommand{\FIG}[4]{
\begin{figure}[htbp]
\centering
{\includegraphics[scale=#3]{#4}}
\caption{#2}
\label{#1}
\end{figure}
}

\title[A Logic Programming Playground for Lambda Terms Types and Tree-based Arithmetic]
        {A Logic Programming Playground for Lambda Terms, Combinators, Types and Tree-based Arithmetic Computations
}

\author[Paul Tarau]
         {PAUL TARAU\\
         {Department of Computer Science and Engineering}\\
         \email{paul.tarau@unt.edu}}

\submitted{?}
\revised{?}
\accepted{?}

\begin{document}
\pagerange{\pageref{firstpage}--\pageref{lastpage}}
\label{firstpage}
\maketitle

\begin{abstract}

With sound unification, Definite Clause Grammars
and compact expression of
combinatorial generation algorithms,
logic programming is shown to conveniently 
host a declarative playground 
where interesting properties and behaviors emerge 
from the interaction of heterogenous but
deeply connected computational objects.

Compact combinatorial generation algorithms are given
for several families of lambda terms, including
open, closed, simply typed and linear terms as well as
type inference and normal order reduction algorithms.
We describe a Prolog-based combined lambda term generator 
and type-inferrer for closed 
well-typed terms of a given size, in de Bruijn notation.

We introduce a compressed 
de Bruijn representation of lambda terms and define
its bijections to  standard representations.
Our compressed terms facilitate derivation of
size-proportionate ranking and unranking algorithms 
of lambda terms  and their inferred simple types.

By taking advantage of Prolog's unique bidirectional 
execution model and sound unification algorithm, 
our generator can build ``customized'' closed terms of a given type. 
This relational view of terms and their types
enables the discovery of interesting
patterns about frequently used type expressions 
occurring in well-typed functional programs. 
Our study uncovers the most ``popular'' types that 
govern function applications among a about a million
small-sized lambda terms and hints toward practical uses 
to combinatorial software testing.

The S and K combinator expressions form a well-known
Turing-complete subset of the lambda calculus.
We specify evaluation, 
type inference
and combinatorial generation
algorithms for
SK-combinator trees.
In the process, we unravel 
properties shedding
new light on interesting
aspects of their structure
and distribution.
We study the proportion of well-typed
terms among size-limited SK-expressions
as well as the type-directed
generation of terms of sizes smaller
then the size of their simple types.
We also introduce the {\em well-typed frontier
of an untypable term} and we use it
to design a simplification algorithm
for untypable terms taking advantage of the
fact that well-typed terms are normalizable.

A uniform representation, as binary 
trees with empty leaves, is given to
expressions built with Rosser's X-combinator,
natural numbers, lambda terms and simple types. 
Using this shared representation, ranking/unranking
algorithm of lambda terms to 
tree-based natural numbers are described.

Our algorithms, expressed as an incrementally developed
literate Prolog program, implement a declarative playground
for exploration of representations, encodings and
computations with uniformly represented
lambda terms, types, combinators and
tree-based arithmetic.

\end{abstract}


\section{Introduction}

This paper is an  extended synthesis of \cite{padl15,cicm15,ppdp15tarau}.
It is organized as literate Prolog program and provides a comprehensive playground
for both classic algorithms working on lambda terms and combinators as well
as a large number of new algorithms and data representations covering 
their combinatorial generation, type inference
and their ranking/unranking to both
standard and tree-represented natural numbers.

Logic programming
provides a convenient 
metalanguage
for modeling 
data types and computations taken
from other programming paradigms.
Properties of logic variables, unification with occurs-check 
and exploration of solution spaces via backtracking
facilitate compact algorithms for 
inferring types or generating
terms for various calculi. This holds
in particular for lambda terms and combinators \cite{bar84}.

While possibly one of the most
heavily researched computational objects,
lambda terms and combinators offer
an endless stream of surprises
to anyone digging just deep enough
below their intriguingly simple
surface.
Lambda terms  provide
a foundation to modern functional languages,
type theory and proof assistants and, as a sign of their lasting relevance,
they have been lately incorporated into
mainstream programming languages including
Java 8, C\# and Apple's newly designed programming language
Swift.

Generation of lambda terms has practical
applications to testing compilers that rely on
lambda calculus as an intermediate language,
as well as to the generation of random tests
for user-level programs and data types. 
At the same time, several instances of 
lambda calculus are of significant theoretical
interest given their correspondence with
logic and proofs.

The use of modern Prologs' unification with occurs-check, 
backtracking mechanisms and
Definite Clause Grammars (DCGs) are instrumental in designing
compact algorithms for inferring simple types
or for generating linear, linear affine or lambda terms
with bounded unary height as well as  in implementing 
normalization algorithms.

Of particular interest are representations that are canonical 
up to alpha-conversion (variable renamings), among which
the most well-known ones are de Bruijn's indices \cite{dbruijn72}, 
representing bound variables as the number
of binders to traverse to the lambda abstraction binding them.
As a sequence of binders, in de Bruijn notation, 
can be seen as a natural number expressed in unary notation, 
we introduce a compressed representation of the binders that
puts in a new light the underlying combinatorial 
structure of lambda terms and highlights their connection to the 
{\em Catalan family} of combinatorial objects \cite{StanleyEC},
among which binary trees are the most well known.
The proposed compressed de Bruijn notation
also simplifies generation of some families of lambda terms.

A joke about the de Bruijn indices representation of lambda terms  is that it can be used to tell apart Cylons from humans \cite{numhack}. Arguably, the compressed de Bruijn representation that we introduce in here is taking their fictional use  one step further. To alleviate the legitimate fears of our (most likely, for now, human) reader, these representations will be  mapped bijectively to  conventional ones.
To be able to use the most natural representation for each
of the proposed algorithms,
we implement bijective transformations between
lambda terms in standard as well as de Bruijn and
compressed de Bruijn representation.

A merit of our compressed representation is to simplify the underlying combinatorial structure of lambda terms, by exploiting their connection to the 
Catalan family of combinatorial objects \cite{StanleyEC}.
This leads to algorithms that focus on
their (bijective) natural number encodings - known 
to combinatorialists as {\em ranking/unranking} functions \cite{combi99}
and to logicians as {\em G\"odel-numberings} \cite{Goedel:31}.
Among the most obvious practical applications, such encodings can be used to generate random terms for testing tools like compilers or source-to-source program transformers.
At the same time, as our encodings are ``size-proportionate'', they provide a compact serialization mechanism for lambda terms.

To derive a bijection to $\N$ (seen as made of conventional bitstring-represented numbers),
that is size-proportionate, we will
first extract a ``Catalan skeleton'' abstracting away the recursive structure of the compressed
de Bruijn term,
then implement a bijection from it to $\N$. The ``content'' fleshing out the term, represented
as a list of natural numbers, will have its own bijection to $\N$ by using a generalized Cantor
tupling / untupling function, that will also help pairing / unpairing the code of the skeleton
and the code of the content of the term.

{\em Combinators} are closed lambda terms without bound variables.
They actually
predate lambda calculus , being discovered in the 1920s by
Sch\"onfinkel and then independently by Curry. 
With {\em function application} as their unique 
operation, and a convenient {\em base},
(for instance, $K = ~\lambda x. ~\lambda y. ~ x$ and 
$S = ~\lambda f. ~\lambda g. ~\lambda x.(f~x)~(g~x)$),
 they form a Turing-complete
subset of the lambda calculus. 
We will focus, using some essential
Prolog ingredients, on synergies
between generation and type inference
on the language of S and K combinators.
SK-combinator expressions are
binary trees with leaves labeled $S$ or $K$ and internal nodes representing
function application.
While working with a drastically simplified model of computation,
interesting patterns emerge, some of which may extend
to the richer combinator languages used in compilers for
functional languages like Haskell and ML 
and proof assistants like Coq and Agda.

Of particular interest are type inference algorithms
and the possibility to melt them together with 
generators of terms of a limited size.
Evaluation of SK-combinator expressions
has the nice property
of always terminating on terms that have simple
types. This suggests looking into how this
property (called {\em strong normalization}) can
be used to  simplify
untypable (and possibly not normalizable)
terms through normalization of their
maximal typable subterms.
At the same time, this suggests trying to
discover some empirical hints about
the distribution of well-typed
terms and the structure induced by typable subterms
of untypable terms.

We also follow some of the
consequences of a very simple idea: what
can happen if combinators, their types,
their computationally equivalent lambda terms
would all share the same basic representation
as the natural numbers that have been used
as encodings of formulas and proofs in such
important fundamental results as G\"odel's
incompleteness theorems, as well as
for mundane purposes like doing 
arithmetic operations in a programming language.

We have shown in the past 
\cite{ppdp14tarau,lata14,sacs14tarau,serpro} that
arithmetic operations and encodings of
various data structures can be performed with 
{\em tree-based numbering
systems} in average time and space complexity
that is comparable with the 
traditional  binary numbers.
One of the properties that singles out
such numbering systems  is their ability
to favor objects with a regular structure
on which representation size and complexity of
operations can be significantly better than
with the usual bitstring representations.
At the same time, we will take advantage of the fact that
Rosser's X-combinator expressions
\cite{fokker92} (a 1-point basis for combinatory logic)
can also be hosted, together with  function application nodes
on top of our ubiquitous binary tree representation.
While the representation of X-combinator
expressions and their types collapses with that
 of binary trees representing natural numbers,
with a few additional steps, we can also
derive size proportionate
ranking and unranking algorithms for 
general lambda terms, this time having
tree-based natural numbers as targets.

Ranking and unranking of lambda terms (i.e., their bijective
 mapping to unique natural number codes) has practical
applications to testing compilers that rely on
lambda calculus as an intermediate language,
as well as in generation of random tests
for user-level programs and data types. 
At the same time, several instances of 
lambda calculus are of significant theoretical
interest given their correspondence with
logic and proofs.
This results in a shared representation
of combinators, simple types, natural numbers
and general lambda terms defining a
common declarative playground for experiments
connecting their computational properties.
Prolog's ability to support ``relational'' queries
enables us to easily explore the population of
de Bruijn terms up to a given size and answer questions
like the following:
\BE
\I How many distinct types occur for terms up to a given size?
\I What are the most popular types?
\I What are the terms that share a given type?
\I What is the smallest term that has a given type?
\I What smaller terms have the same type as this term?
\EE

{\em The remaining of the paper  (with the content of the most salient subsections also pointed out) is organized as follows.}

Section \ref{morphings}  introduces the compressed
de Bruijn terms (\ref{comp}) and bijective transformations
from them to standard lambda terms.

Section \ref{typinf} describes a type inference algorithms for lambda terms.

Section \ref{gener} describes
generators for several classes
of lambda terms, including closed,
simply typed, linear, affine as well as terms
with bounded unary height and terms
in the binary lambda calculus encoding.
Subsection \ref{dbgen} 
introduces a generator for lambda terms in de Bruijn form. 
Subsection \ref{typedgen} introduces an algorithm
combining term generation and type inference.

Section \ref{eval} describes a normal order
reduction algorithm for lambda terms
relaying on their de Bruijn representation.

Section \ref{combs} describes combinators with emphasis on the SK and X combinator bases.
SK-combinator trees together with a generator
and an evaluation algorithm.
Subsection \ref{sktypes} defines simple types
for SK-combinator expressions and
describes a type inference algorithm on for SK-combinator trees. 
Subsection \ref{xctree} introduces
X-combinator trees together with a generator
and an evaluation algorithm.
Subsection \ref{xtypes} defines simple types
for X-combinator expressions via their equivalent
lambda terms,
describes type inference algorithms for X-combinator expressions. It also
explores consequences of expressions and
types sharing the same binary tree representation.

Section \ref{sprop} describes
size-proportionate bijective encodings of lambda terms and combinators.
Subsection \ref{cantor} builds mappings
from lambda terms to Catalan families
of combinatorial objects, with
focus on binary trees representing
their inferred types and their
applicative skeletons.
These mappings lead
in subsection \ref{ranks} to 
size-proportionate ranking
and unranking algorithms for
lambda terms  and their
inferred types.
Subsection \ref{ntree} interprets
X-combinator trees as natural numbers
on which it defines arithmetic operations.
Subsection \ref{goedel} describes a bijection
from lambda terms to
binary trees implementing tree-based arithmetic operations
that leads to a different mechanism for
size-proportionate ranking
and unranking algorithms for
lambda terms.

Section \ref{play} shows examples of applications of our declarative playground.
Subsection \ref{pats} uses our combined term generation and type inference 
 algorithm to discover frequently occurring type patterns.
Subsection \ref{typedir} describes a type-directed algorithm for the
generation of closed typable lambda terms.
We also
explore consequences that emerge
from interactions between such heterogeneous
computational objects sharing the same 
binary tree representation.
Subsection \ref{wtf} introduces the
well-typed frontier of an untypable SK-expression
and describes a partial normalization-based
simplification algorithm
that terminates on all SK-expressions.

Section \ref{rels} discusses related work and section \ref{concl}
concludes the paper.

The paper is structured as a literate Prolog program.
The  code has been tested with SWI-Prolog 6.6.6 and YAP 6.3.4.

It is available 
\url{http://www.cse.unt.edu/~tarau/research/2015/play.pro}.

\begin{codeh}
:-use_module(library(lists)).
\end{codeh}

\section{Morphing between representations of lambda terms} \label{morphings}

Logic variables can be used in Prolog for
 connecting a lambda binder and its related
 variable occurrences. This representation can be made canonical
 by ensuring that each lambda binder is marked with a distinct logic
 variable. 
 For instance, the lambda term 
 $\lambda a.((\lambda b.(a (b~b))) (\lambda c.(a (c~c))))$
 will be represented as {\tt l(A,a(l(B, a(A,a(B,B))), l(C, a(A,a(C,C)))))}.
It is however convenient in most algorithms to avoid any confusion between
variables in our meta-language (Prolog) and lambda variables.
We will achieve this by using the 
{\em de Bruijn representation of lambda terms}.

{\subsection{De Bruijn Indices}
De Bruijn indices \cite{dbruijn72} 
provide a {\em name-free} representation of lambda
terms. All closed terms that can be 
transformed by a renaming of variables ($\alpha$-conversion)
will share a unique representation.
Variables following lambda abstractions are omitted and their occurrences are marked
with positive integers {\em counting the number of lambdas until the one binding them} 
is found on the  way up to the root of the  term. We represent them using
the constructor {\tt a/2} for application, {\tt l/1} for lambda abstractions
(that we will call shortly {\em binders})
and {\tt v/1} for marking the integers corresponding to the 
de Bruijn indices.

For instance, the term {\tt l(A,a(l(B,a(A,a(B,B))),l(C,a(A,a(C,C)))))} is represented
as {\tt l(a(l(a(v(1),a(v(0),v(0)))),l(a(v(1),a(v(0),v(0))))))}, given
that {\tt v(1)} is bound by the outermost lambda (two steps away, counting from {\tt 0}) and the occurrences of {\tt v(0)} are bound each by the closest lambda, represented
by the constructor {\tt l/1}.

We will also define the size of a lambda expression in
de Bruijn form as the number of its internal nodes, 
implemented by the predicate {\tt dbTermSize}.
\begin{code}
dbTermSize(v(_),0).
dbTermSize(l(A),R):-
  dbTermSize(A,RA),
  R is RA+1.
dbTermSize(a(A,B),R):-
  dbTermSize(A,RA),
  dbTermSize(B,RB),
  R is 1+RA+RB.
\end{code}

\subsection{Open and closed terms}
A lambda term is called {\em closed} if it
contains no free variables.
The predicate {\tt isClosedB}  defines 
this property for de Bruijn terms.
\begin{code}
isClosedB(T):-isClosed1B(T,0).

isClosed1B(v(N),D):-N<D.
isClosed1B(l(A),D):-D1 is D+1,
  isClosed1B(A,D1).
isClosed1B(a(X,Y),D):-
  isClosed1B(X,D),
  isClosed1B(Y,D).
\end{code}

Besides being closed, lambda terms interesting for functional
languages and proof assistants, are also well-typed.
We will start with an algorithm inferring types directly on
the de Bruijn terms.

\subsection{From   de Bruijn to lambda terms with canonical names}
The predicate {\tt b2l} converts from the de Bruijn representation 
to lambda terms whose canonical names are provided by logic variables.
We will call them terms in {\em standard notation}.
\begin{code}
b2l(A,T):-b2l(A,T,_Vs).

b2l(v(I),V,Vs):-nth0(I,Vs,V).
b2l(a(A,B),a(X,Y),Vs):-b2l(A,X,Vs),b2l(B,Y,Vs).
b2l(l(A),l(V,Y),Vs):-b2l(A,Y,[V|Vs]).
\end{code}
Note the use of the built-in {\tt nth0/3} that associates
to an index {\tt I} a variable {\tt V} on the list {\tt Vs}.
As we initialize in {\tt b2l/2} the list of logic variables
as a free variable {\tt \_Vs}, free variables in
open terms, represented with indices larger than the number of
available binders will also be consistently 
mapped to logic variables. By replacing {\tt \_Vs} with {\tt []}
in the definition pf {\tt b2l/2} one could enforce that only
closed terms (having no free variables) are accepted.
\BX
illustrates the bijection defined by predicates {\tt l2b} and {\tt b2l}.
\begin{codex}
?- LT=l(A,l(B,l(C,a(a(A,C),a(B,C))))),l2b(LT,BT),b2l(BT,LT1),LT=LT1. 
LT = LT1, LT1 = l(A, l(B, l(C, a(a(A, C), a(B, C))))),
BT = l(l(l(a(a(v(2), v(0)), a(v(1), v(0)))))).
\end{codex}
\EX

\subsection{From lambda terms with canonical names to de Bruijn terms}
Logic variables provide canonical names for lambda variables.
An easy way to manipulate them at meta-language level is to
turn them into special ``\$VAR/1'' terms - a mechanism provided by Prolog's
built-in {\tt numbervars/3} predicate. Given that ``\$VAR/1'' is distinct from
the constructors lambda terms are built from (
{\tt l/2 and a/2}) this is a safe (and invertible) transformation. 
To avoid any side effect on the original term, we will uniformly rename its variables to fresh ones with Prolog's
{\tt copy\_term/2} built-in. We will adopt this technique through the paper
each time our operations would mutate an input argument otherwise.

\begin{code}
l2b(A,T):-copy_term(A,CA),numbervars(CA,0,_),l2b(CA,T,_Vs).

l2b('$VAR'(V),v(I),Vs):-once(nth0(I,Vs,'$VAR'(V))).
l2b(a(X,Y),a(A,B),Vs):-l2b(X,A,Vs),l2b(Y,B,Vs).
l2b(l(V,Y),l(A),Vs):-l2b(Y,A,[V|Vs]).
\end{code}

\subsection{A compressed de Bruijn representation of lambda terms}\label{comp}.

As a step further, we will not look into compressing blocks of lambdas.
Iterated lambdas 
(represented as a block of constructors {\tt l/1} in the de Bruijn notation)
can be seen as a successor arithmetic representation of
a number that counts them. So it makes sense to
represent that number more efficiently in the usual binary
notation. Note that in de Bruijn notation
blocks of  lambdas can wrap either applications
or variable occurrences represented as indices.
This suggests using just two constructors: {\tt v/2}
indicating in a term {\tt v(K,N)} that we have {\tt K}
lambdas wrapped around the de Bruijn index {\tt v(N)} or
{\tt a/3} indicating in a term {\tt a(K,X,Y)} that
{\tt K} lambdas are wrapped around  the application {\tt a(X,Y)}.

We call the terms built this way with the constructors
{\tt v/2} and {\tt a/3} {\em compressed de Bruijn
terms}.

\subsection{From de Bruijn to compressed}

We can make precise the definition of 
compressed deBruijn terms by providing
a bijective transformation between them
and the usual de Bruijn terms.

The predicate {\tt b2c} converts from the usual de Bruijn representation to the compressed one.
It proceeds by case analysis on {\tt v/1, a/2, l/1} and counts the binders {\tt l/1} as
it descends toward the leaves of the tree. Its steps are controlled by 
the predicate {\tt up/2} that increments the counts when crossing 
a binder.
\begin{code}
b2c(v(X),v(0,X)). 
b2c(a(X,Y),a(0,A,B)):-b2c(X,A),b2c(Y,B). 
b2c(l(X),R):-b2c1(0,X,R).

b2c1(K,a(X,Y),a(K1,A,B)):-up(K,K1),b2c(X,A),b2c(Y,B).
b2c1(K, v(X),v(K1,X)):-up(K,K1).
b2c1(K,l(X),R):-up(K,K1),b2c1(K1,X,R).  

up(From,To):-From>=0,To is From+1.
\end{code}

\subsubsection{From compressed to de Bruijn}
The predicate {\tt c2b} converts from the compressed
to the usual de Bruijn representation.
It reverses the effect of {\tt b2c} by expanding
the {\tt K} in {\tt v(K,N)} and {\tt a(K,X,Y)}
into {\tt K+1} {\tt l/1} binders (as counts start at {\tt 0}).
The predicate {\tt iterLam} performs this operation in both cases,
and the predicate {\tt down/2} computes the decrements
at each step. We will reuse the predicates {\tt up/2} and {\tt down/2}
that can be seen as abstracting away the successor/predecessor operation.
\begin{code}
c2b(v(K,X),R):-X>=0,iterLam(K,v(X),R).
c2b(a(K,X,Y),R):-c2b(X,A),c2b(Y,B),iterLam(K,a(A,B),R).

iterLam(0,X,X).
iterLam(K,X,l(R)):-down(K,K1),iterLam(K1,X,R).

down(From,To):-From>0,To is From-1. 
\end{code}

\BX
illustrates the bijection defined by the predicates {\tt b2c} and {\tt c2b}.
\begin{codex}
?- BT=l(l(l(a(a(v(2), v(0)), a(v(1), v(0)))))),b2c(BT,CT),c2b(CT,BT1).
BT = BT1, BT1 = l(l(l(a(a(v(2), v(0)), a(v(1), v(0)))))),
CT = a(3, a(0, v(0, 2), v(0, 0)), a(0, v(0, 1), v(0, 0))) .
\end{codex}
\EX

A convenient way to simplify defining chains of such conversions is
by using Prolog's DCG transformation. For instance, the
predicate {\tt c2l/2} converts from
compressed de Bruijn terms and standard lambda terms using
de Bruijn terms as an intermediate step,
while {\tt l2c/2}  works the other way around.
\begin{code}
c2l --> c2b,b2l.
l2c --> l2b,b2c.
\end{code}

Closed terms can be easily identified by ensuring
that the lambda binders on a given path from root
outnumber the de Bruijn index of a variable occurrence 
ending the path.
The predicate {\tt isClosedC} does that
for compressed de Bruijn terms.
\begin{code}  
isClosedC(T):-isClosedC(T,0).
  
isClosedC(v(K,N),S):-N<S+K.
isClosedC(a(K,X,Y),S1):-S2 is S1+K,isClosedC(X,S2),isClosedC(Y,S2).
\end{code}

\section{Inferring simple types for lambda terms} \label{typinf}

\subsection{Type Inference on standard terms with logic variables}
{\em Simple types}, represented as binary trees 
built with the constructor ``\verb~>/2~'' with empty leaves
representing the unique primitive type ``{\tt x}'', can be seen
as a ``Catalan approximation'' of lambda terms, centered
around ensuring their safe and terminating
evaluation (strong normalization).

While in a functional language inferring types
requires implementing unification with occur check, 
as shown for instance in \cite{grygielGen}, 
this operation is available in Prolog as a built-in.
Also a ``post-mortem'' verification that unification
has not introduced any cycles is provided
by the built-in {\tt acyclic\_term/1}.

The predicate {\tt extractType/2} works by seeing
each logical variable {\tt X} as denoting its type.
As logic variable bindings propagate between binders
and occurrences, this ensures that types are 
consistently inferred.
\begin{code}  
extractType(X,TX):-var(X),!,TX=X. % this matches all variables
extractType(l(TX,A),(TX>TA)):-extractType(A,TA).
extractType(a(A,B),TY):-extractType(A,(TX>TY)),extractType(B,TX).
\end{code}
Instead of (inefficiently) using unification with occurs-check
at each step, we ensure that at the end, our
term is still {\em acyclic}, with the built-in ISO-Prolog predicate
{\tt acyclic\_term/1}.
\begin{code}
polyTypeOf(LTerm,Type):-
  extractType(LTerm,Type),
  acyclic_term(LTerm).
\end{code}
At this point, most general types are inferred by {\tt extractType}
 as fresh variables,  similar to polymorphic types 
 in functional languages, if one interprets
 logic variables as universally quantified.
 Such variables stand for any type expression, as
 schemata in the case of propositional or predicate logic axioms. 
\BX
Type inference for canonically represented standard terms. Note the need
to use {\tt copy\_term} to avoid binding the object-level variables.
\begin{codex}
?- polyTypeOf(l(X,a(X,l(Y,Y))),T).
X = ((A>A)>B),
Y = A,
T = (((A>A)>B)>B).

?- copy_term(l(X,a(X,l(Y,Y))),LT),polyTypeOf(LT,T).              
LT = l((A>A)>B, a((A>A)>B, l(A, A))),
T = (((A>A)>B)>B).
\end{codex}
\EX
However, as we are only interested in simple types, we will
bind uniformly the leaves of our type tree to the constant
``{\tt x}'' representing our only primitive type, by
using the predicate {\tt bindTypeB/1}.
\begin{code}
bindTypeB(x):-!.
bindTypeB((A>B)):-bindTypeB(A),bindTypeB(B).  
\end{code}
The simple type of a compressed de Bruijn term is then defined as:
\begin{code}
hasType(CTerm,Type):-
  c2l(CTerm,LTerm),
  polyTypeOf(LTerm,Type),
  bindTypeB(Type).
\end{code}  
We can also define the predicate {\tt typable/1}
that checks if a lambda term is well typed,
by trying to infer and then ignoring its inferred 
type.
\begin{code}
typable(Term):-hasType(Term,_Type).
\end{code}

\BX
illustrates typability of the term corresponding to the
{\tt S} combinator $ ~\lambda x_0. ~\lambda x_1. ~\lambda x_2.((x_0~x_2)~(x_1~x_2))   $ and untypabilty of the term corresponding to the {\tt Y} combinator $ ~\lambda x_0.( ~\lambda x_1.(x_0~(x_1~x_1)) ~ ~\lambda x_2.(x_0~(x_2~x_2)) ) $, in de Bruijn form.
\begin{codex}
?- hasType(a(3,a(0,v(0,2),v(0,0)),a(0,v(0,1),v(0,0))),T).
T = ((x> (x>x))> ((x>x)> (x>x))) .
?- hasType(
   a(1,a(1,v(0,1),a(0,v(0,0),v(0,0))),a(1,v(0,1),a(0,v(0,0),v(0,0)))),T).
false.
\end{codex}
\EX

\subsection{Type inference for lambda terms in de Bruijn notation}\label{dbinf}

As lambda terms represent functions, inferring their types
provides information on what kind of argument(s) they can be applied to.
For simple types, type inference is decidable \cite{hindley2008lambda} 
and it uses unification
to recursively propagate type information between application sites
of variable occurrences  covered by a given lambda binder. 
We will describe next a type inference algorithm using de Bruijn indices
in Prolog - a somewhat unusual choice, given that logic variables can
play the role of lambda binders directly. One of the reasons we chose them
is that they will be simpler to manipulate at meta-language level,
as they handle object-level variables implicitly. At the same time
this might be useful for other purposes, as
we are not aware of any Prolog implementation of type inference
with this representation of lambda terms.

{\em Simple types} will be defined here also as binary trees 
built with the constructor ``\verb~>/2~'' with empty leaves,
representing the unique primitive type ``{\tt x}''.
Types can be seen as
as a ``binary tree approximation'' of lambda terms, centered
around ensuring their safe and terminating
evaluation (strong normalization), as it is well-known
(e.g., \cite{bar93}) that
lambda terms that have simple types are strongly normalizing.
When a term {\tt X} has a type {\tt T} we say that the
type {\tt T} is {\em inhabited} by the term {\tt X}.

While in a functional language inferring types
requires implementing unification with occur check, 
as shown for instance in the appendix of \cite{grygielGen}, 
this is readily available in Prolog.

The predicate {\tt boundTypeOf/3} works by associating the same
logical variable, denoting its type, to each of its occurrences.
As a unique logic variable is associated to each leaf {\tt v/1}
corresponding 
via its de Bruijn index to the same binder, types are consistently inferred. 
This is ensured by the use of the built-in {\tt nth0(I,Vs,V0)}
that unifies {\tt V0} with the {\tt I}-th element
of the type context {\tt Vs}. 
Note that unification with occurs-check
needs to be used to avoid cycles in the inferred
type formulas.
\begin{code}  
deBruijnTypeOf(v(I),V,Vs):-
  nth0(I,Vs,V0),
  unify_with_occurs_check(V,V0).
deBruijnTypeOf(a(A,B),Y,Vs):-
  deBruijnTypeOf(A,(X>Y),Vs),
  deBruijnTypeOf(B,X,Vs).
deBruijnTypeOf(l(A),(X>Y),Vs):-
  deBruijnTypeOf(A,Y,[X|Vs]).
\end{code}

At this point, most general types are inferred by {\tt deBruijnTypeOf}
 as fresh variables, similar to polymorphic types 
 in functional languages, if one interprets
 logic variables as universally quantified.

 
\BX
Type inferred for the
{\tt S} combinator $ ~\lambda x_0. ~\lambda x_1.  ~\lambda x_2.((x_0~x_2)~(x_1~x_4))$ in 
 de Bruijn form.
\begin{codex}
?- X=l(l(l(a(a(v(2), v(0)), a(v(1), v(0)))))),deBruijnTypeOf(X,T,0).
X = l(l(l(a(a(v(2), v(0)), a(v(1), v(0)))))),
T = ((A> (B>C))> ((A>B)> (A>C))).
\end{codex}
\EX

However, as we are only interested 
in simple types of closed terms with only one basic type, we will
bind uniformly the leaves of our type tree to the constant
``{\tt x}'' representing our only primitive type, by
using the predicate {\tt bindTypeB/1}.

\begin{code}
boundTypeOf(A,T):-deBruijnTypeOf(A,T0,[]),bindTypeB(T0),!,T=T0.
\end{code}

\BX
Simple type inferred for the 
{\tt S} combinator and failure to assign a type to the {\tt Y} combinator 
$ ~\lambda x_0.( ~\lambda x_1. (x_0~(x_1~x_2)) ~\lambda x_2.(x_1~(x_2~x_2)) ) $.
\begin{codex}
?- boundTypeOf(l(l(l(a(a(v(2), v(0)), a(v(1), v(0)))))),T).
T = T = ((x> (x>x))> ((x>x)> (x>x))).
?- boundTypeOf(l(a(l(a(v(1), a(v(0), v(0)))), 
                      l(a(v(1), a(v(0), v(0)))))),T). 
false.
\end{codex}
\EX

\section{Generating  families of lambda terms} \label{gener}

We can see our compressed de Bruijn terms 
as binary trees decorated with integer labels.
The binary trees provide a skeleton
that  describes the
applicative structure of the
underlying lambda terms.
At the same time, types in the {\em simple typed
lambda calculus} \cite{bar93} share  a similar
binary tree structure.

\subsection{Generating binary trees}\label{cats}

Binary trees are among the most well-known members of the 
Catalan family of combinatorial objects \cite{StanleyEC}, that
has at least 58 structurally distinct members,
covering several data structures, geometric objects and
formal languages.

We will build  binary trees with the constructor \verb~>/2~ for branches
and the constant {\tt x} for its leaves. This will match the usual notation for 
simple types \cite{bar93} 
of lambda terms that can be represented as binary trees.

\subsubsection{Generating binary trees of given depth}
A generator / recognizer of binary trees of a limited depth, 
counted by  entry {\tt A003095} in \cite{intseq}
is defined by the predicate {\tt genTreeByDepth/2}.
\begin{code}
genTreeByDepth(_,x).
genTreeByDepth(D1,(X>Y)):-down(D1,D2),
  genTreeByDepth(D2,X),
  genTreeByDepth(D2,Y).
\end{code}

\BX
illustrates trees of depth at most 2 generated  by the predicate {\tt genTreeByDepth}.
\begin{codex}
?- genTreeByDepth(2,T).
T = x ;
T = (x>x) ;
T = (x> (x>x)) ;
T = ((x>x)>x) ;
T = ((x>x)> (x>x)).
\end{codex}
\EX

\subsubsection{Generating binary trees of given size} \label{gentree}

A generator / recognizer of binary trees of a fixed size (seen
as the number of internal nodes,
counted by  entry {\tt A000108} in \cite{intseq})
is defined by the predicate {\tt genTree/2}.

\begin{code}  
genTree(N,T):-genTree(T,N,0).
 
genTree(x)-->[].
genTree((X>Y))-->down,
  genTree(X),
  genTree(Y).
\end{code}
Note the creative use of Prolog's DCG-grammar transformation. After the DCG expansion,
the code for {\tt genTree/3} becomes something like:
\begin{codex}
genTree(x,K,K).
genTree((X>Y),K1,K3):-down(K1,K2),
  genTree(X,K2,K3),
  genTree(K3,K4).
\end{codex}
Given that down(K1,K2) unfolds to \verb~K1>0,K2 is K1-1~ it is clear that
this code ensures that the total number of nodes {\tt N} passed
by {\tt genTree/2} to {\tt genTree/3}
controls the size of the generated trees.
We will reuse this pattern through the paper, as it simplifies the writing of
generators for various combinatorial objects.
It is also convenient to standardize on the number of 
{\em internal nodes} as defining the {\em size} of our terms.
\BX
illustrates trees with 3 internal nodes (built with the constructor ``\verb~>/2~'')
generated  by {\tt genTree/2}.
\begin{codex}
?- genTree(3,BT).                                       
BT = (x> (x> (x>x))) ;
BT = (x> ((x>x)>x)) ;
BT = ((x>x)> (x>x)) ;
BT = ((x> (x>x))>x) ;
BT = (((x>x)>x)>x) .
\end{codex}
\EX

The predicate {\tt tsize} defines the size
of a binary tree in terms of the
number of its internal nodes.
\begin{code}  
tsize(x,0).
tsize((X>Y),S):-tsize(X,A),tsize(Y,B),S is 1+A+B.
\end{code}

\subsection{Generating Motzkin trees}
Motzkin-trees (also called binary-unary trees)
have internal nodes of arities 1 or 2. Thus they can be seen
as an abstraction of lambda terms that ignores de Bruijn indices at the leaves.
The predicate {\tt motzkinTree/2} generates Motzkin trees
with L internal and leaf nodes. 
\begin{code}
motzkinTree(L,T):-motzkinTree(T,L,0).

motzkinTree(u)-->down.
motzkinTree(l(A))-->down,motzkinTree(A).
motzkinTree(a(A,B))-->down,motzkinTree(A),motzkinTree(B).
\end{code}
Motzkin-trees are counted by the sequence
{\tt A001006} in \cite{intseq}.
If we replace the first clause  of {\tt motzkinTree/2} with
\verb~motzkinTree(u)-->[]~, we obtain binary-unary trees with {\tt L}
internal nodes, counted
by the entry {\tt A006318} (Large Schr\"oder Numbers)
of \cite{intseq}.

\subsection{Generating closed lambda terms in standard notation}
With logic variables representing binders and their
occurrences, one can also generate lambda 
terms in standard notation directly.
The predicate {\tt genLambda/2} equivalent to {\tt genStandard/2},
builds a list of logic variables as it generates binders. 
When generating a leaf,
it picks nondeterministically one of the binders among
the list of binders available, {\tt Vs}. As usual, the predicate
{\tt down/2} controls the number of internal nodes.

\begin{code}
genLambda(L,T):-genLambda(T,[],L,0).

genLambda(X,Vs)-->{member(X,Vs)}.
genLambda(l(X,A),Vs)-->down,genLambda(A,[X|Vs]).
genLambda(a(A,B),Vs)-->down,genLambda(A,Vs),genLambda(B,Vs).
\end{code}


To generate lambda terms of a given size, we can write
generators similar to the ones for 
binary trees in section \ref{cats}.
Moreover, we have the choice to use generators for
standard, de Bruijn or compressed de Bruijn terms
and then bijectively morph the resulting terms
in the desired representation, as outlined is section \ref{comp}.

\subsection{Deriving a generator for lambda terms in de Bruijn form}\label{dbgen}

We can derive a generator for closed lambda terms in de Bruijn form
by extending a {\em Motzkin} or {\em unary-binary} tree 
generator to keep track of the lambda binders. 
When reaching a leaf {\tt v/1}, one of the 
available binders (expressed as a de
Bruijn index) will be assigned to it
nondeterministically.

The predicate {\tt genDBterm/4}  generates
closed de Bruijn terms with a fixed number of
internal (non-index) nodes, as counted by entry
{\tt A220894} in \cite{intseq}.

\begin{code}
genDBterm(v(X),V)-->
  {down(V,V0)},
  {between(0,V0,X)}.
genDBterm(l(A),V)-->down,
  {up(V,NewV)},
  genDBterm(A,NewV).
genDBterm(a(A,B),V)-->down,
  genDBterm(A,V),
  genDBterm(B,V).  


\end{code}
The range of possible indices is provided by
Prolog's built-in integer range generator
{\tt between/3}, that provides
values from {\tt 0} to {\tt V0}.
Note also the use of {\tt down/2}
abstracting away the predecessor operation and
{\tt up/2} abstracting away the successor operation.
Together, they control the amount of available nodes
and the incrementing of de Bruijn indices at each lambda node.

Our generator of de Bruijn terms is exposed through two interfaces: {\tt genDBterm/2}
that generates closed de Bruijn 
terms with exactly {\tt L} non-index nodes
and {\tt genDBterms/2} that generates terms with up to {\tt L}
non-index nodes, by not enforcing that
exactly {\tt L} internal nodes must be used.
\begin{code}
genDBterm(L,T):-genDBterm(T,0,L,0).

genDBterms(L,T):-genDBterm(T,0,L,_).
\end{code}
Inserting a {\tt down} operation in
the first clause of {\tt genDBterm/4}
will enumerate terms counted by sequence {\tt A135501} instead of {\tt A220894},
 as this would imply assuming size 1 for variables.
in \cite{intseq}.
\BX
Generation of terms with up to {\tt 2} internal nodes.
\begin{codex}
?- genDBterms(2,T).
T = l(v(0)) ;
T = l(l(v(0))) ;
T = l(l(v(1))) ;
T = l(a(v(0), v(0))).
\end{codex}
\EX

\subsection{Deriving generators for closed terms in compressed de Bruijn form}
A generator for compressed de Bruijn terms can be derived by using
 {\tt DCG} syntax to compose a generator for closed
de Bruijn terms {\tt genDBterm} and {\tt genDBterms} 
and a transformer to compressed terms {\tt b2c/2}. 
\begin{code}
genCompressed --> genDBterm,b2c.
genCompresseds--> genDBterms,b2c.
\end{code}

\subsection{Generators for closed terms in standard notation}
\begin{code}
genStandard-->genDBterm,b2l.
genStandards-->genDBterms,b2l.
\end{code}

\BX
illustrates generators for closed terms in compressed de Bruijn and standard notation
with logic variables providing lambda variable names.
\begin{codex}
?- genCompressed(2,T).
T = v(2, 0) ;
T = v(2, 1) ;
T = a(1, v(0, 0), v(0, 0)).

?- genStandard(2,T).
T = l(_G3434, l(_G3440, _G3440)) ;
T = l(_G3434, l(_G3440, _G3434)) ;
T = l(_G3437, a(_G3437, _G3437)).
\end{codex}
\EX

\subsection{Generating normal forms}
Normal forms are lambda terms that cannot be further reduced.
A normal form should not be an application with a lambda as its left branch 
and, recursively, its subterms should also be normal forms.
The predicate {\tt nf/4} defines this inductively and generates
all normal forms with {\tt L} internal nodes in de Bruijn form.
\begin{code}
nf(v(X),V)-->{down(V,V0),between(0,V0,X)}.
nf(l(A),V)-->down,{up(V,NewV)},nf(A,NewV).
nf(a(v(X),B),V)-->down,nf(v(X),V),nf(B,V).  
nf(a(a(X,Y),B),V)-->down,nf(a(X,Y),V),nf(B,V).  
\end{code}
As we standardize our generators to produce compressed de Bruijn terms,
we combine {\tt nf/4} and the converter {\tt b2c/2} to produce normal forms
of size exactly  {\tt L} (predicate {\tt nf/2}) and with size 
up to {\tt L} (predicate {\tt nfs/2}).

\begin{code}
nf(L,T):-nf(B,0,L,0),b2c(B,T).
nfs(L,T):-nf(B,0,L,_),b2c(B,T).
\end{code}

\BX
illustrates normal forms with exactly 2 non-index nodes.
\begin{codex}
?- nf(2,T).
T = v(2, 0) ;
T = v(2, 1) ;
T = a(1, v(0, 0), v(0, 0)) .
\end{codex}
\EX
The number of solutions of our 
generator replicates entry {\tt A224345} 
in \cite{intseq} that counts closed normal forms of various
sizes.

\subsection{Generation of linear lambda terms}
{\em Linear lambda terms} \cite{BCI13} restrict binders to {\em exactly
one} occurrence.

The predicate {\tt linLamb/4}
uses logic variables both as leaves and as lambda binders and generates
terms in standard form.
In the process, binders accumulated on
the way down from the root, must be split between the two
branches of an application node.
The predicate
{\tt subset\_and\_complement\_of/3} achieves this by 
generating all such possible splits
of the set of binders.
\begin{code}
linLamb(X,[X])-->[].
linLamb(l(X,A),Vs)-->down,linLamb(A,[X|Vs]).
linLamb(a(A,B),Vs)-->down,
  {subset_and_complement_of(Vs,As,Bs)},
  linLamb(A,As),linLamb(B,Bs).
\end{code}
At each step of {\tt subset\_and\_complement\_of/3},
{\tt place\_element/5} is called to distribute each element
of a set to exactly one of two disjoint subsets.
\begin{code}
subset_and_complement_of([],[],[]).
subset_and_complement_of([X|Xs],NewYs,NewZs):-
  subset_and_complement_of(Xs,Ys,Zs),
  place_element(X,Ys,Zs,NewYs,NewZs).
\end{code}
\begin{code}
place_element(X,Ys,Zs,[X|Ys],Zs).
place_element(X,Ys,Zs,Ys,[X|Zs]).
\end{code}
As usual, we standardize the generated terms
by converting them with {\tt l2c} to
compressed de Bruijn terms.
\begin{code}
linLamb(L,CT):-linLamb(T,[],L,0),l2c(T,CT).
\end{code}

\BX
illustrates linear lambda terms for {\tt L=3}.
\begin{codex}
?- linLamb(3,T).
T = a(2, v(0, 1), v(0, 0)) ;
T = a(2, v(0, 0), v(0, 1)) ;
T = a(1, v(0, 0), v(1, 0)) ;
T = a(1, v(1, 0), v(0, 0)) ;
T = a(0, v(1, 0), v(1, 0)) .
\end{codex}
\EX

\subsection{Generation of affine linear lambda terms}
Linear affine lambda terms \cite{BCI13} restrict binders to {\em at most
one} occurrence.

\begin{code}
afLinLamb(L,CT):-afLinLamb(T,[],L,0),l2c(T,CT).

afLinLamb(X,[X|_])-->[].
afLinLamb(l(X,A),Vs)-->down,afLinLamb(A,[X|Vs]).
afLinLamb(a(A,B),Vs)-->down,
  {subset_and_complement_of(Vs,As,Bs)},
  afLinLamb(A,As),afLinLamb(B,Bs).
\end{code}

\BX
illustrates generation of affine linear lambda terms 
in compressed de Bruijn form.
\begin{codex}
?- afLinLamb(3,T).          
T = v(3, 0) ;
T = a(2, v(0, 1), v(0, 0)) ;
T = a(2, v(0, 0), v(0, 1)) ;
T = a(1, v(0, 0), v(1, 0)) ;
T = a(1, v(1, 0), v(0, 0)) ;
T = a(0, v(1, 0), v(1, 0)) ;
\end{codex}
\EX
Clearly all linear terms are affine. It is also known
that all affine terms are typable.

\subsubsection{Generating lambda terms of bounded unary height}

Lambda terms of bounded unary height are introduced in \cite{bodini11} where 
it is argued that  such terms are naturally occurring in programs 
and it is shown that their asymptotic behavior is
easier to study.

They are specified by giving a bound on the number of lambda binders
from a de Bruijn index to the root of the term.


\begin{code}
boundedUnary(v(X),V,_D)-->{down(V,V0),between(0,V0,X)}.
boundedUnary(l(A),V,D1)-->down,
  {down(D1,D2),up(V,NewV)},
  boundedUnary(A,NewV,D2).
boundedUnary(a(A,B),V,D)-->down,
  boundedUnary(A,V,D),boundedUnary(B,V,D).  
\end{code}

The predicate {\tt boundedUnary/5} generates lambda terms
of size {\tt L} 
in compressed de Bruijn form with unary hight {\tt D}.
\begin{code}
boundedUnary(D,L,T):-boundedUnary(B,0,D,L,0),b2c(B,T).
boundedUnarys(D,L,T):-boundedUnary(B,0,D,L,_),b2c(B,T).
\end{code}

\BX
illustrates terms of unary height 1 with size up to 3.
\begin{codex}
?- boundedUnarys(1,3,R).
R = v(1, 0) ;
R = a(1, v(0, 0), v(0, 0)) ;
R = a(1, v(0, 0), a(0, v(0, 0), v(0, 0))) ;
R = a(1, a(0, v(0, 0), v(0, 0)), v(0, 0)) ;
R = a(0, v(1, 0), v(1, 0)) .
\end{codex}
\EX




\subsection{Generating terms in binary lambda calculus encoding}

Generating de Bruijn terms based 
on the size of their binary lambda calculus encoding
\cite{binlamb}
works by using a DCG mechanism to build the actual code as
a list {\tt Cs} of {\tt 0} 
and {\tt 1} digits and specifying the size of the code in advance.

\begin{codeh}
blc(L,T):-blc(L,T,_Cs).
\end{codeh}
\begin{code}
blc(L,T,Cs):-length(Cs,L),blc(B,0,Cs,[]),b2c(B,T).

blc(v(X),V)-->{between(1,V,X)},encvar(X).
blc(l(A),V)-->[0,0],{NewV is V+1},blc(A,NewV).
blc(a(A,B),V)-->[0,1],blc(A,V),blc(B,V).  
\end{code}
Note that de Bruijn binders are encoded as {\tt 00},
applications as {\tt 01} and de Bruijn indices
in unary notation are encoded as {\tt 00$\ldots$01}.
This operation is preformed by the predicate {\tt encvar/3},
that, in DCG notation, uses {\tt down/2} at each step to generate the
sequence of {\tt 1} terminated {\tt0} digits.
\begin{code}  
encvar(0)-->[0].
encvar(N)-->{down(N,N1)},[1],encvar(N1).
\end{code}

\BX
illustrates generation of 8-bit binary lambda terms ({\tt Cs})
together with their compressed de Bruijn form ({\tt T}).
\begin{codex}
?- blc(8,T,Cs).
T = v(3, 1),
Cs = [0, 0, 0, 0, 0, 0, 1, 0] ;
T = a(1, v(0, 1), v(0, 1)),
Cs = [0, 0, 0, 1, 1, 0, 1, 0] .
\end{codex}
\EX
Note that while not bijective, the binary encoding
has the advantage of being a self-delimiting code. This
facilitates its use in an unusually compact interpreter.

\subsection{Generating typable terms}
The predicate {\tt genTypable/2} generates closed typable terms of size {\tt L}.
These are counted by entry {\tt A220471} in \cite{intseq}.
\begin{code}
genTypable(L,T):-genCompressed(L,T),typable(T).
genTypables(L,T):-genCompresseds(L,T),typable(T).
\end{code}

\BX
illustrates a generator for
closed typable terms.
\begin{codex}
?- genCompressed(2,T).
T = v(2, 0) ;
T = v(2, 1) ;
T = a(1, v(0, 0), v(0, 0)).
\end{codex}
\EX

\subsection{Combining term generation and type inference}

One could combine a generator for closed terms and a type inferrer 
in  a ``generate-and-test'' style as follows:
\begin{code}
genTypedTerm1(L,Term,Type):-
  genDBterm(L,Term),
  boundTypeOf(Term,Type).
\end{code}

\begin{codeh}  
genTypedTerms1(L,Term,Type):-
  genDBterms(L,Term),
  boundTypeOf(Term,Type).
\end{codeh}
Note that when one wants to select only terms having a given type
this is quite inefficient. Next, we will show how to
combine size-bound term generation,
testing for closed terms and type inference into a single
predicate. This will enable efficient querying
about {\em what terms inhabit a given type}, as one
would expect from Prolog's multi-directional
execution model.

\subsection{Generating closed well-typed terms of a given size} \label{typedgen}

One can derive, from the type inferrer
{\tt boundTypeOf}, a more efficient generator for
de Bruijn terms with a given number
of internal nodes.

The predicate {\tt genTypedTerm/5}
relies on Prolog's DCG notation
to thread together the steps
controlled by the predicate {\tt down}.
Note also the nondeterministic use of
the built-in {\tt nth0} that enumerates
values for both {\tt I} and {\tt V}
ranging over the list 
of available variables {\tt Vs},
as well as the use of {\tt
unify\_with\_occurs\_check}
to ensure that unification
of candidate types 
does not create cycles.

\begin{code}
genTypedTerm(v(I),V,Vs)-->
  {
   nth0(I,Vs,V0),
   unify_with_occurs_check(V,V0)
  }.
genTypedTerm(a(A,B),Y,Vs)-->down,
  genTypedTerm(A,(X>Y),Vs),
  genTypedTerm(B,X,Vs).
genTypedTerm(l(A),(X>Y),Vs)-->down,
  genTypedTerm(A,Y,[X|Vs]).  
\end{code}
Two interfaces are offered:
{\tt genTypedTerm} that generates de Bruijn terms
of with exactly {\tt L} internal nodes and
{\tt genTypedTerms} that generates terms with {\tt L}
internal nodes or less.
\begin{code}
genTypedTerm(L,B,T):-
  genTypedTerm(B,T,[],L,0),
  bindTypeB(T).

genTypedTerms(L,B,T):-
  genTypedTerm(B,T,[],L,_),
  bindTypeB(T).
\end{code} 
As expected, the number of solutions, computed as the sequence 
{1, 2, 9, 40, 238, 1564, 11807, 98529, 904318, 9006364, 96709332, 1110858977~$~\dots$}
 for sizes  $1,2,3, ~\ldots~$,12,$~\ldots~$
 matches entry {\tt 
A220471} in \cite{intseq}.
Note that the last 2 terms are not (yet) in the {\tt 
A220471} in \cite{intseq} as the generate and filter method used in
\cite{grygielGen} is limited by the super-exponential growth of
the closed lambda terms among which the relatively 
few well-typed ones need to be found (e.g. more than 12 billion terms for size 12).
Interestingly, by interleaving generation of closed terms and
type inference in the predicate {\tt genTypedTerm} the
time to generate all the well-typed terms
is actually shorter than the time to generate
all closed terms of the same size, e.g.. 3.2 vs 4.3 seconds for size 9 
with SWI-Prolog. As via the Curry-Howard isomorphism closed simply typed terms
correspond to proofs of tautologies in minimal logic, co-generation of terms and
types corresponds to co-generation of tautologies and their proofs for
proofs of given length.

\BX
Generation of well-typed closed de Bruijn terms of size 3.
\begin{codex}
?- genTypedTerm(3,Term,Type).
Term = a(l(v(0)), l(v(0))),
Type = (x>x) ;
Term = l(a(v(0), l(v(0)))),
Type = (((x>x)>x)>x) ;
Term = l(a(l(v(0)), v(0))),
Type = (x>x) ;
Term = l(a(l(v(1)), v(0))),
Type = (x>x) ;
Term = l(l(a(v(0), v(1)))),
Type = (x> ((x>x)>x)) ;
Term = l(l(a(v(1), v(0)))),
Type = ((x>x)> (x>x)) ;
Term = l(l(l(v(0)))),
Type = (x> (x> (x>x))) ;
Term = l(l(l(v(1)))),
Type = (x> (x> (x>x))) ;
Term = l(l(l(v(2)))),
Type = (x> (x> (x>x))) .
\end{codex}
\EX

\section{Normalization of lambda terms}  \label{eval}

Evaluation of lambda terms involves {\em $\beta$-reduction}, 
 a transformation 
of a term like {\tt a(l(X, A),B)} by replacing every 
occurrence of {\tt X} in {\tt A} by {\tt B},
under the assumption that {\tt X} does not occur in {\tt B} and 
{\em $\eta$-conversion}, the transformation of an application term
{\tt a(l(X,A),X)} into {\tt A}, under the assumption that
X does not occur in {\tt A}.

The first tool we need to implement normalization of lambda
terms is a safe substitution operation.
In lambda-calculus based functional languages this can be achieved
through a HOAS (Higher-Order Abstract Syntax) mechanism, that borrows
the substitution operation from the underlying ``meta-language''.
To this end, lambdas are implemented as functions which
get executed (usually  lazily) when
substitutions occur.
We refer to \cite{hoas} for the original description of this mechanism,
widely used these days for implementing embedded 
domain specific languages and proof assistants in 
languages like Haskell or ML.

While logic variables offer a fast and easy way to perform
{\em substitutions}, they do not offer any elegant mechanism
to ensure that substitutions are {\em capture-free}.
Moreover, no HOAS-like mechanism exists in Prolog for
borrowing anything close to {\em normal order reduction} from the 
underlying system, as Prolog would provide, through 
meta-programming, only a {\em call-by-value} model.

We will devise here a simple and safe interpreter
for lambda terms supporting normal order $\beta$-reduction 
by using de Bruijn terms, which also ensures that terms
are unique up to $\alpha$-equivalence. 
As usual, we will omit $\eta$-conversion, known to interfere 
with things like type inference, as the redundant argument(s) 
that it removes might carry useful type information.

The predicate {\tt beta/3} implements the $\beta$-conversion operation
corresponding to the binder {\tt l(A)}. It calls {\tt subst/4} that
replaces in {\tt A} occurrences corresponding the the binder {\tt l/1}.
\begin{code}
beta(l(A),B,R):-subst(A,0,B,R).
\end{code}
The predicate {\tt subst/4} counts, starting from {\tt 0}
the lambda binders down to an occurrence {\tt v(N)}.
Replacement occurs at at level {\tt I} when {\tt I=N}.
\begin{code}
subst(a(A1,A2),I,B,a(R1,R2)):-I>=0,
  subst(A1,I,B,R1),
  subst(A2,I,B,R2).   
subst(l(A),I,B,l(R)):-I>=0,I1 is I+1,subst(A,I1,B,R).
subst(v(N),I,_B,v(N1)):-I>=0,N>I,N1 is N-1. 
subst(v(N),I,_B,v(N)):-I>=0,N<I.
subst(v(N),I,B,R):-I>=0,N=:=I,shift_var(I,0,B,R).
\end{code}
When the right occurrence {\tt v(N)} is reached, the term
substituted for it is shifted such that its variables
are marked with the new, incremented distance to their binders.
The predicate {\tt shift\_var/4} implements
this operation.
\begin{code}
shift_var(I,K,a(A,B),a(RA,RB)):-K>=0,I>=0,
  shift_var(I,K,A,RA),
  shift_var(I,K,B,RB).
shift_var(I,K,l(A),l(R)):-K>=0,I>=0,K1 is K+1,shift_var(I,K1,A,R).
shift_var(I,K,v(N),v(M)):-K>=0,I>=0,N>=K,M is N+I.
shift_var(I,K,v(N),v(N)):-K>=0,I>=0,N<K.
\end{code}

Normal order evaluation of a lambda term, if it terminates,
leads to a unique normal form, as a consequence
of the Church-Rosser theorem,
elegantly proven in \cite{dbruijn72} using de Bruijn terms.
Termination holds, for instance, in the case of simply typed lambda terms.
Its implementation is well known; 
we will follow here the algorithm described in \cite{sestoftLam}.
We first compute the {\em weak head normal form} using {\tt wh\_nf/2}.

\begin{code}
wh_nf(v(X),v(X)).
wh_nf(l(E),l(E)).
wh_nf(a(X,Y),Z):-wh_nf(X,X1),wh_nf1(X1,Y,Z).
\end{code}
The predicate {\tt wh\_nf1/3} does the case analysis of application terms {\tt a/2}.
The key step is the $\beta$-reduction 
in its second clause, when it detects an ``eliminator'' lambda
expression as its left argument, in which case it performs the
substitution of its binder, with  its right argument.
\begin{code}
wh_nf1(v(X),Y,a(v(X),Y)).
wh_nf1(l(E),Y,Z):-beta(l(E),Y,NewE),wh_nf(NewE,Z).
wh_nf1(a(X1,X2),Y,a(a(X1,X2),Y)).
\end{code}
The predicate {\tt to\_nf} implements normal order reduction.
It follows the same skeleton as {\tt wh\_nf},
which is called in the third clause to perform
reduction to weak head normal form, starting from
the outermost lambda binder.

\begin{code}
to_nf(v(X),v(X)).
to_nf(l(E),l(NE)):-to_nf(E,NE).
to_nf(a(E1,E2),R):-wh_nf(E1,NE),to_nf1(NE,E2,R).
\end{code}
Case analysis of application terms for possible $\beta$-reduction
is performed by {\tt to\_nf1/3}, where the second clause
calls {\tt beta/3} and recurses on its result.
\begin{code}
to_nf1(v(E1),E2,a(v(E1),NE2)):-to_nf(E2,NE2).
to_nf1(l(E),E2,R):-beta(l(E),E2,NewE),to_nf(NewE,R).
to_nf1(a(A,B),E2,a(NE1,NE2)):-to_nf(a(A,B),NE1),to_nf(E2,NE2).
\end{code}

Therefore, the predicate {\tt evalDeBruijn} 
\begin{code}
evalDeBruijn --> to_nf.
\end{code}
provides a
Turing-complete lambda calculus interpreter
working on de Bruijn terms. It is
guaranteed to compute a normal form, if it exists.
The predicate {\tt evalStandard/2} works on standard lambda terms,
that in converts to de Bruijn terms and then back after evaluation.
The predicate {\tt evalCompressed/2} works
in a similar way on compressed de Bruijn terms.
We express them as a composition of functions (first argument in, second out)
using Prolog's DCG notation.
\begin{code}
evalStandard-->l2b,to_nf,b2l.
evalCompressed-->c2b,to_nf,b2c.
\end{code}
\BX
illustrates evaluation of the lambda term
$SKK=$\\$(( ~\lambda x_0. ~\lambda x_1. ~\lambda x_2.((x_0~x_2)~(x_1~x_2))   ~ ~\lambda x_3. ~\lambda x_4.x_3  )~ ~\lambda x_5. ~\lambda x_6.x_5  )$ in compressed de Brijn form, resulting
in the definition of the identity combinator $I=\lambda x_0.x_0$.
\begin{codex}
?- S=a(3,a(0,v(0,2),v(0,0)),a(0,v(0,1),v(0,0))),K=v(2,1),
     evalCompressed(a(0,a(0,S,K),K),R).
S = a(3, a(0, v(0, 2), v(0, 0)), a(0, v(0, 1), v(0, 0))),
K = v(2, 1),
R = v(1, 0).
\end{codex}
\EX

\section{Combinators} \label{combs}

Combinators are closed lambda terms placed exclusively as labels at the leaves of application trees. 
Thus combinator expressions are lambda terms represented
as binary trees having applications as internal nodes 
and  combinators
as leaves.
We will explore here two families of combinator expressions, one well-known (SK-combinators) and another that has been mostly forgotten for a more than a half century  (Rosser's X-combinator).

\subsection{SK-Combinator Trees}\label{ctree}

 A {\em combinator basis} is
a set of combinators in  terms of which any other 
combinators can be expressed.

The most well known basis for combinator calculus 
consists of $ K = ~\lambda x_0. ~\lambda x_1.x_0 $ 
and $S = ~\lambda x_0. ~\lambda x_1. ~\lambda x_2.((x_0~x_2)~(x_1~x_2)) $.
$SK$-combinator expressions can be seen as binary trees with leaves labeled
with symbols $S$ and $K$, having function 
applications as internal nodes.
Together with the primitive operation of application, $K$ and $S$
can be used as a 2-point basis to  define a Turing-complete language.

\subsubsection{Generating combinator trees}\label{xcoSK}
Prolog is an ideal language to define in a few lines
generators for various classes of combinatorial objects.
The predicate {\tt genSK} generates SK-combinator trees
with a limited number of internal nodes.

\begin{code}
genSK(k)-->[].
genSK(s)-->[].
genSK(X*Y)-->down,genSK(X),genSK(Y).
\end{code}


Note the use of Prolog's definite clause grammar (DCG) notation 
in combination with the predicate {\tt down/2}
that counts downward the number of available internal nodes.

The predicate {\tt genSK/3}
 provides two interfaces: {\tt genSK/2}
that generates trees with exactly $N$
internal nodes and {\tt genSKs/2} that
generates trees with $N$ or less internal nodes.

\begin{code}  
genSK(N,X):-genSK(X,N,0).
\end{code} 
\begin{code} 
genSKs(N,X):-genSK(X,N,_).
\end{code}

\BX
SK-combinator trees with up to 1 internal nodes (and up to 2 leaves).
\begin{codex}
?- genSKs(1,T).
T = k ;
T = s ;
T = k*k ;
T = k*s ;
T = s*k ;
T = s*s .
\end{codex}
\EX

The predicate {\tt csize} defines the size
of an SK-combinator tree in terms of the
number of its internal nodes.

\begin{code}  
csize(k,0).
csize(s,0).
csize((X*Y),S):-csize(X,A),csize(Y,B),S is 1+A+B.
\end{code}

\subsubsection{A Turing-complete evaluator for SK-combinator trees}\label{evalSK}

An  evaluator for SK-combinator trees recurses over
application nodes, evaluates their subtrees and then
applies the left one to the right one.
\begin{code}
evalSK(k,k).
evalSK(s,s). 
evalSK(F*G,R):-evalSK(F,F1),evalSK(G,G1),appSK(F1,G1,R).
\end{code}
In the predicate {\tt app}, handling
the application of the first argument to
the second, we describe
in the first two clauses the actions
corresponding to {\tt K} and {\tt S}.
The final clause returns the unevaluated
application as its third argument.
\begin{code}

appSK((s*X)*Y,Z,R):-!,  % S
  appSK(X,Z,R1),
  appSK(Y,Z,R2),
  appSK(R1,R2,R).
appSK(k*X,_Y,R):-!,R=X. % K  
appSK(F,G,F*G).
\end{code}
\BX
Applications of SKK and SKS, both
implementing the identity combinator $I=\lambda x.x$.
\begin{codex}
?- appSK(s*k*k,s,R).
R = s.

?- appSK(s*k*s,k,R).
R = k.
\end{codex}
\EX

\subsubsection{De Bruijn equivalents of SK-combinator expressions}

De Bruijn indices \cite{dbruijn72} 
provide a {\em name-free} representation of lambda
terms. All terms closed that can be 
transformed by a renaming of variables ($\alpha$-conversion)
will share a unique representation.
Variables following lambda abstractions are omitted and their occurrences are marked
with positive integers {\em counting the number of lambdas until the one binding them} 
is found on the  way up to the root of the  term. We represent them using
the constructor {\tt a/2} for application, {\tt l/1} for lambda abstractions
(that we will call shortly {\em binders})
and {\tt v/1} for marking the integers corresponding to the 
de Bruijn indices.

For instance,
$ ~\lambda x_0.( ~\lambda x_1.(x_0~(x_1~x_1)) ~ ~\lambda x_2.(x_0~(x_2~x_2)) )$ 
becomes {{\tt l(a(l(a(v(1), a(v(0),v(0)))), l(a(v(1), a(v(0), v(0))))))}}, 
corresponding to
the fact that {\tt v(1)} is bound by the 
outermost lambda (two steps away, counting 
from {\tt 0}) and the occurrences of {\tt v(0)} 
are bound each by the closest lambda, represented
by the constructor {\tt l/1}.
The predicates {\tt kB} and {\tt sB} define the $K$ and $S$ combinators
in de bruijn form.
\begin{code}
kB(l(l(v(1)))).

sB(l(l(l(a(a(v(2),v(0)),a(v(1),v(0))))))).
\end{code}

The predicate {\tt sk2b} transforms an SK-combinator tree
in its lambda expression form, in de Bruijn notation,
by replacing leaves with their de Bruijn form of the S and K combinators
and replacing recursively 
the constructor ``{\tt *}''/2 with the application nodes ``{\tt a}''/2.
\begin{code}
sk2b(s,S):-sB(S).
sk2b(k,K):-kB(K).
sk2b((X*Y),a(A,B)):-sk2b(X,A),sk2b(Y,B).
\end{code}
\BX
Expansion of some small SK-combinator trees to de Bruijn forms.
\begin{codex}
?- sk2b(k*k,R).
R = a(l(l(v(1))), l(l(v(1)))).

?- sk2b(k*s,R).
R = a(l(l(v(1))),l(l(l(a(a(v(2),v(0)),a(v(1),v(0))))))).
\end{codex}
\EX
Clearly
their de Bruijn
equivalents are significantly larger
than the corresponding combinator trees,
but it is easy to see that this is only by
a constant factor, i.e. at most 
the size of the S combinator.

A lambda term is called {\em closed} if it
contains no free variables. 
\begin{prop}
The lambda terms equivalent to SK-combinators computed
by {\tt sk2b} are closed.
\end{prop}
\begin{proof}
As the lambda term equivalent of the
SK-combinator term is clearly a closed expression,
the proposition follows from the definition of {\tt sk2b}, as it
builds terms that apply
closed terms to closed terms.
\end{proof}
This well-known property holds, in fact, for all combinator expressions.
It follows that combinator expressions have a stronger, {\em hereditary}
closedness property: every subtree
of a combinator tree also represents a closed expression.

Besides being closed, lambda terms interesting for functional
languages and proof assistants are also {\em well-typed}.
While the K and S combinators are known to be well-typed,
we would like to see how this property extends to
SK-combinator trees. In particular, we would like
to have an idea on the asymptotic density of
well-typed SK-combinator tree expressions.
We will take advantage of Prolog's sound unification
algorithm to define a type inferrer directly on SK-terms.

\subsection{Inferring simple types for SK-combinator trees}\label{sktypes}
A natural way to define types for combinator expressions is
to borrow them from their lambda calculus equivalents.
This makes sense, as they represent the same function i.e., they
are extensionally the same.
However, this is equivalent to just borrowing
the well-known types of the {\tt S} and {\tt K} combinators
and then recurse over application nodes.

We will next describe an algorithm for inferring types directly on
SK-combinator trees.
  
\subsubsection{A type inference algorithm for SK-terms}\label{dbtypes}

{\em Simple types} will be defined here also as binary trees 
built with the constructor ``\verb~>/2~'' with empty leaves,
representing the unique primitive type ``{\tt x}''.
For brevity, we will mean simple types when mentioning
types, from now on.
Types can be seen as
as a ``binary tree approximation'' of lambda terms, centered
around ensuring their safe and terminating
evaluation (called {\em strong normalization}), 
as the following
well known property states \cite{bar93}.

\begin{prop}
Lambda terms (and combinator expressions, in particular)
that have simple types are strongly normalizing.
\end{prop}

When modeling lambda terms in a functional or procedural language,
inferring types requires implementing unification with occurs-check, 
as shown for instance in the appendix of \cite{grygielGen}. On the other hand 
this operation is readily available in today's Prolog systems.

\begin{code}
skTypeOf(k,(A>(_B>A))).  
skTypeOf(s,(((A>(B>C))> ((A>B)>(A>C) )))).
skTypeOf(A*B,Y):-
  skTypeOf(A,T),
  skTypeOf(B,X),
  unify_with_occurs_check(T,(X>Y)).
\end{code}

At this point, most general types are inferred by {\tt skTypeOf}
 as fresh variables,  similar to multi-parameter polymorphic types 
 in functional languages, if one interprets
 logic variables as universally quantified.

\BX
Type inferred for some SK-combinator expressions. Note the failure
to infer a type for $SSI=SS(SKK)$.
\begin{codex}
?- skTypeOf((((k*k)*k)*k)*k,T).
T = (A>(B>A)).
?- skTypeOf((k*s)*k,T).
T = ((A> (B>C))> ((A>B)> (A>C))).
?- skTypeOf((s*s)*((s*k)*k),T).
false.
\end{codex}
\EX

As we are only interested 
in simple types with only one base type, we will
bind uniformly the leaves of our type tree to the constant
``{\tt x}'' representing our only primitive type, by
using the predicate {\tt bindWithBaseType/1}.

\begin{code}
simpleTypeOf(A,T):-
  skTypeOf(A,T),
  bindWithBaseType(T).

% bind all variables with type 'x'
bindWithBaseType(x):-!. 
bindWithBaseType((A>B)):-
  bindWithBaseType(A),
  bindWithBaseType(B).
\end{code}

\BX
Simple type inferred for combinators $KSK$, $B$ and $C$.
\begin{codex}
?- simpleTypeOf(k*s*k,T).
T = ((x> (x>x))> ((x>x)> (x>x))).
?- B=s*(k*s)*k,C=s*(B*B*s)*(k*k),simpleTypeOf(B,TB),simpleTypeOf(C,TC).
B = s* (k*s)*k,
C = s* (s* (k*s)*k* (s* (k*s)*k)*s)* (k*k),
TB = ((x>x)> ((x>x)> (x>x))),
TC = ((x> (x>x))> (x> (x>x))).
\end{codex}
\EX
It is also useful to define the predicate {\tt typableSK} that
succeeds when a type can be inferred.
\begin{code}
typableSK(X):-skTypeOf(X,_).
\end{code}

\subsection{Rosser's X-combinator}\label{xctree}

We will know explore expressions built with a less well-known
combinator, that provides a 1-point basis for combinator calculi.

It is shown in \cite{OnePoint} 
that a countable number of 
(somewhat artificially constructed) 1-point bases
exist for combinator calculi, but
we will focus here on {\em Rosser's $X$-combinator}, one of the
simplest 1-point bases that is naturally connected
through mutual definitions to the combinators $K$ and $S$.

\subsubsection{The X-combinator in terms of S and K and vice-versa}\label{xco}

A derivation of Rosser's X-combinator
is described in \cite{fokker92}.

Defined as  $X = \lambda f.fKSK$, this combinator 
has the nice property of expressing
both $K$ and $S$ in a symmetric way.
\begin{equation}\label{defK}
 K = (X X) X
\end{equation} 
\begin{equation}\label{defS}
 S=X (X X)
\end{equation}
Moreover, as shown in \cite{fokker92}
the following holds.
\begin{equation}\label{kk}
 K K = X X = ~\lambda x_0. ~\lambda x_1. ~\lambda x_2.x_1
\end{equation} 
As a result, X-combinator expressions are within a (small, see Prop. \ref{const15})
constant factor of their equivalent $SK$-expressions.

Denoting ``{\tt x}'' the empty leaf corresponding to the X-combinator
and ``{\tt >}'' the (non-associative, infix)
constructor for the binary tree's internal nodes,
the predicates {\tt sT}, {\tt kT} and {\tt xxT} define the
Prolog expressions for the $S$, $K$ and $K K = X X $ combinators,
respectively.
\begin{code}
sT(x>(x>x)).
kT((x>x)>x).
xxT(x>x).
\end{code}
This symmetry is part of the motivation for choosing the X-combinator
basis, rather than any of the more well-known ones (see \cite{hindley2008lambda}).

\paragraph{Generating the combinator trees}
As X-combinator trees are just plain binary trees, with leaves denoted {\tt x} and
internal nodes denoted \verb~>~,
we can reuse the predicate {\tt genTree} defined in \ref{gentree}
to generate X-combinator trees
with a given number of internal nodes.

\subsubsection{An evaluator for the Turing-complete language of X-combinator trees}\label{evalT}
We can derive an evaluator for X-combinator trees from a well-known
evaluator for SK-combinator trees.
\begin{code}
evalX((F>G),R):-!,evalX(F,F1),evalX(G,G1),appX(F1,G1,R).
evalX(X,X).
\end{code}
In the predicate {\tt appX/3} handling
the application of the first argument to
the second, we describe
in the first two clauses the actions
corresponding to {\tt K} and {\tt S}.
The final clause returns the unevaluated
application as its third argument.
\begin{code}
appX((((x>x)>x)>X),_Y,R):-!,R=X. % K
appX((((x>(x>x))>X)>Y),Z,R):-!,  % S
  appX(X,Z,R1),
  appX(Y,Z,R2),
  appX(R1,R2,R).
%app((((x>x)>_X)>Y),_Z,R):-!,R=Y. 
%app((x>x)>x,(x>x)>x,R):-!,app(x,x,R).
appX(F,G,(F>G)).
\end{code}
Note also the commented out clauses, that
can shortcut some evaluation steps,
using the identity (\ref{kk}).
\BX
Evaluation of SKK and SKX, equivalent
implementations of the identity combinator $I=\lambda x.x$.
\begin{codex}
?- SKK=(((x>(x>x))>((x>x)>x))>((x>x)>x)),evalX(SKK>x,R).
SKK = (((x>(x>x))>((x>x)>x))>((x>x)>x)),
R = x.

?- SKX=(((x> (x>x))> ((x>x)>x))>x),evalX(SKX>x,R).         
SKX = (((x> (x>x))> ((x>x)>x))>x),
R = x.
\end{codex}
\EX

\subsubsection{De Bruijn equivalents of X-combinator expressions}

De Bruijn indices \cite{dbruijn72} 
provide a {\em name-free} representation of lambda
terms. All terms that can be 
transformed by a renaming of variables ($\alpha$-conversion)
will share a unique representation.
Variables following lambda abstractions are omitted and their occurrences are marked
with positive integers {\em counting the number of lambdas until the one binding them} 
is found on the  way up to the root of the  term. We represent them using
the constructor {\tt a/2} for application, {\tt l/1} for lambda abstractions
(that we will call shortly {\em binders})
and {\tt v/1} for marking the integers corresponding to the 
de Bruijn indices.
For instance,
$ ~\lambda x_0.( ~\lambda x_1.(x_0~(x_1~x_1)) ~ ~\lambda x_2.(x_0~(x_2~x_2)) )$ 
becomes {\small {\tt l(a(l(a(v(1), a(v(0),v(0)))), l(a(v(1),a(v(0),v(0))))))}}, corresponding to
the fact that {\tt v(1)} is bound by the outermost lambda (two steps away, counting from {\tt 0}) and the occurrences of {\tt v(0)} are bound each by the closest lambda, represented
by the constructor {\tt l/1}.

We obtain the X-combinator's definition in terms of $S$ and $K$, in de Bruijn form, by
using the equation $X f = f K S K$ derived from its lambda expression $\lambda f.fKSK$.
The predicate {\tt xB} implements it.
\begin{code}
xB(X):-F=v(0),kB(K),sB(S),X=l(a(a(a(F,K),S),K)).
\end{code}

The predicate {\tt t2b} transforms an X-combinator tree
in its lambda expression form, in de Bruijn notation,
by replacing leaves with the de Bruijn form of the X-combinator
and replacing recursively 
the constructor ``{\tt >}''/2 with the application nodes ``{\tt a}''/2.
\begin{code}
t2b(x,X):-xB(X).
t2b((X>Y),a(A,B)):-t2b(X,A),t2b(Y,B).
\end{code}

\BX
Expansion of small X-combinator trees to de Bruijn forms.
\begin{codex}
?- t2b(x,X).  
X = l(a(a(a(v(0), l(l(v(1)))), l(l(l(a(a(v(2), v(0)), 
                    a(v(1), v(0))))))), l(l(v(1))))).
      
?- t2b(x>x,XX).
XX=a(
     l(a(a(a(v(0),l(l(v(1)))),l(l(l(a(a(v(2),v(0)),
                    a(v(1),v(0))))))),l(l(v(1))))),
     l(a(a(a(v(0),l(l(v(1)))),l(l(l(a(a(v(2),v(0)),
                   a(v(1),v(0))))))),l(l(v(1)))))
   ).     
\end{codex}
\EX
Clearly
their de Bruijn
equivalents are significantly larger
than the corresponding combinator trees,
but we will show that this is only by
a constant factor.
We will also see that often normalization can bring down
significantly the size of such expressions, given
that nodes like {\tt x>x} are equivalent to 
smaller lambda expressions like
$~\lambda x_0. ~\lambda x_1. ~\lambda x_2.x_1$.

\begin{prop}\label{const15}
The size of the lambda term equivalent to an X-combinator
tree with N internal nodes is 15N+14.
\end{prop}
\begin{proof}
Note that the an X-combinator tree with N internal
nodes has N+1 leaves. The de Bruijn tree
built by the predicate {\tt t2b} has also N application nodes,
and is obtained by having leaves replaced in the X-combinator
term, with terms  bringing 14 internal nodes each, 
corresponding to {\tt x}. 
Therefore it has a total of $N+14(N+1)=15N+14$ internal nodes.
\end{proof}

Note also that the lambda terms equivalent to X-combinators computed
by {\tt t2b} are closed, given that the lambda term equivalent to the
X-combinator is a closed expression and {\tt t2b}
builds terms that apply only closed terms to closed terms.

\subsection{Comparing the two evaluators} 

While, as shown in subsection \ref{evalT},  X-combinator
trees can be evaluated directly, it makes sense
to investigate if more compact 
equivalent normal forms can be obtained for them
via their mapping to lambda terms.

One can now compare the evaluation
performed on X-combinator trees
to that performed on their
corresponding lambda expressions.
The predicate {\tt evalAsT}
first evaluates and then
converts while the predicate
{\tt evalAsB} first converts to
a de Bruijn terms and then evaluates
it, with opportunities 
for additional reductions.
\begin{code}
evalAsT --> evalX,t2b.
evalAsB --> t2b,evalDeBruijn.
\end{code}
We express these two predicates
as a {\em composition of functions} (first argument in, second out)
using Prolog's DCG notation.
\BX
Additional reductions obtained from a term of size 29 to a term of size 3 
on the de Bruijn terms associated to an X-combinator expression.
\begin{codex}
?- evalAsT(x>x,R),dbTermSize(R,Size),write(Size),nl,fail.
29

?- evalAsB(x>x,R),dbTermSize(R,Size).
R = l(l(l(v(1)))),
Size = 3 .
\end{codex}
\EX
Note however, as predicted by
the Church-Rosser theorem \cite{dbruijn72,bar84},
applying normalization via {\tt evalDeBruijn} to 
the result of {\tt evalAsT} reaches the same final
normal form. This property is  called {\em confluence}.

\BX
Confluence of evaluation as X-combinator tree and as lambda term.
\begin{codex}
?- evalAsT(x>x,R),evalDeBruijn(R,FinalR).
R = a(l(a(a(a(v(0),...,l(l(v(1)))))),
FinalR = l(l(l(v(1)))) .
\end{codex}
\EX

\subsection{Inferring simple types for X-combinator trees}\label{xtypes}
A natural way to define types for combinator expressions is
to borrow them from their lambda calculus equivalents.
This makes sense, as they represent the same function i.e., they
are extensionally the same.

We will start with an algorithm inferring types  on
the de Bruijn equivalents of X-combinator trees.

\subsubsection{Type trees as combinator trees}\label{tlam}

Besides being closed, lambda terms interesting for functional
languages and proof assistants are also well-typed.
While the K and S combinators are known to be well-typed,
we would like to see how this property extends to
X-combinator trees. In particular, we would like
to have an idea on the asymptotic density of
well-typed X-combinator tree expressions, that
we will explore in subsection \ref{xdens}.

We can define the type of a combinator expression as
the type of its  lambda expression translation.
The predicate {\tt xtype} defines a function from
binary trees to binary trees mapping an X-combinator expression
to its type, as inferred on its equivalent 
lambda term in de Bruijn notation.


\begin{code}
xtype(X,T):-t2b(X,B),boundTypeOf(B,T).
\end{code}
Observe that this only makes sense if the
combinator basis is well-typed. Fortunately
this is the case of the X-combinator
$~\lambda x_0.(((x_0~ \\~\lambda x_1. ~\lambda x_2.x_1  )~ ~\lambda x_3. ~\lambda x_4. ~\lambda x_5.((x_3~x_5)~(x_4~x_5))   )~ ~\lambda x_6. ~\lambda x7.x_6  )$.

\BX
The X-combinator is well-typed.
\begin{codex}
?- xtype(x,T).
T = (((x> (x>x))> (((x> (x>x))> ((x>x)> (x>x)))> ((x> (x>x))>x)))>x).
\end{codex}
\EX

\subsubsection{Inferring types of X-combinator trees directly}\label{dirinf}
The predicate {\tt xt}, that can be seen as a ``partially evaluated''
version of {\tt xtype}, infers the type of the combinators directly.
\begin{code}
xt(X,T):-poly_xt(X,T),bindTypeB(T).

xT(T):-t2b(x,B),boundTypeOf(B,T,[]).

poly_xt(x,T):-xT(T).
poly_xt(A>B,Y):-poly_xt(A,T),poly_xt(B,X),
  unify_with_occurs_check(T,(X>Y)).
\end{code}
It proceeds by first borrowing the type of {\tt x}
from its de Bruijn equivalent. Then, after
calling {\tt poly\_xt}
to infer polymorphic types,
it binds them to our simple-type representation
by calling {\tt bindTypeB}. 

\BX
Simple type inferred directly on X-combinator trees. 
\begin{codex}
?- skkT(X),xt(X,DirectT),xtype(X,BorrowedT).
X = (((x> (x>x))> ((x>x)>x))> ((x>x)>x)),
DirectT = BorrowedT, BorrowedT = (x>x).
\end{codex}
\EX

\section{Size-proportionate bijective encodings of lambda terms and combinators} \label{sprop}

We will describe here two encodings. The first one, in subsections
\ref{cantor} and \ref{ranks} does it ``the hard way''
by working with bitstring-represented natural number codes.
The second one, in subsections
\ref{ntree} and \ref{goedel}
 does it ``the easy-way'', by defining an alternate, tree-based
natural number representation to which fairly strait-forward bijections exist
from combinator expressions, lambda terms and types.

\subsection{An encoding based on Cantor's $\N^k$ to $\N$ bijection} \label{cantor}

We can see our compressed de Bruijn terms 
as binary trees decorated with integer labels.
The underlying binary trees provide a skeleton
that  describes the
applicative structure of the
terms.

\subsubsection{The Catalan family of combinatorial objects} 

Binary trees are among the most well-known members of the 
Catalan family of combinatorial objects \cite{StanleyEC}, that
has at least 58 structurally distinct members,
covering several data structures, geometric objects and
formal languages.

\subsubsection{Size-proportionate encodings} 

In the presence of a bijection between two, usually infinite sets of data objects,
it is possible that representation sizes on one side or the other are exponentially larger that on the other. Well-known encodings like Ackermann's bijection for hereditarily finite sets to natural numbers, defined as $f(\{\})=0,f(x)$ =  $\sum_{a \in
x}2^{f(a)}$, fall in this category.

We will say that a bijection is {\em size-proportionate} if the representation sizes
for corresponding terms on its two sides are ``close enough'' up to a constant factor multiplied with at most the logarithm of any of the sizes.
\begin{df}
Given a bijection between sets of terms of two datatypes denoted $M$ and $N$,
$f : M \to N$, let $m(x)$ be the representation size
of a term $x \in M$ and $n(y)$ be the representation size
of $y \in N$. Then $f$ is called {\em size-proportionate} if 
$|m(x) - n(y)| \in O(log(max(m(x),n(y))))$.
\end{df}
 Informally we also assume that the constants involved are small enough such that the printed representation of two data objects connected by the bijections is about the same.

\subsubsection{The language of balanced parentheses}

Binary trees are in a well-known size-proportionate bijection with 
the language of  balanced parentheses \cite{StanleyEC},
from which we will borrow an efficient ranking/unranking bijection. 
The reversible predicate {\tt catpar/2} transforms
between binary trees and lists of balanced parentheses,
with {\tt 0} denoting the open parentheses and {\tt 1} 
denoting the closing one.
\begin{code} 
catpar(T,Ps):-catpar(T,0,1,Ps,[]).
catpar(X,L,R) --> [L],catpars(X,L,R).

catpars(x,_,R) --> [R].
catpars((X>Xs),L,R)-->catpar(X,L,R),catpars(Xs,L,R). 
\end{code}

\BX
illustrates the work of the reversible predicate {\tt catpar/2}.
\begin{codex}
?- catpar(((x>x)>(x>x)),Ps),catpar(T,Ps).   
Ps = [0, 0, 0, 1, 1, 0, 1, 1], T = ((x>x)> (x>x)) .
\end{codex}
\EX
Note the extra opening/closing parentheses, compared to the
usual definition of Dyck words \cite{StanleyEC}, that make the
sequence self-delimiting.

\subsubsection{A bijection from the language of balanced parenthesis lists to $\N$}

This algorithm follows closely the procedural implementation
described in \cite{combi99}. 

The code of the helper predicates called by {\tt rankCatalan} and {\tt unrankCatalan}
is provided in  \url{http://www.cse.unt.edu/~tarau/research/2015/dbr.pro}.
The details of the algorithms for computing {\tt localRank} and {\tt localunRank}
are described at 
\url{http://www.cse.unt.edu/~tarau/research/2015/dbrApp.pdf}.

The predicate {\tt rankCatalan} uses the Catalan numbers computed by
{\tt cat} in {\tt rankLoop} to shift the ranking over the ranks
of smaller sequences, after calling {\tt localRank}.
\begin{code}
rankCatalan(Xs,R):-
  length(Xs,XL),XL>=2, 
  L is XL-2, I is L // 2,
  localRank(I,Xs,N),
  S is 0, PI is I-1,
  rankLoop(PI,S,NewS),
  R is NewS+N.
\end{code}

The predicate {\tt unrankCatalan} uses the Catalan numbers computed by
{\tt cat} in {\tt unrankLoop} to shift over 
smaller sequences, before calling {\tt localUnrank}.
\begin{code} 
unrankCatalan(R,Xs):- 
  S is 0, I is 0,
  unrankLoop(R,S,I,NewS,NewI),
  LR is R-NewS, 
  L is 2*NewI+1,
  length(As,L),
  localUnrank(NewI,LR,As),
  As=[_|Bs], 
  append([0|Bs],[1],Xs).
\end{code}
The following example illustrates the ranking
and unranking algorithms:
\begin{codex}
?- unrankCatalan(2015,Ps),rankCatalan(Ps,Rank).  
Ps = [0,0,1,0,1,0,1,0,0,0,0,1,0,1,1,1,1,1],Rank = 2015
\end{codex}

\subsubsection{Ranking and unranking simple types}
After putting together the bijections between binary trees and balanced parentheses
with the ranking/unranking of the later we obtain the 
size-proportionate ranking/unranking algorithms
for simple types.
\begin{code} 
rankType(T,Code):-
  catpar(T,Ps),
  rankCatalan(Ps,Code).
  
unrankType(Code,Term):-
  unrankCatalan(Code,Ps),
  catpar(Term,Ps).
\end{code}

\BX
illustrates the ranking and unranking of simple types.
\begin{codex}
?- I=100, unrankType(I,T),rankType(T,R).
I = R, R = 100,
T = (((x>x)> ((x> (x>x))>x))>x) .
\end{codex}
\EX

As there are $O({{4^n} \over {n^{3 \over 2}}})$) binary trees of size n  corresponding to $2^n$ natural numbers of bitsize up to $n$ and
our ranking algorithm visits them in lexicographic order,
it follows that:
\begin{prop}\label{treeSP}
The bijection between types and their ranks is size-proportionate.
\end{prop}

\subsubsection{Catalan skeletons of compressed de Bruijn terms}

As compressed de Bruijn terms can be seen as binary trees
with labels on their leaves and internal nodes, 
their ``Catalan skeleton''
is simply the underlying binary tree.
The predicate {\tt cskel/3} extracts this skeleton
as well as the list of the labels, in depth-first order, as
encountered in the process.
\begin{code}
cskel(S,Vs, T):-cskel(T,S,Vs,[]).

cskel(v(K,N),x)-->[K,N].
cskel(a(K,X,Y),(A>B))-->[K],cskel(X,A),cskel(Y,B).
\end{code}

The predicates {\tt toSkel} and {\tt fromSkel} add
conversion between this binary tree and lists
of balanced parenthesis by using the (reversible)
predicate {\tt catpar}.
\begin{code} 
toSkel(T,Skel,Vs):-
  cskel(T,Cat,Vs,[]),
  catpar(Cat,Skel).
  
fromSkel(Skel,Vs, T):-
  catpar(Cat,Skel),
  cskel(T,Cat,Vs,[]).
\end{code}

\BX
illustrates the Catalan skeleton {\tt Skel} and the list of variable labels {\tt Vs}
extracted from a compressed de Bruijn term corresponding to the {\tt S} combinator.
\begin{codex}
?- T = a(3, a(0, v(0, 2), v(0, 0)), a(0, v(0, 1), v(0, 0))),
   toSkel(T,Skel,Vs),fromSkel(Skel,Vs,T1).
T = T1, T1 = a(3, a(0, v(0, 2), v(0, 0)), a(0, v(0, 1), v(0, 0))),
Skel = [0,0,0,1,1,0,1,1],Vs = [3,0,0,2,0,0,0,0,1,0,0] .
\end{codex}
\EX

\subsubsection{The generalized Cantor $k$-tupling bijection}  \label{cantorgen}
As we we have already solved the problem of ranking and unranking lists
of balanced parentheses, the remaining problem is that
of finding a bijection between the lists of labels collected from
the nodes of a compressed de Bruijn term and natural numbers.

We will use the generalized Cantor bijection between $\N^n$ and $\N$
as the first step in defining this bijection.
The formula, given in \cite{ceg99} p.4,
looks as follows:

\begin{equation}\label{kn}
K_n(x_1,\ldots,x_n) = 
  {\displaystyle\sum_{k=1}^n{{k-1+s_k} \choose k}} 
  \textit{ where } {s_k=\displaystyle\sum_{i=1}^k{x_i}} 
\end{equation}
Note that $n \choose k$ represents the number of subsets 
of $k$ elements of a set of $n$ elements, that also
corresponds to the binomial coefficient of $x^k$ 
in the expansion of $(x+y)^n$, and $K_n(x_1,\ldots,x_n)$ denotes
the natural number associated to the tuple $(x_1,\ldots,x_n)$.
It is easy to see that the generalized Cantor $n$-tupling function 
defined by equation (\ref{kn}) is 
a polynomial of degree $n$
in its arguments.

\subsubsection{The bijection between sets and sequences of natural numbers}
We recognize in the equation (\ref{kn}) the {\em prefix sums} $s_k$ incremented
with values of $k$ starting at $0$.
It represents the ``set side'' of the bijection between sequences
of $n$ natural numbers and sets of $n$ natural numbers
described in \cite{ppdp09pISO}.
It is implemented in the online Appendix as the bijection {\tt list2set}
together with its inverse {\tt set2list}.
For example, {\tt list2set} transforms {\tt [2,0,1,5]} to [2, 3, 5, 11]
as {\tt 3=2+0+1,5=3+1+1,11=5+5+1} and {\tt set2list} transforms it back
by computing the differences between consecutive members, reduced by 1.

\subsubsection{The $\N^n \to \N$ bijection} \label{revisit}
  
The bijection $K_n:\N^n \to \N$ is
basically just summing up a set
of binomial coefficients. 
The predicate {\tt fromCantorTuple} implements the
the $\N^n \to \N$ bijection in Prolog,
using the predicate {\tt fromKSet} that sums up
the binomials in formula \ref{kn} using 
the predicate {\tt untuplingLoop},
as well as the sequence to set transformer 
{\tt list2set}.
\begin{code}
fromCantorTuple(Ns,N):-
  list2set(Ns,Xs),
  fromKSet(Xs,N).

fromKSet(Xs,N):-untuplingLoop(Xs,0,0,N).

untuplingLoop([],_L,B,B).
untuplingLoop([X|Xs],L1,B1,Bn):-L2 is L1+1, 
  binomial(X,L2,B),B2 is B1+B,
  untuplingLoop(Xs,L2,B2,Bn).  
\end{code}

\subsubsection{The $\N \to \N^n$ bijection}

We split our problem in two  simpler ones:
inverting {\tt fromKSet}   and then
applying {\tt set2list} to get back from sets to lists.

We observe that the predicate 
{\tt untuplingLoop} used by {\tt fromKSet}
implements
the sum of the combinations 
${X_1 \choose 1}+{X_2 \choose 2}+\ldots+{X_K \choose K}=N$,
which is nothing but the representation of N in the
{\em combinatorial number system of degree $K$}
due to \cite{lehmer64}. Fortunately, efficient
conversion algorithms between the conventional and the
combinatorial number system are well 
known, \cite{knuth_comb}.

We are ready to implement the Prolog predicate {\tt toKSet(K,N,Ds)}, 
which, given the degree {\tt K}, indicating the number 
of ``combinatorial digits'',
finds and repeatedly subtracts the greatest 
binomial smaller than {\tt N}.
It calls the predicate {\tt combinatoriallDigits}
that returns these ``digits''
in increasing order, providing the canonical 
set representations that
{\tt set2list} needs.
\begin{code}
toKSet(K,N,Ds):-combinatoriallDigits(K,N,[],Ds).

combinatoriallDigits(0,_,Ds,Ds).
combinatoriallDigits(K,N,Ds,NewDs):-K>0,K1 is K-1,
  upperBinomial(K,N,M),M1 is M-1,
  binomial(M1,K,BDigit),N1 is N-BDigit,
  combinatoriallDigits(K1,N1,[M1|Ds],NewDs).
\end{code}

\begin{code}  
upperBinomial(K,N,R):-S is N+K,
  roughLimit(K,S,K,M),L is M // 2,
  binarySearch(K,N,L,M,R).
\end{code}
The predicate {\tt roughLimit} compares successive powers of
{\tt 2} with binomials $I \choose K$ and finds the first
{\tt I} for which the binomial is between successive powers
of {\tt 2}.
\begin{code}  
roughLimit(K,N,I, L):-binomial(I,K,B),B>N,!,L=I.
roughLimit(K,N,I, L):-J is 2*I,
  roughLimit(K,N,J,L).
\end{code}
The predicate {\tt binarySearch} finds the exact value
of the combinatorial digit in the interval {\tt [L,M]},
narrowed down by {\tt roughLimit}.
\begin{code}
binarySearch(_K,_N,From,From,R):-!,R=From.
binarySearch(K,N,From,To,R):-Mid is (From+To) // 2,binomial(Mid,K,B),
  splitSearchOn(B,K,N,From,Mid,To,R).

splitSearchOn(B,K,N,From,Mid,_To,R):-B>N,!,
  binarySearch(K,N,From,Mid,R).
splitSearchOn(_B,K,N,_From,Mid,To,R):-Mid1 is Mid+1,
  binarySearch(K,N,Mid1,To,R).  
\end{code}

The predicates {\tt toKSet} and {\tt fromKSet} implement
inverse functions, mapping natural numbers to
canonically represented sets of {\tt K} natural numbers.
\begin{codex}
?- toKSet(5,2014,Set),fromKSet(Set,N).
Set = [0, 3, 4, 5, 14], N = 2014 .
\end{codex}
The efficient inverse of Cantor's N-tupling is now simply:
\begin{code}
toCantorTuple(K,N,Ns):-
  toKSet(K,N,Ds),
  set2list(Ds,Ns).
\end{code}
\BX
illustrates the work of the generalized cantor bijection, on some large numbers:
\begin{codex}
?- K=1000,pow(2014,103,N),toCantorTuple(K,N,Ns),fromCantorTuple(Ns,N).
K = 1000, N = 208029545585703688484419851459547264831381665...567744,
Ns = [0, 0, 2, 0, 0, 0, 0, 0, 1|...] .
\end{codex}
\EX

As the image of a tuple is a polynomial
of degree $n$ it means that the its bitsize is within constant
factor of the sum of the bitsizes of the members of the tuple,
thus:
\begin{prop}\label{cantorSP}
The bijection between $\N^n$ and $\N$ is size-proportionate.
\end{prop}

\subsection{Ranking/unranking of compressed de Bruijn terms}\label{ranks}
We will implement a size-proportionate bijective encoding of
compressed de Bruijn terms following the technique
described in \cite{serpro}. The algorithm will split
a lambda tree into its {\em Catalan skeleton} and
the list of atomic objects labeling its nodes.
In our case, 
the Catalan skeleton abstracts away the applicative structure
of the term. It also provides the key for decoding 
unambiguously the integer labels in both the leaves (two integers)
and internal nodes (one integer).
Our ranking/unranking algorithms will rely on
the encoding/decoding of the Catalan skeleton provided by
the predicates {\tt rankCatalan/2} and {\tt unrankCatalan/2}
as well as for the encoding/decoding of
the labels, provided by the
predicates {\tt toCantorTuple/3} and {\tt fromCantorTuple/2}.

The predicate {\tt rankTerm/2}
defines the bijective encoding of a (possibly open) compressed 
de Bruijn term.
\begin{code}
rankTerm(Term,Code):-
  toSkel(Term,Ps,Ns),
  rankCatalan(Ps,CatCode),
  fromCantorTuple(Ns,VarsCode),
  fromCantorTuple([CatCode,VarsCode],Code).
\end{code}
The predicate {\tt rankTerm/2}
defines the bijective decoding of a natural
number into a (possibly open) compressed  
de Bruijn term.
\begin{code}
unrankTerm(Code,Term):-
  toCantorTuple(2,Code,[CatCode,VarsCode]),
  unrankCatalan(CatCode,Ps), 
  length(Ps,L2),L is (L2-2) div 2, L3 is 3*L+2,
  toCantorTuple(L3,VarsCode,Ns),
  fromSkel(Ps,Ns,Term).
\end{code}
Note that given the unranking of {\tt CatCode}
as a list of balanced parentheses of 
length {\tt 2*L+2}, we can determine the number
{\tt L} of internal nodes of the tree and
the number {\tt L+1} of leaves.
Then we have {\tt 2*(L+1)} labels for
the leaves and {\tt L} labels for the internal nodes,
for a total of {\tt 3L+2}, value needed to
decode the  labels encoded as {\tt VarsCode}.

It follows from Prop. \ref{treeSP} and Prop. \ref{cantorSP} that:
\begin{prop}
A compressed de Bruijn terms is size-proportionate to its rank. 
\end{prop}

\BX
illustrates the ``size-proportionate'' encoding of the compressed de Bruijn terms
corresponding to the combinators {\tt S} and {\tt Y}.
\begin{codex}
?- T = a(3,a(0,v(0,2),v(0,0)),a(0,v(0, 1),v(0,0))),
   rankTerm(T,R),unrankTerm(R,T1).          
T = T1,T1 = a(3,a(0,v(0,2),v(0,0)),a(0,v(0, 1),v(0,0))),
R = 56493141 .

?- T=a(1,a(1,v(0,1),a(0,v(0,0),v(0,0))),a(1,v(0,1),a(0,v(0,0),v(0,0)))),
   rankTerm(T,R),unrankTerm(R,T1).                         
T=T1,T1=a(1,a(1,v(0,1),a(0,v(0,0),v(0,0))),a(1,v(0,1),a(0,v(0,0),v(0,0)))),
R = 261507060 .
\end{codex}
\EX

\subsubsection{Generation of lambda terms via unranking}
While direct enumeration of terms constrained by number of nodes or depth
is  straightforward in Prolog, an unranking algorithm is also
usable for term generation, including generation of random terms.

\paragraph{Generating open terms in compressed de Bruijn form}
Open terms are generated simply by iterating over an initial segment of $\N$ with the
built-in {\tt between/3} and calling the predicate {\tt unrankTerm/2}.
\begin{code}  
ogen(M,T):-between(0,M,I),unrankTerm(I,T).  
\end{code}
Reusing unranking-based open term generators for more constrained families
of lambda terms works when their
asymptotic density is relatively high. 

\paragraph{Generating closed and well-typed terms in compressed de Bruijn form}
The extensive quantitative analysis available in the literature
\cite{grygielGen,ranlamb09,normalizing13} indicates that density
of closed and typed terms decreasing very quickly with size,
making generation by filtering impractical for very large
terms.

The predicate {\tt cgen/2} generates closed terms by filtering
the results of {\tt ogen/2} with the predicate {\tt isClosedC}
and {\tt tgen} generates typable terms by filtering
the results of {\tt cgen/2} with {\tt typable/2}.
\begin{code}
cgen(M,IT):-ogen(M,IT),isClosedC(IT).

tgen(M,IT):-cgen(M,IT),typable(IT).
\end{code}

\BX
Generation of well-typed terms via unranking.
\begin{codex}
?- tgen(200,T).
T = v(1, 0) ;
T = v(2, 0) ;
T = v(2, 1) ;
T = v(3, 0) ;
T = v(3, 1) ;
T = v(4, 0) ;
T = a(0, v(1, 0), v(1, 0)) ;
T = a(1, v(0, 0), v(1, 0)) ;
T = v(3, 2) ;
T = v(4, 1) .
false.
\end{codex}
\EX

\subsection{X-combinator trees as natural numbers}\label{ntree}
G\"odel-numberings seen as injective mappings from formulas and proofs
to natural numbers have been used for important
theoretical results in the past \cite{conf/icalp/HartmanisB74}
among which G\"odel's
incompleteness theorems are the most significant \cite{Goedel:31}.

In the form of ranking and unranking functions, bijections
from families of combinatorial objects to natural numbers
have been devised with often practical uses in mind,
like generation of random inputs for software testing.

Ensuring that such bijections are also size-proportionate,
adds an additional challenge to the problem, as the fast
growth of the number of combinatorial objects of a given
size makes it difficult to impossible to associate to all of them
comparably small unique natural numbers. 
As another challenge,  computation of the
unranking function often involves some form of binary
or multiway tree search to
locate the object corresponding to a given 
natural number \cite{grygielGen,serpro},
which precludes their use on very large objects.
Our solution described here consists in two steps, the second one
involving an arguably surprising twist.

First, we define a bijection between natural numbers and trees.
Next we define arithmetic operations directly on trees
and ensure that they mimic exactly their natural number
equivalents.
{\em This turns our trees into natural numbers
(they become yet another model or Peano's axioms),
hence we can make them the target of ranking algorithms
and the source of unranking ones.}

As we are now dealing with bijections between trees and tree-like
data structures, making them size proportionate  becomes
surprisingly easy. 
We will define such a bijection to general lambda terms
in section \ref{goedel}.

\subsubsection{A bijection from binary trees to natural numbers}\label{tn}
The  (big-endian) binary representation of a natural number can be written as a concatenation of binary digits of the form
\begin{equation} \label{bin}
n=b_0^{k_0}b_1^{k_1}\ldots b_i^{k_i} \ldots b_m^{k_m}
\end{equation}
with  $b_i \in \{0,1\}$ and the highest digit  $b_m=1$.
The following hold.
\begin{prop}
An even number of the form $0^ij$ corresponds
to the operation $2^ij$ and an odd number of the form $1^ij$ corresponds
to the operation $2^i(j+1)-1$. 
\end{prop}
\begin{proof}
It is clearly the case that $0^ij$ corresponds to multiplication by a power of $2$. If $f(i)=2i+1$, then it can be shown by induction that the $i$-th iterate of $f$, $f^i$ is computed as in the equation (\ref{fiter})
\begin{equation}
f^i(j)=2^i(j+1)-1 \label{fiter}
\end{equation}
Observe that each block $1^i$ in $n$, represented as $1^ij$ in equation (\ref{bin}), corresponds to the iterated application of $f$, $i$ times,
$n=f^i(j)$.
\end{proof}

\begin{prop} \label{parity}
A number $n$ is even if and only if it contains an even number of blocks of the form 
$b_i^{k_i}$ in equation (\ref{bin}). 
A number $n$ is odd if and only if it contains an odd number of blocks of the form 
$b_i^{k_i}$ in equation (\ref{bin}). 
\end{prop}
\begin{proof}
It follows from the fact that the highest digit (and therefore the last block in big-endian representation) is $1$ and the parity of the blocks alternate.
\end{proof}

This suggests defining a {\tt cons} operation on natural numbers
as follows.
\begin{equation}\label{cnat}
cons(i,j)=
\begin{cases}
{2^{i+1}j}  & {\text{if } j \text{ is odd}},\\
{2^{i+1}(j+1)-1} & {\text{if } j \text{ is even}}.
\end{cases}
\end{equation}
Note that the exponents are $i+1$ instead of $i$ as we start counting at $0$. Note also
that  $cons (i,j)$ will be even when $j$ is odd and odd when $j$ is even.

\begin{prop}
The equation (\ref{cnat}) defines a bijection  $c: \N \times \N \to  \N^+=\N-\{0\}$.
\end{prop}
Therefore {\tt cons} has an inverse {\tt decons}, that we will constructively define
together with it. 
\begin{code}
cons(I,J,C) :- I>=0,J>=0,
  D is mod(J+1,2),
  C is 2^(I+1)*(J+D)-D.
\end{code}
The definition of the inverse {\tt decons} relies on the {\em dyadic valuation} of a number $n$, $\nu_2(n)$, defined as the largest exponent of 2 dividing $n$,
 implemented as the helper predicate {\tt dyadicVal}, which computes the least significant bit of
its first argument with help from the built-in {\tt lsb}.
\begin{code}
decons(K,I1,J1):-K>0,B is mod(K,2),KB is K+B,
  dyadicVal(KB,I,J),
  I1 is max(0,I-1),J1 is J-B.

dyadicVal(KB,I,J):-I is lsb(KB),J is KB // (2^I).
\end{code}

\BX
The inverse {\tt cons} and {\tt decons} operations.
\begin{codex}
?- decons(2016,A,B),cons(A,B,N).
A = 4,
B = 63,
N = 2016.
\end{codex}
\EX

We can compute a natural number from an X-combinator tree
by mapping recursively the ``{\tt >}'' constructor to {\tt cons}.
\begin{code}
n(x,0).
n((A>B),K):-n(A,I),n(B,J),cons(I,J,K).
\end{code}
Similarly, we can build an X-combinator tree
from a natural number by recursing over {\tt decons}.
\begin{code}
t(0,x).
t(K,(A>B)):-K>0,decons(K,I,J),t(I,A),t(J,B).
\end{code}
Note the small codes corresponding to some interesting combinators.
\BX
Encodings of combinators X, S, K and XX=KK.
\begin{codex}
?- n(x,N).
N = 0.
?- n(x>x,N).    
N = 1.
?- sT(X),n(X,N).
X = (x> (x>x)), N = 2.
?- kT(X),n(X,N).
X = ((x>x)>x), N = 3.
\end{codex}
\EX

\begin{prop}
The predicates {\tt n} and {\tt t} define inverse functions
between natural numbers and X-combinator trees.
\end{prop}
\begin{proof}
It follows from the fact that {\tt cons} and {\tt decons}
implement inverse functions.
\end{proof}

\BX
The work of {\tt t} and {\tt n} on the first 8 natural numbers.
\begin{codex}
?- maplist(t,[0,1,2,3,4,5,6,7],Ts),maplist(n,Ts,Ns).
Ts = [x,x>x,x> (x>x), (x>x)>x, (x>x)> (x>x),
      x> (x> (x>x)),x> ((x>x)>x), (x> (x>x))>x],
Ns = [0, 1, 2, 3, 4, 5, 6, 7].
\end{codex}
\EX

\subsubsection{Binary tree arithmetic}\label{sucpred}

As we know for sure that
natural numbers support arithmetic operations, we will
try to mimic their behavior with  binary trees built with
the constructor ``{\tt >}'' and empty leaves {\tt x}
that we have interpreted so far as X-combinator expressions
and simple types.
 
The operations {\tt even\_} and {\tt odd\_} implement the
observation following from of Prop. \ref{parity} that parity (staring with $1$
at the highest block) alternates with each block of distinct $0$ or $1$ digits.
\begin{code}
parity(x,0).
parity(_>x,1).
parity(_>(X>Xs),P1):-parity(X>Xs,P0),P1 is 1-P0.

even_(_>Xs):-parity(Xs,1).
odd_(_>Xs):-parity(Xs,0).
\end{code}

We will now specify successor and predecessor
through two mutually recursive predicates, {\tt s} and {\tt p}.

They first decompose their arguments as if using {\tt decons}. Then,
after transforming them as a result of adding {\tt 1},
they place back the results as if using the {\tt cons} operation,
both emulated by the use of the constructor ``{\tt >}''.
Note that the two functions work on trees with steps corresponding to
{\em a block of {\tt 0} or {\tt 1}
digits at a time}. They are based on arithmetic observations
about the behavior of these blocks when incrementing or decrementing
a binary number by {\tt 1}.
\begin{code}
s(x,x>x).
s(X>x,X>(x>x)):-!.
s(X>Xs,Z):-parity(X>Xs,P),s1(P,X,Xs,Z).
\end{code}
After computing parity, 
the successor predicate {\tt s} delegates the transformation of
the blocks of $0$ and $1$ digits to predicate  {\tt s1}
handling both the {\tt even\_} and  {\tt odd\_} cases.

\begin{code}
s1(0,x,X>Xs,SX>Xs):-s(X,SX).
s1(0,X>Ys,Xs,x>(PX>Xs)):-p(X>Ys,PX).
s1(1,X,x>(Y>Xs),X>(SY>Xs)):-s(Y,SY).
s1(1,X,Y>Xs,X>(x>(PY>Xs))):-p(Y,PY).
\end{code}

The predecessor function {\tt p} inverts the work of {\tt s}
\begin{code}
p(x>x,x).
p(X>(x>x),X>x):-!.
p(X>Xs,Z):-parity(X>Xs,P),p1(P,X,Xs,Z).
\end{code}
After computing parity,
the predecessor predicate {\tt p} delegates the transformation of
the blocks of $0$ and $1$ digits to  {\tt p1}
handling separately the
{\tt even\_} and {\tt odd\_} cases.
\begin{code}
p1(0,X,x>(Y>Xs),X>(SY>Xs)):-s(Y,SY).
p1(0,X,(Y>Ys)>Xs,X>(x>(PY>Xs))):-p(Y>Ys,PY).
p1(1,x,X>Xs,SX>Xs):-s(X,SX).
p1(1,X>Ys,Xs, x>(PX>Xs)):-p(X>Ys,PX).
\end{code}

\begin{prop}
Assuming parity information is kept explicitly, 
the operations {\tt s} and {\tt p} work on
a binary tree of size $N$ in time 
constant on average and
and $O(log^*(N))$ in the worst case
\end{prop}
\begin{proof}
See \cite{arxiv_cats}.
\end{proof}

\begin{prop}
The operations {\tt s} and {\tt p} implement successor and predecessor
operations such that their results correspond to the same operations
on natural numbers,i.e., the following hold.
\begin{equation}
t(A,X),s(X,Y),B~is~A+1,n(Y,C) \rightarrow B=C
\end{equation}

\begin{equation}
t(A,X),p(X,Y),B~is~A-1,n(Y,C) \rightarrow B=C
\end{equation}
\end{prop}
\begin{proof}
See \cite{arxiv_cats}.
\end{proof}

\BX
{\tt s} and {\tt p} implement arithmetic correctly.
\begin{codex}
?- A=10,t(A,X),s(X,Y),B is A+1,n(Y,C).       
A = 10,X = (x> (x> (x> (x>x)))),Y = ((x>x)> (x> (x>x))),
B = C, C = 11 .
\end{codex}
\EX

Our binary trees can be seen as a model of {\em Peano Arithmetic},
in the same sense as unary or binary arithmetic. Note also,
that while any enumeration would provide  
unary arithmetic, our representation
implements the equivalent (or better) of {\em binary arithmetic}.
We refer to \cite{lata14} and \cite{sacs14tarau} for 
the description of algorithms covering
all the usual arithmetic operations with
equivalent representations working on other members of
the Catalan family and to \cite{arxiv_cats} for
a generic implementation using Haskell type classes.
Hence our X-combinator trees
can provide an implementation of arithmetic operations (including
extension to integers and rational numbers). 
Moreover, they can also
become the target of ranking and unranking
functions that associate unique natural number codes
to various combinatorial objects. In section
\ref{goedel} they will play this role for general
lambda terms.

We refer to \cite{arxiv_cats} for the development
of a complete arithmetic system for the Catalan family 
of combinatorial objects, of
which binary trees are the most well known instance.

\subsection{A size-proportionate G\"odel-numbering bijection for lambda terms}\label{goedel}

We are finally ready to define our simple, linear time,
size-proportionate bijection between tree-represented natural
numbers and general lambda terms in de Bruijn notation.

\subsubsection{Ranking and unranking de Bruijn terms to binary-tree represented natural numbers}\label{binrank}

The predicate {\tt rank} defines
a bijection from lambda expressions in de Bruijn notation
to binary trees, seen here as implementing natural numbers.
Variables {\tt v/1} are represented as trees with the left {\tt x}
as their left branch, lambdas {\tt l/1} as trees with {\tt x} as their
right branch. To avoid ambiguity, ranks for application nodes will be 
incremented by one using the successor predicate {\tt s/2}.
\begin{code}
rank(v(0),x).
rank(l(A),x>T):-rank(A,T).
rank(v(K),T>x):-K>0,t(K,T).
rank(a(A,B),X1>Y1):-
  rank(A,X),s(X,X1),
  rank(B,Y),s(Y,Y1).
\end{code}
The predicate {\tt unrank} defines the inverse bijection
from binary trees, seen as natural numbers, to
lambda expressions in de Bruijn notation. It works by
case analysis on trees with branches marked with {\tt x}
and decrements branches using predicate 
{\tt p/2} to ensure it inverts the action of {\tt
rank} on application nodes. Note also that both predicates
use the bijections {\tt t} and respective {\tt n}
to convert between tree-based naturals and their standard
natural number equivalents.
\begin{code}
unrank(x,v(0)).
unrank(x>T,l(A)):-!,unrank(T,A).
unrank(T>x,v(N)):-!,n(T,N).
unrank(X>Y,a(A,B)):-
  p(X,X1),unrank(X1,A),
  p(Y,Y1),unrank(Y1,B).
\end{code}
\begin{prop}
Assuming variable indices are small (word-size) integers,
{\tt rank} and {\tt unrank} define a size-proportionate bijection between
lambda terms in de Bruijn form and X-combinator trees. 
Their runtime is proportional
to the size of their input.
\end{prop}
\begin{proof}
If variable indices are fixed sized small integers, one can assume that
{\tt t} and {\tt n} work in constant time.
Then, observe that each step of both predicates works in time proportional to {\tt s} or
{\tt p} for a total proportional to the number of internal nodes.
\end{proof}
As an interesting variation, for very large terms,
one could actually {\em use binary tree-based natural numbers
for the indices of {\tt v/1} in de Bruijn terms},
and completely bypass the use of {\tt t} and {\tt n}, and
thus lifting the assumption about variable indices being fixed size integers.
\BX
Ranking and unranking of K and S combinators in de Bruijn form.
\begin{codex}
?- kB(K),rank(K,B),unrank(B,K1).
K = K1, K1 = l(l(v(1))),
B = (x> (x> ((x>x)>x))) .

?- sB(S),rank(S,B),unrank(B,S1).
S = S1, S1 = l(l(l(a(a(v(2), v(0)), a(v(1), v(0)))))),
B = (x> (x> (x> ((x> (((x> (x>x))>x)> (x>x)))> (x> (((x>x)>x)> (x>x))))))) .
\end{codex}
\EX


\section{Playing with the playground} \label{play}

The following quote from Donald Knuth, in answering a question of Frank Ruskey about
the short term economics behind research ({\small \url{http://www.informit.com/articles/article.aspx?p=2213858}})
 and prominently displayed at Mayer Goldberg's
home page at { \url{http://www.little-lisper.org/website/}},
summarizes our motivation behind building this declarative playground:
\begin{quote}
Everybody seems to understand that astronomers do astronomy because astronomy is interesting. Why don't they understand that I do computer science because computer science is interesting? 
\end{quote}
This being said, we will sketch here a few use cases, some of possible practical significance.

\subsection{Querying a generator for specific types} \label{tquery}

Coming with Prolog's unification and non-deterministic search,
is the ability to make more specific queries
by providing a type pattern, that selects only terms
that match it, while generating terms and inferring their types.

The predicate {\tt queryTypedTerm} finds
closed terms of a given type of
size exactly {\tt L}.
\begin{code}
queryTypedTerm(L,QueryType,Term):-
  genTypedTerm(L,Term,QueryType),
  boundTypeOf(Term,QueryType).
\end{code}
Similarly, the predicate {\tt queryTypedTerm} finds
closed terms of a given type of
size {\tt L} or less.
\begin{code}  
queryTypedTerms(L,QueryType,Term):-
  genTypedTerms(L,Term,QueryType),
  boundTypeOf(Term,QueryType).
\end{code}
Note that giving the query type ahead of executing
{\tt genTypedTerm} would unify with more general
``false positives'', as type checking, contrary 
to type synthesis, proceeds bottom-up.
This justifies filtering out the false positives 
simply by testing with the deterministic predicate
{\tt boundTypeOf} at the end. Despite
the extra call to {\tt boundTypeOf}, the
performance improvements are significant, as shown in
Figure \ref{perf}. The figure also shows that when the
slow generate-and-test predicate {\tt genTypedTerm1} is used,
the result (in ``logical-inferences-per-second'') does
not depend on the pattern, contrary to the fast {\tt
queryTypedTerm} that prunes mismatching types while
inferring the type of the  terms as it generates them.

\begin{figure}
\begin{center}
\begin{tabular}{|r||r|r|r|c|c||}
\cline{1-6}
\cline{1-6}
  \multicolumn{1}{|c||}{{Size}} & 
  \multicolumn{1}{c|}{{ {\small Slow \verb~x>x~}}} &
  \multicolumn{1}{c|}{{ {\small Slow \verb~x>(x>x)~}}} &
  \multicolumn{1}{c|}{{ {\small Fast \verb~x>x~}}} &
  \multicolumn{1}{c|}{{ {\small Fast \verb~x>(x>x)~}}} &
  \multicolumn{1}{c||}{{ {\small Fast \verb~x~}}} \\
\cline{1-6} 
\cline{1-6}

1 & 39& 39& 38 &27&15       \\ \cline{1-6}
2 & 126& 126& 60&109&36        \\ \cline{1-6}
3 &  552& 552&240&200&88        \\ \cline{1-6}
4 &  3,108&3,108& 634& 1,063& 290     \\ \cline{1-6}
5 &  21,840& 21,840& 3,213&3,001& 1,039   \\ \cline{1-6}
6 &  181,566&181,566& 12,721&19,598&4,762    \\ \cline{1-6}
7 &  1,724,131& 1,724,131&76,473&81,290& 23,142 \\ \cline{1-6}
8 &  18,307,585& 18,307,585&407,639&584,226&  133,554\\ \cline{1-6}
9 &  213,940,146& 213,940,146&2,809,853&3,254,363 &812,730\\ \cline{1-6}

\end{tabular} \\
\medskip
\caption{Number of logical inferences as counted by SWI-Prolog for our algorithms when querying generators
with type patterns given in advance
\label{perf}}
\end{center}
\end{figure}

\BX 
Terms of type {\tt x>x} of size 4.
\begin{codex}
?- queryTypedTerm(3,(x>x),Term). 
Term = a(l(v(0)), l(v(0))) ;
Term = l(a(l(v(0)), v(0))) ;
Term = l(a(l(v(1)), v(0))) .

?- queryTypedTerms(12,(x>x)>x,T). 
false.
\end{codex}
\EX
Note that the last query, taking about a minute,
shows that no  closed terms
of type {\tt (x>x)>x}
exist up to size 12. In fact, it is known that no such terms exist, as 
the corresponding logic formula is not a tautology in minimal logic.

\subsection{Same-type siblings}
Given a closed well-typed
lambda term, we can ask what other terms of the same
size or smaller share the same type. This can be interesting for
finding possibly alternative implementations of a given function
or for generation of similar siblings in genetic programming.

The predicate {\tt typeSiblingOf} lists all the terms of the same or smaller 
size having the same type as a given term.
\begin{code}
typeSiblingOf(Term,Sibling):-
  dbTermSize(Term,L),
  boundTypeOf(Term,Type),
  queryTypedTerms(L,Type,Sibling).
\end{code}

\BX
\begin{codex}
?- typeSiblingOf(l(l(a(v(0),a(v(0),v(1))))),T).
T = l(l(a(v(0), v(1)))) ; % <= smaller sibling
T = l(l(a(v(0), a(v(0), v(1))))) .
\end{codex}
\EX

\subsection{Discovering  frequently occurring type patterns}\label{pats}

The ability to run ``relational queries'' about terms and their types
extends to compute interesting statistics, giving a glimpse at
their distribution.

\subsubsection{The ``Popular'' type patterns}

As types can be seen as an approximation of their inhabitants,
we expect them to be shared among distinct terms. As we can enumerate
all the terms for small sizes and infer their types, we would like to know
what are the most frequently occurring ones. This can be meaningful
as a comparison base for types that are used in human-written programs
of comparable size. In approaches like \cite{palka11}, where types
are used to direct the generation of random terms, focusing on
the most frequent types might help with generation of more realistic
random tests.

Figure \ref{ratio} describes counts for terms and their types for small sizes. It also
shows the first two most frequent types with the count of terms they apply to.

\begin{codeh}  
countForType(GivenType,M,Rs):-
  findall(R,
    (
       between(1,M,L),sols(queryTypedTerm(L,GivenType,_B),R)
    ),
  Rs).  
\end{codeh}

\begin{figure}
\begin{center}
\begin{tabular}{|r||r|r|r|c|c|}
\cline{1-6}
\cline{1-6}
\multicolumn{1}{|c|}{{Term size}} & 
  \multicolumn{1}{c|}{{ Types}} &
  \multicolumn{1}{c|}{{ Terms}} &
  \multicolumn{1}{c|}{{ Ratio}} &
  \multicolumn{1}{c|}{{ 1-st frequent}} &
 \multicolumn{1}{c|}{{ 2-nd frequent}} \\
\cline{1-6} 
\cline{1-6}

1&1&1&1.0& 1: \verb~x>x~ & \\ \cline{1-6}
2&1&2&0.5&  2: \verb~x>(x>x)~ & \\ \cline{1-6}
3 & 5 & 9 & 0.555  &  3: \verb~x>(x>(x>x))~  &  3: \verb~x>x~ \\ \cline{1-6}
4 & 16 & 40 & 0.4 & 14: \verb~x>(x>x)~  & 4: \verb~x>x>(x>(x>x))~  \\ \cline{1-6}
5 & 55 & 238 & 0.231  & 38: \verb~x>(x>(x>x))~ & 31: \verb~x>x~  \\ \cline{1-6}
6 & 235 & 1564 & 0.150  & 201: \verb~x>(x>x)~ & 80: \verb~x>x>(x>(x>x))~  \\ \cline{1-6}
7 & 1102 & 11807 & 0.093  & 732: \verb~x>(x>(x>x))~ & 596: \verb~x>x~  \\ \cline{1-6}
8 & 5757 & 98529 & 0.058  & 4632: \verb~x>(x>x)~ & 2500: \verb~x>x~  \\ \cline{1-6}
9 &  33251 &  904318 &  0.036  &  20214: \verb~x>(x>(x>x))~ &  19855: \verb~(x>x)>(x>x)~  \\ \cline{1-6}

\end{tabular} \\
\medskip
\caption{Counts for terms and types for sizes 1 to 9 and the first two most frequent types
\label{ratio}}
\end{center}
\end{figure}

Figure \ref{freq} shows the ``most popular types'' for the about 1 million
closed well-typed terms up to size 9 and the count of their inhabitants.

\begin{figure}
\begin{center}
\begin{tabular}{|r|l|}
\cline{1-2}
\cline{1-2}
\multicolumn{1}{|c|}{{Count}} & \multicolumn{1}{c|}{{Type}}\\
\cline{1-2} 
\cline{1-2}
\cline{1-2}
23095 & \verb~x>(x>x)~ \\ \cline{1-2}
22811 & \verb~(x>x)>(x>x)~ \\ \cline{1-2}
22514 & \verb~x>x>(x>x)~ \\ \cline{1-2}
21686 & \verb~x>x~ \\ \cline{1-2}
18271 & \verb~x> ((x>x)>x)~ \\ \cline{1-2}
14159 & \verb~(x>x)>(x>(x>x))~ \\ \cline{1-2}
13254 & \verb~((x>x)>x)> ((x>x)>x)~ \\ \cline{1-2}
12921 & \verb~x> (x>x)>(x>x)~ \\ \cline{1-2}
11541 & \verb~(x>x)> ((x>x)>x)>x~ \\ \cline{1-2}
10919 & \verb~(x>(x>x))>(x>(x>x))~ \\ \cline{1-2}
\end{tabular} \\
\medskip
\caption{Most frequent types, out of a total of {\tt 33972} distinct types, of {\tt 1016508} terms up to size 9.
\label{freq}}
\end{center}
\end{figure}
We can observe that, like in some human-written programs, 
functions representing 
binary operations of type \verb~x>(x>x)~ are the most popular.
Ternary operations \verb~x>(x>(x>x))~ come third and unary operations \verb~x>x~ come fourth.
Somewhat surprisingly, a higher order function type
\verb~(x>x)>(x>x)~ applying a function to an argument
to return a result comes second and multi-argument variants of
it are also among the top 10.

\subsubsection{Growth sequences of some popular types}
We can make use of our generator's  efficient
specialization to a given type to explore
 empirical estimates for some types interesting to
 human  programmers.

Contrary to the total absence of the type {\tt (x>x)>x} among
terms of size up to 12, ``binary operations'' of type
{\tt x>(x>x)} turn 
out to be quite frequent, giving, by increasing sizes,
the sequence                    
[0, 2, 0, 14, 12, 201, 445, 4632, 17789, 158271, 891635].

{\em Transformers} of type
{\tt x>x}, by increasing sizes, give
the  sequence [1, 0, 3, 3, 31, 78, 596, 2500, 18474, 110265].
While type {\tt (x>x)>x} turns our to be absent up to size 12,
the type {\tt (x>x)>(x>x)}, describing {\em transformers of transformers}
turns out to be quite popular, as shown by the sequence
[0, 0, 1, 1, 18, 52, 503, 2381, 19855, 125599].
The same turns out to be true also for
{\tt (x>x)>((x>x)>(x>x))}, giving
[0, 0, 0, 0, 2, 6, 96, 505, 5287, 36769]
and {\tt ((x>x)>(x>x)) > ((x>x)>(x>x))} giving
[0, 0, 0, 0, 0, 6, 23, 432, 2450, 29924].
One might speculate that homotopy type theory \cite{htypes},
that focuses on such transformations 
and transformations of transformations etc.
has a rich population of lambda terms from which
to chose interesting inhabitants of such types!

Another interface, generating closed simply-typed terms
of a given size, restricted to have at most a given number of
free de Bruijn indices, is implemented by the predicate
{\tt genTypedWithSomeFree}.
\begin{code}
genTypedWithSomeFree(Size,NbFree,B,T):-
   between(0,NbFree,NbVs),
   length(FreeVs,NbVs),
   genTypedTerm(B,T,FreeVs,Size,0),
   bindTypeB(T).
\end{code}
The first 9 numbers counting closed simply-typed terms with at most
one free variable (not yet in \cite{intseq}), 
are [3, 10, 45, 256, 1688, 12671, 105743, 969032, 9639606].

Note that, as our generator performs the early pruning of untypable terms,
rather than as a post-processing step, enumeration and counting of these 
terms happens in a few seconds.

\begin{codeh} 
typeCountsFor(L,T:Count):-
  setof(X,genTypedTerm(L,X,T),Xs),
  length(Xs,Count).
  
typeCountsUpTo(L,T:Count):-
  setof(X,genTypedTerms(L,X,T),Xs),
  length(Xs,Count).
  
popularTypesFor(L,K,[Difs,Sols,Ratio],PopularTs):-
  sols(genTypedTerm(L,_,_),Sols),
  setof(Count-T,typeCountsFor(L,T:Count),KTs),
  reverse(KTs,Sorted),
  take_(K,Sorted,PopularTs),
  length(KTs,Difs),
  Ratio is Difs/Sols.

popularTypesUpTo(L,K,[Difs,Sols,Ratio],PopularTs):-
  sols(genTypedTerms(L,_,_),Sols),
  setof(Count-T,typeCountsUpTo(L,T:Count),KTs),
  reverse(KTs,Sorted),
  take_(K,Sorted,PopularTs),
  length(KTs,Difs),
  Ratio is Difs/Sols.
  
take_(0,_Xs,[]):-!.
take_(_,[],[]):-!.
take_(N,[X|Xs],[X|Ys]):-N>0,N1 is N-1,take_(N1,Xs,Ys).

\end{codeh}


\subsection{Generating closed typable lambda terms by types}\label{typedir}

In \cite{palka11}  a ``type-directed'' mechanism for
the generation of random terms is introduced,
resulting in more realistic (while not uniformly random) terms,
used successfully in discovering some GHC bugs.

We can organize in a similar way the interface of
our combined generator and type inferrer.

\subsubsection{Generating type trees}

The predicate {\tt genType} generates binary trees representing simple
types with a single base type {\tt ``x''}. As we represent types as
binary trees with leaves {\tt x} and internal nodes \verb~>~ we
can reuse the predicate {\tt genTree}.
\begin{code}
genType --> genTree.
genTypes --> genTrees.
\end{code}
Like {\tt genTree},
it provides two interfaces, for generating types of exactly size N or
up to size N.

Next, we will combine this type generator with the generator
that efficiently produces terms matching each type pattern.

\subsubsection{Generating lambda terms by increasing type sizes}

The predicate {\tt genByType} first generates types
(seen simply as binary trees) with {\tt genType}
and then uses the unification-based querying mechanism
to generate all closed well-typed  de Bruijn terms
 with fewer internal nodes then 
their binary tree type.
\begin{code}
genByType(L,B,T):-
  genType(L,T),
  queryTypedTerms(L,T,B).
\end{code}

\BX
Enumeration of closed simply-typed de Bruijn terms
with types of size 3 and terms of a given type
with at most 3 internal nodes.
\begin{codex}
?- genByType(3,B,T).
B = l(l(l(v(0)))),
T = (x> (x> (x>x))) ;
B = l(l(l(v(1)))),
T = (x> (x> (x>x))) ;
B = l(l(l(v(2)))),
T = (x> (x> (x>x))) ;
B = l(l(a(v(0), v(1)))),
T = (x> ((x>x)>x)) ;
B = l(l(a(v(1), v(0)))),
T = ((x>x)> (x>x)) ;
B = l(a(v(0), l(v(0)))),
T = (((x>x)>x)>x) .
\end{codex}
\EX
Given that various constraints are naturally interleaved by our generator
we obtain in a few seconds the sequence counting these terms having types
up to size 8, [1, 2, 6, 18, 84, 376, 2344, 15327].
Intuitively this means that despite of their growing sizes, types
have an increasingly large number of inhabitants of sizes 
smaller than their size. This is somewhat contrary to what
we see in human-written code, where types are almost always
simpler and smaller than the programs inhabiting them.

\begin{codeh}
countByType(M,Rs):-
  findall(R,
    (
       between(1,M,L),sols(genByType(L,_B,_T),R)
    ),
  Rs).
\end{codeh}

\subsubsection{Generation of random lambda terms}
Generation of random lambda terms, resulting from the unranking
of random integers of a give bit-size, is implemented by
the predicate {\tt ranTerm/3}, that applies the predicate {\tt Filter}
repeatedly until a term is found for which the predicate {\tt Filter} holds.
\begin{code}
ranTerm(Filter,Bits,T):-X is 2^Bits,N is X+random(X),M is N+X,
  between(N,M,I),
   unrankTerm(I,T),call(Filter,T),
  !.
\end{code}
Random open terms are generated by {\tt ranOpen/2}, random closed
terms by the predicate {\tt ranClosed}, random typable term by
{\tt ranTyped} and closed typable terms
by {\tt closedTypable/2}.
\begin{code}  
ranOpen(Bits,T):-ranTerm(=(_),Bits,T).  

ranClosed(Bits,T):-ranTerm(isClosedC,Bits,T).

ranTyped(Bits,T):-ranTerm(closedTypable,Bits,T).

closedTypable(T):-isClosedC(T),typable(T).
\end{code}
Open terms based on unranking random numbers of 3000 bits of size above 1000, 
closed terms of size above 55 for 150 bits  and closed 
typable terms of size above 13
for 30 bits can be generated within a few seconds. The limited scalability
for closed and well-typed terms is a consequence of their low asymptotic 
density, as shown in \cite{ranlamb09,grygielGen}. We refer to \cite{grygielGen} 
for algorithms supporting random generation of large lambda terms.
\BX
illustrates generation of some closed and well-typed terms
in compressed de Bruijn form.
\begin{codex}
?- ranClosed(10,T).
T = a(1, a(0, v(0, 0), v(0, 0)), a(0, a(0, v(0, 0), v(0, 0)), v(1, 0))).

?- ranTyped(20,T).
T = a(3, v(3, 1), v(2, 0)).
\end{codex}
\EX

\subsection{Estimating the proportion of well-typed SK-combinator trees}
Given the low density of closed well-typed lambda terms,
an interesting question arises at this point:  {\em
what proportion of SK-combinator trees of a given size are well-typed}?
While the analytic study of the 
{\em asymptotic density} 
has been successfully performed on several families of lambda terms
\cite{bodini11,normalizing13,BCI13,grygielGen},
it is considered an open 
problem for well-typed terms. We will limit ourselves here to empirically estimate it,
as it is done in \cite{grygielGen} for general lambda terms, where experiments indicate
the extreme sparsity for very large terms.

We can use our  generator {\tt genSK} to enumerate SK-combinator trees
among which we can then count the number of
well-typed ones.
\BX
Types inferred for terms with 2 internal nodes.
\begin{codex}
?- genSK(1,X),simpleTypeOf(X,T).                                                            
X = k*k,
T = (x> (x> (x>x))) ;
X = k*s,
T = (x> ((x> (x>x))> ((x>x)> (x>x)))) ;
X = s*k,
T = ((x>x)> (x>x)) ;
X = s*s,
T = (((x> (x>x))> (x>x))> ((x> (x>x))> (x>x))) .
\end{codex}
\EX

Similarly, we can use it also to enumerate untypable terms.
\BX
The smallest two untypable SK-expressions.
\begin{codex}
?- genSKs(2,X), \+typableSK(X).
X = s*s*k ;
X = s*s*s .
\end{codex}
\EX

We can implement a generator for well-typed SK-trees,
to be used to compute the ratio between the number of well-typed SK-trees and 
the total number of SK-trees of size $n$, as well as one for the
untypable SK-trees.
\begin{code}
genTypedSK(L,X,T):-genSK(L,X),simpleTypeOf(X,T).

genUntypableSK(L,X):-genSK(L,X),\+skTypeOf(X,_).
\end{code}

To compute the proportion of well-typed terms among terms of a given size
we will also need to count the number of SK-trees with $n$
internal nodes.

\begin{prop}
There are $2^{n+1}C_n$ SK-trees with $n$ nodes, where $C_n$ is the $n$-th Catalan number.
\end{prop}
\begin{proof}
If follows from the fact that $C_n$ counts the number of binary trees with $n$ internal nodes,
each of which has $n+1$ leaves, each of which can be either $S$ or $K$.
\end{proof}

The predicate {\tt cat/2} computes the nth-Catalan number efficiently using the recurrence
$C_0=1, C_{n}={{2(2n-1)}\over{n+1}}C_{n-1}$ \cite{StanleyEC}.
\begin{code}
cat(0,1).
cat(N,R):-N>0, 
  PN is N-1,
  cat(PN,R1),
  R is 2*(2*N-1)*R1//(N+1).
\end{code}

\begin{codeh}  
countTyped(L,Typed,SKs,Prop):-
  sols(genTyped(L,_,_),Typed),
  cat(L,Cats),SKs is 2^(L+1)*Cats,
  Prop is Typed/SKs.
  
tcounts:-
  between(0,9,I),countTyped(I,Typed,All,Prop),
  write((I,&,Typed,&,All,&,Prop)),nl,fail.
\end{codeh}

Figure \ref{ratioSK} shows the counts
for well-typed SK-combinator expressions
and their ratio to the total number of SK-trees
of given size.

\begin{figure}
\begin{center}
\begin{tabular}{||r||r|r|r|}
\cline{1-4}
\cline{1-4}
\multicolumn{1}{||c||}{{Term size}} & 
  \multicolumn{1}{c|}{{ Well-typed}} &
  \multicolumn{1}{c|}{{ Total}} &
  \multicolumn{1}{c|}{{ Ratio}}\\
\cline{1-4} 
\cline{1-4}
0&2& 2&1 \\ \cline{1-4}
1&4& 4&1 \\ \cline{1-4}
2&14& 16&0.875 \\ \cline{1-4}
3&67& 80&0.8375 \\ \cline{1-4}
4&337& 448&0.752 \\ \cline{1-4}
5&1867& 2688&0.694 \\ \cline{1-4}
6&10699& 16896&0.633 \\ \cline{1-4}
7&63567& 109824&0.578 \\ \cline{1-4}
8&387080& 732160&0.528 \\ \cline{1-4}
9&2401657& 4978688&0.482 \\ \cline{1-4}
\end{tabular} \\
\medskip
\caption{Proportion of well-typed SK-combinator terms
\label{ratioSK}}
\end{center}
\end{figure}
Somewhat surprisingly, a large
proportion of well-typed SK-combinator
terms is present among the
binary trees of a given size,
indicating the possible existence 
of a lower bound that might be
easier to determine analytically
than in the case of general lambda terms.

\subsubsection{Generating typed SK-combinator trees by types}

In \cite{palka11} generation of random terms is guided by their types,
resulting in more realistic (while not uniformly random) terms,
used successfully in discovering some GHC bugs.

\subsubsection{Generating SK-trees by increasing type sizes}

The predicate {\tt genByType} first generates simple types
 with {\tt genType}
and then uses the unification-based querying mechanism
to generate, for each of the types,
its inhabitant SK-trees with fewer internal nodes then 
their their type.
\begin{code}
genByTypeSK(L,X,T):-
  genType(L,T),
  genSKs(L,X),
  simpleTypeOf(X,T).
\end{code}
The number of such terms grows quite fast, the sequence describing the number of terms with sizes smaller or equal than
the size of their types up to 7 is {\tt 0, 3, 29, 250, 3381, 48968, 809092}.

\begin{codeh}
tsizeAll(V,R):-var(V),!,R=0.
tsizeAll(A,0):-atomic(A),!.
tsizeAll((X*Y),R):-tsizeAll(X,R1),tsizeAll(Y,R2),R is 1+R1+R2. 
tsizeAll((X>Y),R):-tsizeAll(X,R1),tsizeAll(Y,R2),R is 1+R1+R2. 

queryByTypeSK(L,X,T):-queryByType(X,T,L,0),simpleTypeOf(X,T).

queryByTypeSKs(L,X,T):-queryByType(X,T,L,_),simpleTypeOf(X,T).

queryByType(k,(A>(_B>A)))-->[]. 
queryByType(s,(((A>(B>C))> ((A>B)>(A>C)))))-->[].
queryByType((A*B),Y)-->
  down,
  queryByType(A,T),
  queryByType(B,X),
  {unify_with_occurs_check(T,(X>Y))}.

xgenByTypeSK(L,X,T):-
  genType(L,T),
  queryByTypeSKs(L,X,T).  
\end{codeh}
\BX
Enumeration of closed simply-typed de SK combinator
trees
with types of size 2 and less then 2 internal nodes.
\begin{codex}
?- genByTypeSK(2,B,T).
B = k,
T = (x> (x>x)) ;
B = k*k*k,
T = (x> (x>x)) ;
B = k*k*s,
T = (x> (x>x)) .
\end{codex}
\EX

\begin{codeh}
countByTypeSK(M,Rs):-
  findall(R,
    (
       between(1,M,L),sols(genByTypeSK(L,_B,_T),R)
    ),
  Rs).
\end{codeh}

\subsection{The well-typed frontier of an untypable SK-expression}\label{wtf}
As in the case of 
lambda terms,
untypable SK-expressions 
become the majority as soon as the size of the expression
reaches some threshold, 9 in this case. 
This actually turns out to be a good thing,
from a programmer's perspective: types help with bug-avoidance
partly because being ``accidentally well-typed''
becomes a low probability event for larger programs.

Driven by a curiosity somewhat similar to that about
distribution and density properties of prime
numbers, one would want to
decompose an untypable SK-expression
into a set of maximal typable ones.
This makes sense, as, contrary
to lambda expressions, SK-trees
are uniquely built with application
operations as their internal nodes.
\begin{df}
We call {\em well-typed frontier} of a combinator tree 
set of its maximal well-typed subtrees.
\end{df}
Note also, that contrary to general lambda terms, 
SK-terms are {\em hereditarily closed} i.e.,
every subterm of a SK-expression is closed. 
Consequently, the well-typed
frontier is made of closed terms.
\begin{df}
We call {\em typeless trunk} of a combinator tree the subtree
starting from the root from which the members of
its well-typed frontier have been removed and replaced
with logic variables.
\end{df}

\subsubsection{Computing the well-typed frontier}

The 
well-typed frontier
of a combinator tree and its typeless trunk are
computed together by the predicate The predicate {\tt wellTypedFrontier} .
It actually proceeds by separating the trunk from
the frontier and marking with fresh logic variables
the replaced subtrees. These variables are  added as
left sides of equations with the frontiers as
their right sides.
\begin{code}
wellTypedFrontier(Term,Trunk,FrontierEqs):-
  wtf(Term, Trunk,FrontierEqs,[]).

wtf(Term,X)-->{typableSK(Term)},!,[X=Term].
wtf(A*B,X*Y)-->wtf(A,X),wtf(B,Y).
\end{code}

\BX
{\em Well-typed frontier} and {\em typeless trunk} of the untypable
term $SSI (SSI)$ (with $I$  represented as $SKK$). 
\begin{codex}
?- wellTypedFrontier(s*s*(s*k*k)*(s*s*(s*k*k)),
                     Trunk,FrontierEqs).
Trunk = A*B* (C*D),
FrontierEqs = [A=s*s, B=s*k*k, C=s*s, D=s*k*k].
\end{codex}
\EX

The list-of-equations representation of the frontier
allows to easily reverse their separation from the trunk by 
a unification based ``grafting'' operation.

The predicate {\tt fuseFrontier} implements this reversing process
while the predicate {\tt extractFrontier } extracts from the frontier-equations
the components of the frontier without the corresponding variables 
marking their location in the trunk.
\begin{code}
fuseFrontier(FrontierEqs):-maplist(call,FrontierEqs).

extractFrontier(FrontierEqs,Frontier):-
  maplist(arg(2),FrontierEqs,Frontier).
\end{code}

\BX
Extracting and grafting back the well-typed frontier to the typeless trunk.
\begin{codex}
?- wellTypedFrontier(s*s*(s*k*k)*(s*s*(s*k*k)),
       Trunk,FrontierEqs),
   extractFrontier(FrontierEqs,Frontier),
   fuseFrontier(FrontierEqs).
Trunk = s*s* (s*k*k)* (s*s* (s*k*k)),
FrontierEqs = [s*s=s*s, s*k*k=s*k*k, 
               s*s=s*s, s*k*k=s*k*k],
Frontier = [s*s, s*k*k, s*s, s*k*k] .   
\end{codex}
Note that after grafting back the frontier, the trunk becomes equal to the term
that we have started with.
\EX

\subsubsection{A comparison of the sizes of the well-typed frontier and the typeless trunk}

An interesting question arises at this point: 
{\em how do the sizes of the frontier and the trunk compare}?

Figure \ref{front} compares the average sizes of the frontier and the trunk for
terms up to size 8. This indicates that, while the size of the frontier dominates for 
small terms, it decreases progressively. This leaves the following open problem: 
{\em does the average ratio of 
the frontier and the trunk converge to a limit as the size of the terms
increases}? 
More empirical information on this can be obtained by studying what happens for
randomly generated large SK-trees.

\begin{figure*}
\begin{center}
\begin{tabular}{||r||r|r|r|r|}
\cline{1-5}
\cline{1-5}
\multicolumn{1}{||c||}{{Term size}} & 
  \multicolumn{1}{c|}{{ Avg. Trunk-size}} &
  \multicolumn{1}{c|}{{ Avg. Frontier-size}} &
   \multicolumn{1}{c|}{{ \% Trunk }} &
  \multicolumn{1}{c|}{{ \% Frontier}}\\
\cline{1-5} 
\cline{1-5}
1&0&1&0&100 \\ \cline{1-5}
2&0.13&1.88&6.25&93.75 \\ \cline{1-5}
3&0.26&2.74&8.75&91.25 \\ \cline{1-5}
4&0.47&3.53&11.77&88.23 \\ \cline{1-5}
5&0.71&4.29&14.11&85.89 \\ \cline{1-5}
6&0.97&5.03&16.24&83.76 \\ \cline{1-5}
7&1.27&5.73&18.11&81.89 \\ \cline{1-5}
8&1.58&6.42&19.76&80.24 \\ \cline{1-5}
\end{tabular} \\
\medskip
\caption{Comparison of sizes of the typeless trunk and the well-typed frontier of SK-terms, by size.
\label{front}}
\end{center}
\end{figure*}

\subsubsection{Simplification as normalization of the well-typed frontier}

Given that well-typed terms are strongly normalizing, we can simplify
an untypable term by normalizing the members of its frontier, for which
we are sure that {\tt evalSK} terminates. Once evaluated, we can 
graft back the results to the typeless trunk,
as implemented the predicate {\tt simplifySK}.   
\begin{code}
simplifySK(Term,Trunk):-
  wellTypedFrontier(Term,Trunk,FrontierEqs),
  extractFrontier(FrontierEqs,Fs),
  maplist(evalSK,Fs,NormalizedFs),
  maplist(arg(1),FrontierEqs,Vs),
  Vs=NormalizedFs.
\end{code}

The following question arises at this point: {\em are there terms that are not normalizable
that can be simplified by extracting and simplifying their well-typed frontier and then
grafting it back}?
Combinatorial search, using the {\tt genSK} predicate finds them starting at size {\tt 8}.
\BX
Simplifying some untypable terms for which normalization is non-terminating. 
\begin{codex}
?- Term= s*s*s* (s*s)*s* (k*s*k),
         simplifySK(Term,Trunk).              
Term = s*s*s* (s*s)*s* (k*s*k),
Trunk = s*s*s* (s*s)*s*s.

?- Term= k* (s*s*s* (s*s)*s* (k*s*k)),
   simplifySK(Term,Trunk).
Term = k* (s*s*s* (s*s)*s* (k*s*k)),
Trunk = k* (s*s*s* (s*s)*s*s).
\end{codex}
Note that, as expected, while simplification does not bring termination
to the normalization predicate {\tt evalSK/2}, it shows the existence
of non-terminating computations for which a terminating
simplification is possible.
\EX

\subsubsection{Discussion}\label{disc}

While the well-typed (and closed) frontier does not make sense for 
general lambda terms where closed terms may have open subterms,
it makes sense for other combinator or supercombinator languages \cite{peyton87},
some with practical uses in the compilation of functional languages.

Among the open problems we leave for future research, is to find out
if concepts like the
well-typed frontier of a richer combinator-language
can be used for suggesting a fix to a program in a typed
functional programming language, or to produce
more precise error messages in case of type errors.
For instance, it would be interesting to know
if a minimal well-typed alternative can be be inferred 
and suggested to the programmer on a type error.

If one replaces the {\tt unify\_with\_occurs\_check} in
predicate {\tt skTypeOf} with the cyclic term unification
(that most modern Prologs use by default), one
can observe that every combinator expression passes
the test! The predicate {\tt uselessTypeOf} implements this
variation.
\begin{code}
uselessTypeOf(k,(A>(_B>A))).
uselessTypeOf(s,(((A>(B>C))> ((A>B)>(A>C))))).
uselessTypeOf((A*B),Y):-
  uselessTypeOf(A,(X>Y)),
  uselessTypeOf(B,X).
\end{code}
After defining the predicates {\tt notReallyTypable}
and {\tt sameAsAny}
\begin{code}
notReallyTypable(X):-uselessTypeOf(X,_).

sameAsAny(L,M):-genSK(L,M),notReallyTypable(M).
\end{code}
one can notice the identical behavior of
{\tt sameAsAny} and {\tt genSK}, meaning
that failing the occurs-check is the
exclusive reason of failure to infer a type.
This happens in the presence of a unique basic
type ``{\tt x}''.
However, in the case of a more
realistic type system with multiple basic 
types like {\tt Boolean, Int, String} etc.,
the failure of type inference 
could also be a consequence of mismatched
basic types. Knowing more about these two
reasons for failure might suggest
weakened type systems where some limited
form of circularity is acceptable, provided
that no basic type mismatches occur.
While strong normalization
would be sacrificed if such circular types
were accepted, one might note that this is already the case
in practical languages, where fixpoint operators
or recursive data type definitions are allowed.

\subsection{Estimating the proportion of well-typed X-combinator trees}\label{xdens}.
An interesting question arises at this point:  what proportion of X-combinator trees of
a given size are well-typed?
While the analytic study of the {\em asymptotic density} 
has been successfully performed on several families of lambda terms \cite{bodini11,BCI13,grygielGen}, it is considered an open 
problem for well-typed terms. We will limit ourselves here to empirically estimate it,
as it is done in \cite{grygielGen} for general lambda terms, where experiments indicate
extreme sparsity for very large terms.

We can use our  generator {\tt genTree} to enumerate X-combinator trees
among which we can then count the number of
well-typed ones.
\BX
Types inferred for terms with 2 internal nodes.
\begin{codex}
?- genTree(2,X),xtype(X,T).
X = (x> (x>x)),
T = ((x> (x>x))> ((x>x)> (x>x))) ;
X = ((x>x)>x),
T = (x> (x>x))
\end{codex}
\EX

Figure \ref{ratioX} shows the counts
for well-typed X-combinator expressions
among the total binary trees
of given size. Note that the total
column is given by the Catalan
numbers (entry A000108 in \cite{intseq}),
as binary trees are a member of the Catalan
family of combinatorial objects \cite{StanleyEC}.
\begin{figure}
\begin{center}
\begin{tabular}{||r|r|r|r||}
\cline{1-4}
\cline{1-4}
\multicolumn{1}{|c|}{{Term size}} & 
  \multicolumn{1}{c|}{{ Well-typed}} &
  \multicolumn{1}{c|}{{ Total}} &
  \multicolumn{1}{c|}{{ Ratio}}\\
\cline{1-4} 
\cline{1-4}
0&1&1&1 \\ \cline{1-4}
1&1&1&1 \\ \cline{1-4}
2&2&2&1 \\ \cline{1-4}
3&5&5&1 \\ \cline{1-4}
4&12&14&0.8571 \\ \cline{1-4}
5&38&42&0.9047 \\ \cline{1-4}
6&113&132&0.8560 \\ \cline{1-4}
7&357&429&0.8321 \\ \cline{1-4}
8&1148&1430&0.8027 \\ \cline{1-4}
9&3794&4862&0.7803 \\ \cline{1-4}
10&12706&16796&0.7564 \\ \cline{1-4}
11&43074&58786&0.7327 \\ \cline{1-4}
12&147697&208012&0.7100\\ \cline{1-4}
\end{tabular} \\
\medskip
\caption{Proportion of well-typed X-combinator terms
\label{ratioX}}
\end{center}
\end{figure}

Somewhat surprisingly, a large
proportion of well-typed X-combinator
terms is present among the
binary trees of a given size,
indicating the possible existence 
of a lower bound that might be
easier to determine analytically
than in the case of general lambda terms.

\subsection{Querying the generator for specific types} \label{tquery}

Coming with Prolog's unification and non-deterministic search,
is the ability to make more specific queries
by providing a type pattern, that selects only terms
of a given type.
\BX
Terms of type {\tt x>x} of size 4.
\begin{codex}
?- genTypedB(4,Term,(x>x)).
Term = a(l(l(v(0))), l(v(0))) ;
Term = l(a(l(v(1)), l(v(0)))) ;
Term = l(a(l(v(1)), l(v(1)))) .

?- genTypedBs(12,T,(x>x)>x). 
false.
\end{codex}
\EX
Note that the last query, taking about a minute,
shows that no  closed terms
of type {\tt (x>x)>x}
exist up to size 12.

We can make use of our generator's  efficient
specialization to a given type to explore
 empirical estimates for some interesting function
types.

Contrary to the total absence of type {\tt (x>x)>x} among
terms of size up to 12, ``binary operations'' of type
{\tt x>(x>x)} turn 
out to be quite frequent, giving, by increasing sizes,
the sequence                    
[0, 2, 0, 14, 12, 201, 445, 4632, 17789, 158271, 891635].

{\em Transformers} of type
{\tt x>x}, by increasing sizes, give
the  sequence [1, 0, 3, 3, 31, 78, 596, 2500, 18474, 110265, 888676].
While type {\tt (x>x)>x} turns our to be absent up to size 12,
the type {\tt (x>x)>(x>x)} describing {\em transformers of transformers}
turns out to be quite popular, as shown by the sequence
[1,1, 4, 11, 55, 227, 1315, 7066, 46731, 309499, 2358951].
The same turns out to be tree also for
{\tt (x>x)>((x>x)>(x>x))}, giving
[0, 2, 1, 16, 29, 272, 940, 7594, 39075, 312797, 2115374]
and {\tt ((x>x)>(x>x)) > ((x>x)>(x>x))} giving
[1, 1, 5, 13, 73, 300, 1846, 10130, 69336, 469217, 3640134].
One might speculate that homotopy type theory \cite{htypes},
that focuses on such transformations 
and transformations of transformations etc.
has a rich population of lambda terms from which
to chose interesting inhabitants of such types!

\subsection{Iterated types}\label{itert}
\BX
As an interesting coincidence,
one might note
that the binary tree 
representation of the type 
of the K combinator is
nothing but the S combinator 
itself.
\begin{codex}
?- kT(K),xtype(K,T),sT(S).
K = ((x>x)>x),
T = S, S = (x> (x>x)).
\end{codex}
\EX

Given that X-combinator expressions
and their inferred simple types
are both represented as binary trees
of often comparable sizes,
one might be curious about what
happens if we iterate
this process. 

By interpreting
a type as its identically represented
X-combinator expression, one
can ask the question: is the type
expression itself well-typed?
If so, is the set of distinct
 iterated types
starting from an X-combinator
finite?

The predicate {\tt iterType}
applies the type inference
operation at most K-times,
until an untypable
term or a fixpoint is reached.
\begin{code}  
iterType(K,X, Ts, Steps):-
  iterType(K,FinalK,X,[],Rs),
  reverse(Rs,Ts),
  Steps is K-FinalK.

iterType(K,FinalK,X,Xs,Ys):-K>0,K1 is K-1,
  xtype(X,T),
  \+(member(T,Xs)),
  !,
  iterType(K1,FinalK,T,[T|Xs],Ys).
iterType(FinalK,FinalK,_,Xs,Xs).
\end{code}

\begin{codeh}

itersFrom(T,Size,Steps,Avg):-
  tsizeAll(T,Size),
  iterType(100,T,Ts0,Steps),
  Ts=[T|Ts0],
  maplist(tsizeAll,Ts,Sizes),
  avg(Sizes,Avg).

avg(Ns,Avg):-
  length(Ns,Len),
  sumlist(Ns,Sum),
  Avg is Sum/Len.

iterLens(M):-
  between(0,M,I),
  findall([Steps,AvgSize],(genTree(I,T),itersFrom(T,_Size,Steps,AvgSize)),Rs),
  length(Rs,Len),
  foldl(msum,Rs,[0,0],Sums),
  maplist(divWith(Len),Sums,Avgs),
  write(avgs=[I|Avgs]),nl,
  fail.
  
msum([A,B],[D,E],[X,Y]):-
  X is A+D,Y is B+E.

msum([A,B,C],[D,E,F],[X,Y,Z]):-
  X is A+D,Y is B+E,Z is C+F.

divWith(X,Y,R):-R is Y / X.
    
\end{codeh}

\begin{codeh}
iterTo(M):-
  between(0,M,I),
  t(I,T),tsizeAll(T,Size),
  iterType(100,T,Ts,Steps),
  maplist(tsizeAll,Ts,Sizes),
  length(Sizes,Len),
  sumlist(Sizes,Sum),
  Avg is Sum / Len,
  write((I,Steps,Size,Avg)),
  nl,
  fail.
\end{codeh}

 
\BX
Iterated types for K and S and I=SKK combinators.
\begin{codex}
?- kT(K),iterType(100,K,Ts,Steps).    
K = ((x>x)>x),
Ts = [x> (x>x), (x> (x>x))> ((x>x)> (x>x)), (x>x)> (x>x)],
Steps = 3.

?- sT(S),iterType(100,S,Ts,Steps).
S = (x>(x>x)),
Ts = [(x> (x>x))>((x>x)> (x>x)),(x>x)>(x>x),x> (x>x)],
Steps = 3.

?- skkT(XX),iterType(100,XX,Ts,Steps).
XX = (((x> (x>x))> ((x>x)>x))> ((x>x)>x)),
Ts = [x>x, x> (x> (x>x)), x> (x>x), 
      (x> (x>x))> ((x>x)> (x>x)), (x>x)> (x>x)],
Steps = 5.
\end{codex}
\EX

Figure \ref{iter} shows the average number of
steps until a un-typable term is found or a
fixpoint is reached as well as the average size of
the terms in the sequence of iterated types.

\begin{figure}
\begin{center}
\begin{tabular}{||r|r|r||}
\cline{1-3} 
\cline{1-3}
\multicolumn{1}{||c||}{{Initial term size}} & 
  \multicolumn{1}{c|}{{ Average steps}} &
  \multicolumn{1}{c|}{{ Average size}}\\
\cline{1-3}
\cline{1-3}
0&1&7\\ \cline{1-3}
1&4&3\\ \cline{1-3}
2&3&3.25\\ \cline{1-3}
3&2.4&7.2799\\ \cline{1-3}
4&2.5714&4.9476\\ \cline{1-3}
5&2.8333&5.5087\\ \cline{1-3}
6&2.5075&6.1571\\ \cline{1-3}
7&2.4405&6.6171\\ \cline{1-3}
8&2.3832&7.0235\\ \cline{1-3}
9&2.3290&7.4627\\ \cline{1-3}
10&2.2547&7.9913\\ \cline{1-3}
11&2.1831&8.5392\\ \cline{1-3}
12&2.1174&9.1143\\ \cline{1-3}
\end{tabular} \\
\medskip
\caption{Average steps and term sizes of iterated types\label{iter}}
\end{center}
\end{figure}

This matches the intuition that types are (smaller) approximations
of programs and suggests that the following holds.
\paragraph{Conjecture.} The set of iterated types is finite
for any X-combinator tree.

\subsection{Self-typed terms}
As X-combinator trees and their types share the same representation,
it makes sense to generate and count terms that
are equal to their types. The predicate {\tt genSelfTypedT}
generates such ``self-typed'' terms.
\begin{code}
genSelfTypedT(L,T):-genTree(L,T),xtype(T,T).
\end{code}

\begin{codeh}
countSelfTypedT(M,Rs):-
  findall(R,
    (
       between(1,M,L),sols(genSelfTypedT(L,_T),R)
    ),
  Rs).  
\end{codeh}

\BX
Self-typed X-combinator trees of size 6.
\begin{codex}
?- genSelfTypedT(6,T).         
T = (x> ((x>x)> ((x>x)> (x>x)))) ;
T = (x> (((x> (x>x))> (x>x))>x)) ;
T = ((x>x)> ((x> (x>x))> (x>x))) ;
T = ((x>x)> (((x>x)>x)> (x>x))).
\end{codex}
\EX
The sequence [0, 0, 0, 1, 2, 4, 14, 34, 101, 315, 1017, 3325, 11042] 
counts the number of self-typed terms
by increasing sizes, up to size 13.

\subsection{Two size-inflating injective functions from terms to terms} \label{infla}
By composing transformations
of X-combinator trees to
their equivalent lambda expressions
two interesting (but injective only)
mappings can be defined from
 X-combinator trees to a subset of them ({\tt t2t})
 and from lambda terms to a subset of them ({\tt b2b}).

\begin{code}
b2b --> rank,t2b.
t2t --> t2b,rank.
\end{code}

\BX
The injective mappings {\tt t2t} and {\tt b2b} can be used
to generate significantly larger X-combinator
trees and lambda expressions.
\begin{codex}
?- between(0,3,N),t(N,T),t2t(T,NewT),tsize(T,S1),
   tsize(NewT,S2),write(S1<S2),write(' '),fail;nl.
0<27 1<57 2<86 2<86

?- skkB(B),dbTermSize(B,S1),b2b(B,BB),dbTermSize(BB,S2),
           write(S1<S2),nl,fail.
12<374
\end{codex}
\EX
It is interesting at this point to see what happens to our building block -- the X-combinator --
when going through some of these transformations.

\BX
Transformations of the X-combinator via {\tt b2b}, {\tt evalDeBruijn}, {\tt boundTypeOf} and {\tt n}.
\begin{codex}
?- xB(X),b2b(X,XX),evalDeBruijn(XX,R),boundTypeOf(R,T),n(T,N).
X = l(a(a(a(...(l(v(1))))),
XX = a(l(a(a(a....l(l(v(1))))))))),
R = l(l(l(l(a(a(a(v(3), v(2)), v(0)), a(v(1), v(0))))))),
T = ((x> (x> (x>x)))> (x> ((x>x)> (x>x)))) .
N = 576
\end{codex}
\EX
While {\tt b2b} significantly inflates the de Bruijn term
corresponding to the X-combinator,  normalization reduces
it to a small, well-typed term.
This suggests the use of our shared representation
for experiments with dynamic systems or genetic programming
where  applications of arithmetic, type inference and
normalization operations are likely to create
interesting trajectories of evolution.

\subsection{Evolution of a multi-operation dynamic system}\label{multidyn}
Normalization, as the lambda calculus is is Turing-complete,
is subject
to non-termination. However, simply-typed terms are
strongly normalizing so it makes sense to play with
combinations of arithmetic operations, type inference
operations and normalization involving X-term
combinator trees as well as their lambda term
equivalents. 

For instance, the predicate {\tt evalOrNextB}
ensures that evaluation only proceeds
on lambda terms for which we are sure it terminates
with a new term
and applies the successor predicate ``{\tt s}''
otherwise, borrowed via the {\tt rank} and 
{\tt unrank} operations.
\begin{code}
evalOrNextB(B,EvB):-boundTypeOf(B,_),evalDeBruijn(B,EvB),EvB\==B,!.
evalOrNextB(B,NextB):-
  rank(B,T),
  s(T,NextT),
  unrank(NextT,NextB).
\end{code}
We can observe the orbits of these dynamic systems \cite{dynSys} starting from
a given lambda term in de Bruijn notation, for a given number of steps
with the predicate {\tt playWithB}.
\begin{code}
playWithB(Term,Steps,Orbit):-
  playWithB(Term,Steps,Orbit,[]).

playWithB(Term,Steps,[NewTerm|Ts1],Ts2):-Steps>0,!,
  Steps1 is Steps-1,
  evalOrNextB(Term,NewTerm),
  playWithB(NewTerm,Steps1,Ts1,Ts2).
playWithB(Term,_,[Term|Ts],Ts).
\end{code}
Note that  ranking these terms to usual bitstring-represented
integers would be intractable given their super-exponential
growth with depth. On the other hand, all the underlying operations
are linear time with ranking and unranking to natural numbers represented as
binary trees.
These terms are rather large,
but  by computing the sizes of the terms one can
have a good guess on their evolution.

Figure \ref{evalX} illustrates the evolution of this dynamic
system starting from the X-combinator's lambda equivalent by
plotting the tree sizes of the terms in its orbit. The plot indicates
that it is very likely that a repetitive pattern has developed.

\FIG{evalX}{Term sizes in the orbit starting from the X-combinator}{0.20}{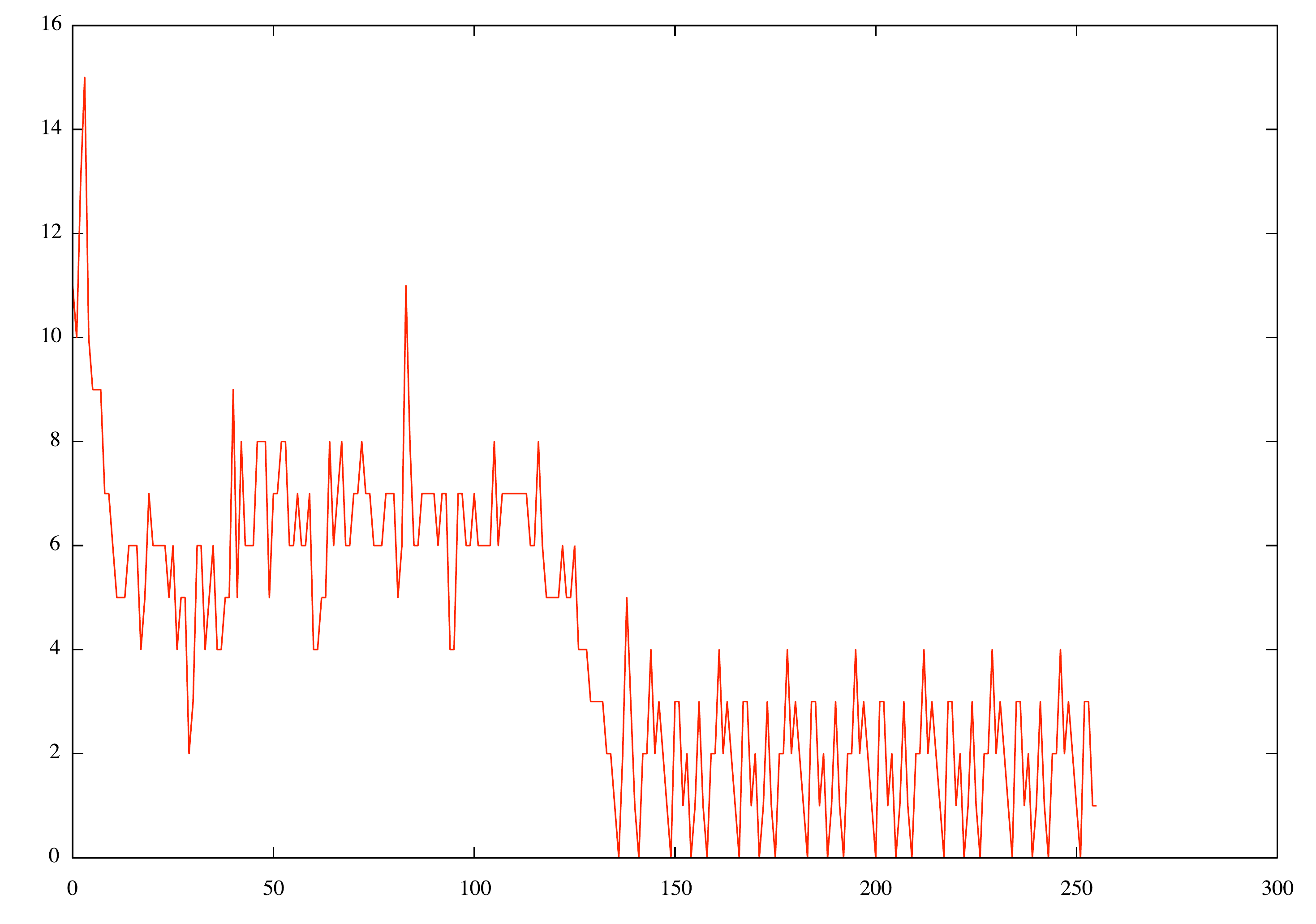}

Figure \ref{evalO} illustrates the evolution of this dynamic
system starting from the term $\omega=SII (SII)$ by
plotting the tree sizes of the terms in its orbit. The plot indicates
that it is very unlikely that a repetitive pattern will develop.
\FIG{evalO}{Term sizes in the orbit 
starting from the term $\omega$}{0.20}{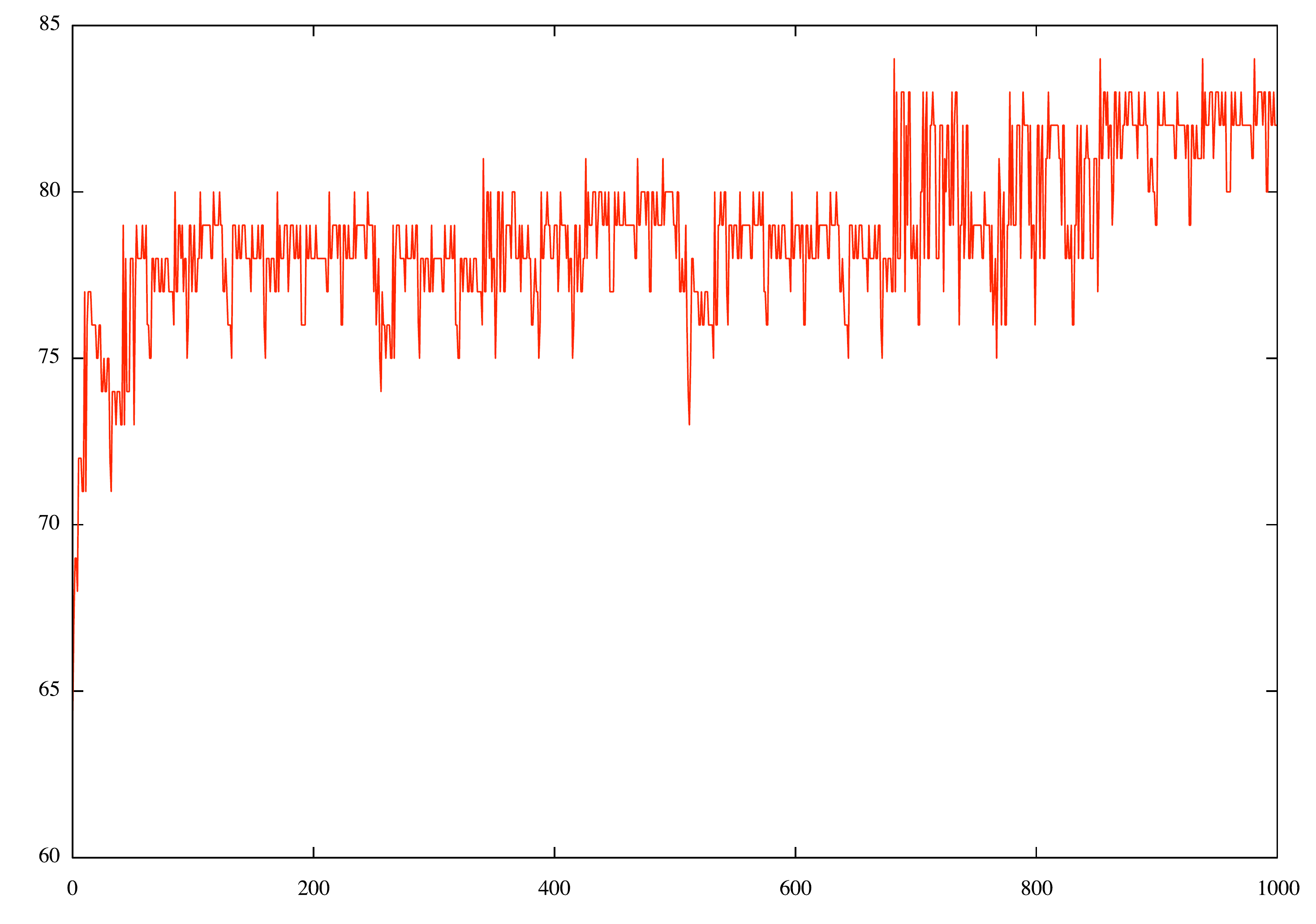}

Besides theoretical curiosity, one might use such operations
for implementing genetic programming algorithms.

\subsection{Memory savings through shared representations}

Given that the ranking and unranking operations
work in time proportional to the size of our lambda terms,
we will  explore some of the memory management
consequences of a shareable
representation of combinators, simple types, natural
numbers and lambda expressions.

We will look first at a well-known isomorphism
that brings us a significantly more
compact memory representation.

\subsubsection{A succinct representation of binary trees}\label{catpar}

Binary trees are in a well-known bijection with the 
language of of balanced parentheses, both 
being a member of the Catalan family of combinatorial
objects \cite{StanleyEC}.
The reversible predicate {\tt t2p/2} transforms
between binary trees and lists of balanced parentheses.
\begin{code} 
t2p(T,Ps):-t2p(T,0,1,Ps,[]).

t2p(X,L,R) --> [L],t2ps(X,L,R).

t2ps(x,_,R) --> [R].
t2ps((X>Xs),L,R) --> t2p(X,L,R),t2ps(Xs,L,R). 
\end{code}

\BX
The work of the reversible predicate {\tt t2p/2}.
\begin{codex}
?- skkT(X),t2p(X,Ps),t2p(NewX,Ps).
X = NewX, NewX = (((x> (x>x))> ((x>x)>x))> ((x>x)>x)),
Ps = [0,0,0,0,1,0,1,1,0,0,1,1,1,0,0,1,1,1].

?- kB(B),rank(B,T),t2p(T,Ps).
B = l(l(v(1))),
T = (x> (x> ((x>x)>x))),
Ps = [0,0,1,0,1,0,0,1,1,1] .
\end{codex}
\EX
Seen as a bitstring, the mapping to a list of
balanced parentheses is a 
succinct representation for our binary trees, 
if one wants to trade time
complexity for space complexity.
It is also a self-delimiting
prefix-free representation, uniquely decodable when
read from left to right. As one might notice,
it is actually is a bifix code,
i.e., it is also prefix-free when read
from right to left.

\subsubsection{A  practical shared memory representation}
In a practical implementation, given the
high frequency of small objects of any of our kinds -- numbers,
lambda expressions, types and combinators, one might consider
a hybrid representation where small
trees are represented within a machine
word as balanced 0,1-parentheses sequences and
larger ones as cons-cells. 2-bit-tagged pointers 
could be used to disambiguate interpretation as
numbers, combinators types or lambda expressions
but their targets could be shared if structurally
identical.
Besides sharing static data or code objects, a shared
representation is likely to also facilitate
memory management by recycling fragments
of computations like $\beta$-reductions or
arithmetic operations.

Graph-based representation of lambda terms has been used 
as early as \cite{lamping90}
to avoid redundant evaluation of redexes.
In a similar way, one could fold our tree-based
representations into DAGs, providing uniform savings
for combinators, types and tree-based natural numbers.

\begin{codeh}
ttsizes:-
 between(0,1000,N),t(N,T),t2t(T,TT),tsizeAll(T,S1),
 tsizeAll(TT,S2),write((N,S1,S2)),nl,fail.
 
nonetyped:-
  between(0,10000,N),t(N,T),t2t(T,TT),xtype(TT,_),
  write(N),nl,fail.
 
closedTo(M,I,B):-
  between(0,M,I),
  t(I,T),
  unrank(T,B),
  isClosedB(B).

typedTo(M,I,B,T):-
  between(0,M,I),
  t(I,T),
  unrank(T,B),
  isClosedB(B),
  boundTypeOf(B,T).
 
ttyped(From,To,I,TT):-between(From,To,I),t(I,T),t2b(T,B),boundTypeOf(B,TT). 

bclosed(From,To,I):-between(From,To,I),t(I,T),unrank(T,B),isClosedB(B).


ccount(M,R):-sols(closedTo(M,_,_),R).

tcount(M,R):-sols(typedTo(M,_,_,_),R). 

tcount1(M,R):-sols(ttyped(M,_I,_TT),R). 
  
    
\end{codeh}


\section{Related work} \label{rels}


The classic reference for lambda calculus is \cite{bar84}.
Various instances of typed lambda calculi are
overviewed in \cite{bar93}.

Originally introduced in \cite{dbruijn72}, the de Bruijn notation makes terms
equivalent up to $\alpha$-conversion and facilitates their normalization
\cite{kamaDB}. Their use in this paper is motivated by their comparative 
simplicity rather than by efficiency considerations, for which several abstract
machines, used in the implementation of functional languages,
have been designed \cite{peyton87}.
The compressed de Bruijn representation
of lambda terms proposed
in this paper (and \cite{cicm15}) is novel, to our best knowledge.

Lambda terms of bounded unary height are introduced in \cite{bodini11}.
John Tromp's binary lambda calculus is only described
through online code and the Wikipedia entry at
\cite{binlamb}.

Generators for closed and well-typed 
lambda terms, as well as their normal forms,
expressed as functional programming algorithms,
are given in \cite{grygielGen}, derived from
combinatorial recurrences. However, they
are significantly more complex than the ones
described here in Prolog. On the other hand,
we have not found in the literature generators
for linear, linear affine terms and lambda terms of bounded unary height.
Normalization of lambda terms and its confluence properties
are described in \cite{bar84} and \cite{kamaDB} with functional programming
algorithms given in \cite{sestoftLam} and HOAS-based 
evaluation first described in \cite{hoas}.

In a logic programming context, unification of simply
typed lambda terms has been used in as the foundation of
the programming language $\lambda$Prolog
\cite{miller91b,teyjus99} and applied to
higher order logic programming \cite{miller12}.

Various instances of typed lambda calculi are
overviewed in \cite{bar93}. Combinators originate
in Moses Sch\"onfinkel's 1924 paper, and  
independently, in Haskell Curry's work in 1927. A modern
introduction to combinators and their
relation to lambda calculus is \cite{hindley2008lambda} and
a first application of an extended set of 
combinators in the implementation
of functional programming languages is \cite{turner79}.

Combinatorics of lambda terms, including enumeration, random generation
and asymptotic behavior has seen an increased interest recently (see
for instance
\cite{bodini11,grygielGen,normalizing13,BCI13}),
partly motivated by applications to software testing, given the
widespread use of lambda terms as an intermediate language in compilers
for functional languages and proof assistants. 
Distribution and density properties
of random lambda terms are described in \cite{ranlamb09}.
In \cite{palka11,fetscher15}, types are used to generate random
terms for software testing. The same naturally ``goal-oriented''
effect is obtained in the generator/type inferrer
for de  Bruijn terms in subsection \ref{typedgen}, by taking advantage
of Prolog's ability to backtrack over possible terms,
while filtering against unification with a specific pattern.
In \cite{padl15}  generation algorithms for 
several sub-families of lambda terms are given
as well as a compressed deBruijn representation is
introduced.
In \cite{ppdp15tarau} Rosser's X-combinator
trees \cite{fokker92} are used as a uniform representation via
bijections top lambda terms in de Bruijn notation,
types and a tree-based number representation.

Of particular
interest are the results of \cite{grygielGen}
where recurrence relations and asymptotic behavior are studied
for several families of lambda terms.
Empirical evaluation of the density of closed
simply-typed general lambda terms
described in \cite{grygielGen} indicates extreme
sparsity for large sizes. However, the problem of
their exact asymptotic behavior is still open.
This has motivated our interest in the empirical
evaluation of the density of simply-typed
X-combinator trees, where we observed 
significantly higher initial densities
and where there's a chance that the also
open problem of their asymptotic
behavior might be easier to tackle.

One-point combinator bases, together with a derivation of
the X-combinator are described in \cite{fokker92}.
In \cite{OnePoint} 
the existence of a countable
number of 1-point bases is proven.
While esoteric programming languages exist 
based on similar 1-point bases \cite{iotaStay}, 
we have not seen
any such development centered around
Rosser's X-combinator, or type inference and normalization algorithms
designed specifically for it, as described in this paper.

Ranking and unranking algorithms for several classes of lambda
terms are also described in \cite{grygielGen},together with
a type inference algorithm for de Bruijn terms.
Ranking and unranking of lambda terms can be seen as a
building block for bijective serialization of practical data
types \cite{everybit} 
as well as for G\"odel-numbering schemes \cite{conf/icalp/HartmanisB74} of
theoretical relevance. In fact, ranking functions for sequences 
can be traced back to G\"{o}del numberings
\cite{Goedel:31} associated to formulas.


While G\"odel-numbering schemes for lambda terms
have been studied in several theoretical
papers on computability, we are not aware of
any size proportionate bijective encoding
as the one described in this paper.


Injective G\"odel-numbering schemes for lambda terms
in de Bruijn notation have been described in
the context of binary lambda calculus \cite{binarylambda}
and as a mechanism to encode datatypes in \cite{everybit,complambda}.
Both these use  prefix-free codes, ensuring unique decoding.
A bijective G\"odel-numbering scheme is associated
to the esoteric programming language Jot \cite{iotaStay},
where every bitstring is considered a valid
executable expression. This is similar  to
ours in the sense that every binary tree
representing an X-combinator expression is
executable.
However, the use of a binary tree based
model of Peano's axioms, playing the role
of the set of natural numbers, and
the corresponding ranking and unranking
algorithms as described in
this paper are novel.

The binary-tree based numbering system defined here is isomorphic to the
ones in \cite{lata14,sacs14tarau},  
where a similar treatment of arithmetic operations
is specialized to the language of balanced parentheses
and multiway trees.
In fact, such an encoding can be used as a prefix-free succinct
representation for our binary trees, if one wants to trade space
complexity for time complexity.
Any enumeration of combinatorial objects
(e.g., \cite{StanleyEC,knuth_trees})
can be seen as providing unary Peano arithmetic operations implicitly.
By contrast, the tree-based arithmetic operations used in this paper
have efficiency comparable to the usual binary numbers, as shown in \cite{arxiv_cats}.
Note also that while \cite{lata14,sacs14tarau} focus exclusively on
arithmetic operations with members of the Catalan family of
combinatorial objects, of which our binary trees are an instance,
their use in this paper, as a target for ranking/unranking of
lambda expressions, relies exclusively on the successor
and predecessor operations, adapted here to work on binary trees.

Univalent foundations of type theory \cite{htypes} 
have recently emphasized isomorphism paths between objects
as a means to unify equality and equivalence
between heterogenous data types sharing 
essential properties and behaviors under 
transformations. While informal, our
executable equivalences between combinators,
lambda terms, types and numbers might be useful
as practical illustrations of these  
concepts.

Some of the algorithms used in the paper,
like type inference and
normalization of combinators
and lambda terms,  
are common knowledge \cite{kamaDB,sestoftLam,bar84}, 
although we are
not aware, for instance,
of Prolog implementations of 
type inference
working directly on de Bruijn 
terms or X-combinator trees. In \cite{padl15} a type
inference algorithm for standard 
terms using Prolog's logic variables is given.
To make the paper self-contained,
we have closely followed
the normalization algorithm of \cite{padl15} 
using a de Bruijn representation of lambda terms.
We refer  to \cite{padl15} 
for a compressed de Bruijn representation
and several
Prolog algorithms that complement
our  playground with
generators for closed,
linear, linear affine, binary lambda terms as well
as lambda terms of bounded binary height.

\section{Conclusions} \label{concl}

We have described compact (and arguably elegant)
combinatorial generation algorithms for
several important families of lambda terms. 
Besides the
newly introduced a compressed
form of de Bruijn terms we have used ordinary
de Bruijn terms as well as a canonical representation
of lambda terms relying on Prolog's logic variables.
In each case, we have selected the representation
that was more appropriate for
tasks like combinatorial generation, type
inference or normalization.
We have switched representation as needed,
though bijective transformers working in
time proportional to the size of the terms.
Our combinatorial generation algorithms
match the corresponding sequence of
counts by size, given in \cite{intseq}
as an empirical validation of their 
correctness.

We have described Prolog-based term and type generation and
as well as type-inference algorithms for de Bruijn terms.
Among the possible applications
of our techniques we mention
compilation and test generation
for lambda-calculus based languages
and proof assistants.
Our merged generation and type inference in an algorithm
 showed a mechanism to build ``customized closed terms of a given type''. 
This ``relational view'' of terms and their types
has enabled the discovery of interesting
patterns about the type expressions 
occurring in well-typed programs. 
We have uncovered the most ``popular'' types that 
govern function applications among a about a million
small-sized lambda terms.

We have also observed some interesting phenomena
about frequently occurring types, that seem
to be similar to those in human-written programs
and we have computed  growth sequences for
the number of inhabitants of some ``popular'' types,
for which we have not found any study in the literature.

A significant contribution of this paper
is the size-proportionate ranking/unranking
algorithm for lambda terms and the compressed
de Bruijn representation that facilitated it.
The ability to encode lambda terms bijectively
can be used as a ``serialization'' mechanism in functional
programming languages and proof assistants 
using them as an intermediate language.

We have selected the minimalist pure combinator
language built from applications of combinators
$S$ and $K$ to explore
aspects of their generation and  type inference
algorithms. While a draconian simplification of
real-life programming languages, this well-known
and well-researched subset of lambda calculus
has revealed
some interesting new facts about the density
and distribution of their types. 
The new concepts of {\em well-typed frontier}
and {\em typeless trunk} of an untypable term
can be generalized to realistic
combinator and supercombinator-based
intermediate languages used by compilers
for functional languages and proof assistants.
As they give precise hints about the points where type
inference failed, they are likely to be
useful for debugging programs and give
more meaningful compile-time error messages.
This also results in an ability to extend (sure) termination
beyond simply-typed terms, by evaluating
and then grafting back their well-typed
frontier.
By sharing the representation of the Turing-complete language of
X-combinator expressions,
natural numbers, lambda terms and their
types, interesting synergies
became available.

The paper has  introduced a number of algorithms
that, at our best knowledge, are novel, at least 
in terms of their logic programming implementation,
among which  we mention the type inference for de Bruijn terms using 
unification with occurs-check in subsection \ref{dbtypes}
and the integrated generation and type inference algorithm 
for closed simply typed de Bruijn terms
in section \ref{typedgen}.
Besides the ability to efficiently query for 
inhabitants of specific types,
our algorithms also support a from of
``query-by-example'' mechanism, for finding (possibly
smaller) terms inhabiting the same type as the query term.
While the main focus of the paper is the
creation of a logic programming based
declarative playground for experiments
with various classes of lambda terms, under the assumption
of a shared representation,
the paper introduces several new concepts among which we mention:
\begin{itemize}
\item a compressed representation of 
de Bruijn terms in subsection \ref{comp}
\item X-combinator trees playing the role of both natural numbers
and types
in subsections \ref{xco}, \ref{tlam} and \ref{tn}
\item a bijection between natural numbers and binary trees 
(predicates {\tt t/2} and {\tt n/2} in
subsection \ref{tn}) that is works consistently
with their isomorphic arithmetic operations
\item a concept of ``iterated types'' in subsection \ref{itert}
\item two size-inflating injective functions from terms to terms in subsection \ref{infla}
\item a multi-operation dynamic system combining normalization and arithmetic operations in
subsection \ref{multidyn}
\end{itemize}
The paper also describes
algorithms
that, at our best knowledge, are novel, at least 
in terms of their logic programming implementation:
\begin{itemize}
\item integrated generation and type inference algorithm for closed simply-typed de Bruijn terms
in subsection \ref{typedgen}
\item successor and predecessor  and arithmetic operations on binary trees in subsection \ref{sucpred}
\item ranking and unranking de Bruijn terms to/from binary-tree represented natural numbers
in subsection \ref{binrank}
\item direct type inference for X-combinator trees in subsection \ref{dirinf}
\end{itemize}

While a non-strict functional language like
Haskell could have been used for deriving
similar algorithms, the synergy between
Prolog's non-determinism, DCG transformation and the
availability of unification with occurs-check
made the code embedded in the 
paper significantly simpler and arguably
clearer.  

Future work is planned along the following lines.
Enumeration or random generation of binary
trees can be extended to general lambda expressions and
various data types expressed in terms of them.
Functional languages like Scheme and Lisp, based on
{\tt cons} operations might be able to
improve memory footprint of symbolic and numerical
data through shared representations of arithmetic
operations and list or tree data structures.
Small steps in the normalization of combinator expressions
or lambda trees can be mapped to possibly interesting
number sequences.
Open problems related to the asymptotic density of typable
combinators and lambda terms might benefit from empirical
estimates computable within our framework 
for very large terms.
Future work will also focus on studying
how our results extend to other families
of combinators and supercombinators
that occur in practical languages as well
as on random SK-tree generation 
e.g.., by extending R\'emy's algorithm
 \cite{remy85}
from binary trees to SK-combinator trees.
This would allow fast generation of
very large SK-combinator expressions
that could give better empirical estimates
on the asymptotic behavior of the concepts
introduced in this paper and their properties.
Also, as a step toward more practical uses,
lifting the concept of well-typed
frontier to general lambda terms (which are not 
hereditarily closed) seems possible by
defining the frontier as being
a sequence of maximal well-typed closed lambda terms.

We hope that the techniques described in this paper,
taking advantage of this unique combination of 
strengths, recommend logic programming
as a convenient meta-language for the manipulation
of various families of lambda terms
and the study of their combinatorial
and computational properties.

\section*{Acknowledgement}
This research has been supported by NSF grant \verb~1423324~.

\bibliographystyle{acmtrans}
\bibliography{theory,tarau,proglang}

\begin{thebibliography}{}

\bibitem[\protect\citeauthoryear{Barendregt}{Barendregt}{1984}]{bar84}
{\sc Barendregt, H.~P.} 1984.
\newblock {\em The Lambda Calculus Its Syntax and Semantics\/}, Revised ed.
  Vol. 103.
\newblock North Holland.

\bibitem[\protect\citeauthoryear{Barendregt}{Barendregt}{1991}]{bar93}
{\sc Barendregt, H.~P.} 1991.
\newblock Lambda calculi with types.
\newblock In {\em Handbook of Logic in Computer Science}. Vol.~2. Oxford
  University Press.

\bibitem[\protect\citeauthoryear{Bodini, Gardy, and Gittenberger}{Bodini
  et~al\mbox{.}}{2011}]{bodini11}
{\sc Bodini, O.}, {\sc Gardy, D.}, {\sc and} {\sc Gittenberger, B.} 2011.
\newblock Lambda-terms of bounded unary height.
\newblock In {\em ANALCO}. SIAM, 23--32.

\bibitem[\protect\citeauthoryear{Cegielski and Richard}{Cegielski and
  Richard}{1999}]{ceg99}
{\sc Cegielski, P.} {\sc and} {\sc Richard, D.} 1999.
\newblock {On arithmetical first-order theories allowing encoding and decoding
  of lists}.
\newblock {\em Theoretical Computer Science\/}~{\em 222,\/}~1-2, 55--75.

\bibitem[\protect\citeauthoryear{David, Grygiel, Kozik, Raffalli, Theyssier,
  and Zaionc}{David et~al\mbox{.}}{2010}]{normalizing13}
{\sc David, R.}, {\sc Grygiel, K.}, {\sc Kozik, J.}, {\sc Raffalli, C.}, {\sc
  Theyssier, G.}, {\sc and} {\sc Zaionc, M.} 2010.
\newblock Asymptotically almost all $\lambda$-terms are strongly normalizing.
\newblock {\em Preprint: arXiv: math. LO/0903.5505 v3\/}.

\bibitem[\protect\citeauthoryear{David, Raffalli, Theyssier, Grygiel, Kozik,
  and Zaionc}{David et~al\mbox{.}}{2009}]{ranlamb09}
{\sc David, R.}, {\sc Raffalli, C.}, {\sc Theyssier, G.}, {\sc Grygiel, K.},
  {\sc Kozik, J.}, {\sc and} {\sc Zaionc, M.} 2009.
\newblock Some properties of random lambda terms.
\newblock {\em Logical Methods in Computer Science\/}~{\em 9,\/}~1.

\bibitem[\protect\citeauthoryear{de~Bruijn}{de~Bruijn}{1972}]{dbruijn72}
{\sc de~Bruijn, N.~G.} 1972.
\newblock {Lambda calculus notation with nameless dummies, a tool for automatic
  formula manipulation, with application to the Church-Rosser Theorem}.
\newblock {\em Indagationes Mathematicae\/}~{\em 34}, 381--392.

\bibitem[\protect\citeauthoryear{Fetscher, Claessen, Palka, Hughes, and
  Findler}{Fetscher et~al\mbox{.}}{2015}]{fetscher15}
{\sc Fetscher, B.}, {\sc Claessen, K.}, {\sc Palka, M.~H.}, {\sc Hughes, J.},
  {\sc and} {\sc Findler, R.~B.} 2015.
\newblock Making random judgments: Automatically generating well-typed terms
  from the definition of a type-system.
\newblock In {\em Programming Languages and Systems - 24th European Symposium
  on Programming, {ESOP} 2015, Held as Part of the European Joint Conferences
  on Theory and Practice of Software, {ETAPS} 2015, London, UK, April 11-18,
  2015. Proceedings}. 383--405.

\bibitem[\protect\citeauthoryear{Fokker}{Fokker}{1992}]{fokker92}
{\sc Fokker, J.} 1992.
\newblock The systematic construction of a one-combinator basis for
  lambda-terms.
\newblock {\em Formal Aspects of Computing\/}~{\em 4}, 776--780.

\bibitem[\protect\citeauthoryear{G\"{o}del}{G\"{o}del}{1931}]{Goedel:31}
{\sc G\"{o}del, K.} 1931.
\newblock \"{U}ber formal unentscheidbare {S\"{a}tze der Principia Mathematica
  und verwandter Systeme I}.
\newblock {\em Monatshefte f\"{u}r Mathematik und Physik\/}~{\em 38}, 173--198.

\bibitem[\protect\citeauthoryear{Goldberg}{Goldberg}{2004}]{OnePoint}
{\sc Goldberg, M.} 2004.
\newblock A construction of one-point bases in extended lambda calculi.
\newblock {\em Inf. Process. Lett.\/}~{\em 89,\/}~6, 281--286.

\bibitem[\protect\citeauthoryear{Grygiel, Idziak, and Zaionc}{Grygiel
  et~al\mbox{.}}{2013}]{BCI13}
{\sc Grygiel, K.}, {\sc Idziak, P.~M.}, {\sc and} {\sc Zaionc, M.} 2013.
\newblock How big is {BCI} fragment of {BCK} logic.
\newblock {\em J. Log. Comput.\/}~{\em 23,\/}~3, 673--691.

\bibitem[\protect\citeauthoryear{Grygiel and Lescanne}{Grygiel and
  Lescanne}{2013}]{grygielGen}
{\sc Grygiel, K.} {\sc and} {\sc Lescanne, P.} 2013.
\newblock Counting and generating lambda terms.
\newblock {\em J. Funct. Program.\/}~{\em 23,\/}~5, 594--628.

\bibitem[\protect\citeauthoryear{Hartmanis and Baker}{Hartmanis and
  Baker}{1974}]{conf/icalp/HartmanisB74}
{\sc Hartmanis, J.} {\sc and} {\sc Baker, T.~P.} 1974.
\newblock {On Simple Goedel Numberings and Translations}.
\newblock In {\em ICALP} (2002-02-01), {J.~Loeckx}, Ed. Lecture Notes in
  Computer Science, vol.~14. Springer, Berlin Heidelberg, 301--316.

\bibitem[\protect\citeauthoryear{Hindley and Seldin}{Hindley and
  Seldin}{2008}]{hindley2008lambda}
{\sc Hindley, J.~R.} {\sc and} {\sc Seldin, J.~P.} 2008.
\newblock {\em Lambda-calculus and combinators: an introduction}. Vol.~13.
\newblock Cambridge University Press Cambridge.

\bibitem[\protect\citeauthoryear{Kamareddine}{Kamareddine}{2001}]{kamaDB}
{\sc Kamareddine, F.} 2001.
\newblock {Reviewing the Classical and the de Bruijn Notation for λ‐calculus
  and Pure Type Systems}.
\newblock {\em Journal of Logic and Computation\/}~{\em 11,\/}~3, 363--394.

\bibitem[\protect\citeauthoryear{Katok and Hasselblatt}{Katok and
  Hasselblatt}{1995}]{dynSys}
{\sc Katok, A.} {\sc and} {\sc Hasselblatt, B.} 1995.
\newblock {\em Introduction to the modern theory of dynamical systems}. Ency.
  of Math. and its App., vol.~54.
\newblock Cambridge Univ. Press.

\bibitem[\protect\citeauthoryear{Knuth}{Knuth}{2005}]{knuth_comb}
{\sc Knuth, D.~E.} 2005.
\newblock {\em The Art of Computer Programming, Volume 4, Fascicle 3:
  Generating All Combinations and Partitions}.
\newblock Addison-Wesley Professional.

\bibitem[\protect\citeauthoryear{Knuth}{Knuth}{2006}]{knuth_trees}
{\sc Knuth, D.~E.} 2006.
\newblock {\em The Art of Computer Programming, Volume 4, Fascicle 4:
  Generating All Trees--History of Combinatorial Generation (Art of Computer
  Programming)}.
\newblock Addison-Wesley Professional.

\bibitem[\protect\citeauthoryear{Kobayashi, Matsuda, and Shinohara}{Kobayashi
  et~al\mbox{.}}{2012}]{complambda}
{\sc Kobayashi, N.}, {\sc Matsuda, K.}, {\sc and} {\sc Shinohara, A.} 2012.
\newblock {Functional Programs as Compressed Data}.
\newblock {\em ACM SIGPLAN 2012 Workshop on Partial Evaluation and Program
  Manipulation\/}.
\newblock ACM Press.

\bibitem[\protect\citeauthoryear{Kreher and Stinson}{Kreher and
  Stinson}{1999}]{combi99}
{\sc Kreher, D.~L.} {\sc and} {\sc Stinson, D.} 1999.
\newblock {\em Combinatorial Algorithms: Generation, Enumeration, and Search}.
\newblock The CRC Press Series on Discrete Mathematics and its Applications.
  CRC PressINC.

\bibitem[\protect\citeauthoryear{Lamping}{Lamping}{1990}]{lamping90}
{\sc Lamping, J.} 1990.
\newblock {An Algorithm for Optimal Lambda Calculus Reduction}.
\newblock In {\em Conference Record of the Seventeenth Annual {ACM} Symposium
  on Principles of Programming Languages, San Francisco, California, USA,
  January 1990}. 16--30.

\bibitem[\protect\citeauthoryear{Lehmer}{Lehmer}{1964}]{lehmer64}
{\sc Lehmer, D.~H.} 1964.
\newblock {The machine tools of combinatorics}.
\newblock In {\em Applied combinatorial mathematics}. Wiley, New York, 5--30.

\bibitem[\protect\citeauthoryear{McBride}{McBride}{2010}]{numhack}
{\sc McBride, C.} 2010.
\newblock {I am not a number, I am a classy hack}.
\newblock {\em Blog entry:}
  \verb~http://mazzo.li/epilogue/index.html%3Fp=773.html~.

\bibitem[\protect\citeauthoryear{Miller}{Miller}{1991}]{miller91b}
{\sc Miller, D.} 1991.
\newblock Unification of simply typed lambda-terms as logic programming.
\newblock In {\em Proc. Int. Conference on Logic Programming (Paris)}. MIT
  Press, 255--269.

\bibitem[\protect\citeauthoryear{Miller and Nadathur}{Miller and
  Nadathur}{2012}]{miller12}
{\sc Miller, D.} {\sc and} {\sc Nadathur, G.} 2012.
\newblock {\em Programming with Higher-Order Logic}.
\newblock Cambridge University Press, New York, NY, USA.

\bibitem[\protect\citeauthoryear{Nadathur and Mitchell}{Nadathur and
  Mitchell}{1999}]{teyjus99}
{\sc Nadathur, G.} {\sc and} {\sc Mitchell, D.} 1999.
\newblock {System Description: Teyjus A Compiler and Abstract Machine Based
  Implementation of $\lambda$Prolog}.
\newblock In {\em Automated Deduction — CADE-16}. Lecture Notes in Computer
  Science, vol. 1632. Springer Berlin Heidelberg, 287--291.

\bibitem[\protect\citeauthoryear{Palka, Claessen, Russo, and Hughes}{Palka
  et~al\mbox{.}}{2011}]{palka11}
{\sc Palka, M.~H.}, {\sc Claessen, K.}, {\sc Russo, A.}, {\sc and} {\sc Hughes,
  J.} 2011.
\newblock Testing an optimising compiler by generating random lambda terms.
\newblock In {\em Proceedings of the 6th International Workshop on Automation
  of Software Test}. AST'11. ACM, New York, NY, USA, 91--97.

\bibitem[\protect\citeauthoryear{Peyton~Jones}{Peyton~Jones}{1987}]{peyton87}
{\sc Peyton~Jones, S.~L.} 1987.
\newblock {\em The Implementation of Functional Programming Languages
  (Prentice-Hall International Series in Computer Science)}.
\newblock Prentice-Hall, Inc., NJ, USA.

\bibitem[\protect\citeauthoryear{Pfenning and Elliot}{Pfenning and
  Elliot}{1988}]{hoas}
{\sc Pfenning, F.} {\sc and} {\sc Elliot, C.} 1988.
\newblock Higher-order abstract syntax.
\newblock In {\em Proceedings of the ACM SIGPLAN 1988 Conference on Programming
  Language Design and Implementation}. PLDI '88. ACM, New York, NY, USA,
  199--208.

\bibitem[\protect\citeauthoryear{R\'emy}{R\'emy}{1985}]{remy85}
{\sc R\'emy, J.-L.} 1985.
\newblock Un proc\'ed\'e it\'eratif de d\'enombrement d'arbres binaires et son
  application \`a leur g\'en\'eration al\'eatoire.
\newblock {\em RAIRO - Theoretical Informatics and Applications - Informatique
  Th\'eorique et Applications\/}~{\em 19,\/}~2, 179--195.

\bibitem[\protect\citeauthoryear{Sestoft}{Sestoft}{2002}]{sestoftLam}
{\sc Sestoft, P.} 2002.
\newblock Demonstrating lambda calculus reduction.
\newblock In {\em The Essence of Computation}, {T.~A. Mogensen}, {D.~A.
  Schmidt}, {and} {I.~H. Sudborough}, Eds. Springer-Verlag New York, Inc., New
  York, NY, USA, 420--435.

\bibitem[\protect\citeauthoryear{Sloane}{Sloane}{2014}]{intseq}
{\sc Sloane, N. J.~A.} 2014.
\newblock {The On-Line Encyclopedia of Integer Sequences.}
\newblock ~Published electronically at https://oeis.org/.

\bibitem[\protect\citeauthoryear{Stanley}{Stanley}{1986}]{StanleyEC}
{\sc Stanley, R.~P.} 1986.
\newblock {\em Enumerative Combinatorics}.
\newblock Wadsworth Publ. Co., Belmont, CA, USA.

\bibitem[\protect\citeauthoryear{Stay}{Stay}{2005}]{iotaStay}
{\sc Stay, M.} 2005.
\newblock Very simple chaitin machines for concrete {AIT}.
\newblock {\em CoRR\/}~{\em abs/cs/0508056}.

\bibitem[\protect\citeauthoryear{Tarau}{Tarau}{2009}]{ppdp09pISO}
{\sc Tarau, P.} 2009.
\newblock {An Embedded Declarative Data Transformation Language}.
\newblock In {\em {Proceedings of 11th International ACM SIGPLAN Symposium PPDP
  2009}}. ACM, Coimbra, Portugal, 171--182.

\bibitem[\protect\citeauthoryear{Tarau}{Tarau}{2013}]{serpro}
{\sc Tarau, P.} 2013.
\newblock {Compact Serialization of Prolog Terms (with Catalan Skeletons,
  Cantor Tupling and G\"odel Numberings) }.
\newblock {\em Theory and Practice of Logic Programming\/}~{\em 13,\/}~4-5,
  847--861.

\bibitem[\protect\citeauthoryear{Tarau}{Tarau}{2014a}]{ppdp14tarau}
{\sc Tarau, P.} 2014a.
\newblock { Bijective Collection Encodings and Boolean Operations with
  Hereditarily Binary Natural Numbers}.
\newblock In {\em {PPDP '14: Proceedings of the 16th international ACM SIGPLAN
  Symposium on Principles and Practice of Declarative Programming}}. ACM, New
  York, NY, USA.

\bibitem[\protect\citeauthoryear{Tarau}{Tarau}{2014b}]{arxiv_cats}
{\sc Tarau, P.} 2014b.
\newblock {A Generic Numbering System based on Catalan Families of
  Combinatorial Objects}.
\newblock {\em CoRR\/}~{\em abs/1406.1796}.

\bibitem[\protect\citeauthoryear{Tarau}{Tarau}{2014c}]{sacs14tarau}
{\sc Tarau, P.} 2014c.
\newblock Arithmetic and boolean operations on recursively run-length
  compressed natural numbers.
\newblock {\em Scientific Annals of Computer Science\/}~{\em 24,\/}~2,
  287--323.

\bibitem[\protect\citeauthoryear{Tarau}{Tarau}{2014d}]{lata14}
{\sc Tarau, P.} 2014d.
\newblock {Computing with Catalan Families}.
\newblock In {\em {Proceedings of Language and Automata Theory and
  Applications, 8th International Conference, LATA 2014}}, {A.-H. Dediu},
  {C.~Martin-Vide}, {J.-L. Sierra}, {and} {B.~Truthe}, Eds. Springer, LNCS,
  Madrid, Spain,, 564--576.

\bibitem[\protect\citeauthoryear{Tarau}{Tarau}{2015a}]{ppdp15tarau}
{\sc Tarau, P.} 2015a.
\newblock { On a Uniform Representation of Combinators, Arithmetic, Lambda
  Terms and Types}.
\newblock In {\em {PPDP'15: Proceedings of the 17th international ACM SIGPLAN
  Symposium on Principles and Practice of Declarative Programming}},
  {E.~Albert}, Ed. ACM, New York, NY, USA, 244--255.

\bibitem[\protect\citeauthoryear{Tarau}{Tarau}{2015b}]{padl15}
{\sc Tarau, P.} 2015b.
\newblock {On Logic Programming Representations of Lambda Terms: de Bruijn
  Indices, Compression, Type Inference, Combinatorial Generation,
  Normalization}.
\newblock In {\em {Proceedings of the Seventeenth International Symposium on
  Practical Aspects of Declarative Languages PADL'15}}, {E.~Pontelli} {and}
  {T.~C. Son}, Eds. Springer, LNCS 8131, Portland, Oregon, USA, 115--131.

\bibitem[\protect\citeauthoryear{Tarau}{Tarau}{2015c}]{cicm15}
{\sc Tarau, P.} 2015c.
\newblock {Ranking/Unranking of Lambda Terms with Compressed de Bruijn
  Indices}.
\newblock In {\em {Proceedings of the 8th Conference on Intelligent Computer
  Mathematics}}, {M.~Kerber}, {J.~Carette}, {C.~Kaliszyk}, {F.~Rabe}, {and}
  {V.~Sorge}, Eds. Springer, LNAI 9150, Washington, D.C., USA, 118--133.

\bibitem[\protect\citeauthoryear{{The Univalent Foundations Program}}{{The
  Univalent Foundations Program}}{2013}]{htypes}
{\sc {The Univalent Foundations Program}}. 2013.
\newblock {\em {Homotopy Type Theory}}.
\newblock Institute of Advanced Studies, Princeton.
\newblock \url{http://homotopytypetheory.org/2013/06/20/the-hott-book/}.

\bibitem[\protect\citeauthoryear{Tromp}{Tromp}{2014}]{binarylambda}
{\sc Tromp, J.} 2014.
\newblock Binary lambda calculus and combinatory logic.

\bibitem[\protect\citeauthoryear{Turner}{Turner}{1979}]{turner79}
{\sc Turner, D.~A.} 1979.
\newblock A new implementation technique for applicative languages.
\newblock {\em Software: Practice and Experience\/}~{\em 9,\/}~1, 31--49.

\bibitem[\protect\citeauthoryear{Vytiniotis and Kennedy}{Vytiniotis and
  Kennedy}{2010}]{everybit}
{\sc Vytiniotis, D.} {\sc and} {\sc Kennedy, A.} 2010.
\newblock {Functional Pearl: Every Bit Counts}.
\newblock {\em ICFP 2010 : The 15th ACM SIGPLAN International Conference on
  Functional Programming\/}.
\newblock ACM Press.

\bibitem[\protect\citeauthoryear{Wikipedia}{Wikipedia}{2015}]{binlamb}
{\sc Wikipedia}. 2015.
\newblock Binary lambda calculus --- wikipedia{,} the free encyclopedia.
\newblock [Online; accessed 20-February-2015].

\end{thebibliography}

\section*{Appendix}

\subsection*{Helper predicates for ranking and unranking balanced parentheses expressions}

The predicate {\tt binDif} computes the difference of two binomials.
\begin{code}
binDif(N,X,Y,R):- N1 is 2*N-X,R1 is N - (X + Y) // 2, R2 is R1-1,
  binomial(N1,R1,B1),binomial(N1,R2,B2),R is B1-B2.  
\end{code}
The predicate {\tt localRank} computes, by binary search
the rank of sequences of a given length.
\begin{code}
localRank(N,As,NewLo):- X is 1, Y is 0, Lo is 0,
  binDif(N,0,0,Hi0),Hi is Hi0-1,
  localRankLoop(As,N,X,Y,Lo,Hi,NewLo,_NewHi).
\end{code}
After finding the appropriate range containing the rank with {\tt binDif},
we delegate the work to the   predicate
{\tt localRankLoop}.
\begin{code}  
localRankLoop(As,N,X,Y,Lo,Hi,FinalLo,FinalHi):-N2 is 2*N,X< N2,!,
  PY is Y-1, SY is Y+1, nth0(X,As,A),
  (0=:=A-> binDif(N,X,PY,Hi1),
     NewHi is Hi-Hi1, NewLo is Lo, NewY is SY
   ; binDif(N,X,SY,Lo1),
     NewLo is Lo+Lo1, NewHi is Hi, NewY is PY
  ), NewX is X+1,
  localRankLoop(As,N,NewX,NewY,NewLo,NewHi,FinalLo,FinalHi).
localRankLoop(_As,_N,_X,_Y,Lo,Hi,Lo,Hi).  
\end{code}

\begin{code}
rankLoop(I,S,FinalS):-I>=0,!,cat(I,C),NewS is S+C, PI is I-1,
  rankLoop(PI,NewS,FinalS).
rankLoop(_,S,S).
\end{code}
Unranking works in a similar way. The predicate {\tt localUnrank}
builds a sequence of balanced parentheses by doing binary search
to locate the sequence in the enumeration of sequences of 
a given length.
\begin{code}
localUnrank(N,R,As):-Y is 0,Lo is 0,binDif(N,0,0,Hi0),Hi is Hi0-1, X is 1,
  localUnrankLoop(X,Y,N,R,Lo,Hi,As).

localUnrankLoop(X,Y,N,R,Lo,Hi,As):-N2 is 2*N,X=<N2,!,
   PY is Y-1, SY is Y+1,
   binDif(N,X,SY,K), LK is Lo+K,
   ( R<LK -> NewHi is LK-1, NewLo is Lo, NewY is SY, Digit=0
   ; NewLo is LK, NewHi is Hi, NewY is PY, Digit=1
   ),nth0(X,As,Digit),NewX is X+1,
   localUnrankLoop(NewX,NewY,N,R,NewLo,NewHi,As).
localUnrankLoop(_X,_Y,_N,_R,_Lo,_Hi,_As). 
\end{code}

\begin{code}
unrankLoop(R,S,I,FinalS,FinalI):-cat(I,C),NewS is S+C, NewS=<R,
   !,NewI is I+1,
   unrankLoop(R,NewS,NewI,FinalS,FinalI).
unrankLoop(_,S,I,S,I).
\end{code}

\subsection*{The bijection between finite lists and sets}
The bijection {\tt list2set} together with its inverse 
{\tt set2list} are defined as follows:
\begin{code}
list2set(Ns,Xs) :- list2set(Ns,-1,Xs).

list2set([],_,[]).
list2set([N|Ns],Y,[X|Xs]) :- 
  X is (N+Y)+1, 
  list2set(Ns,X,Xs).

set2list(Xs,Ns) :- set2list(Xs,-1,Ns).

set2list([],_,[]).
set2list([X|Xs],Y,[N|Ns]) :- 
  N is (X-Y)-1, 
  set2list(Xs,X,Ns).
\end{code}
The following examples illustrate this bijection:
\begin{codex}
?- list2set([2,0,1,4],Set),set2list(Set,List).
Set = [2, 3, 5, 10],
List = [2, 0, 1, 4].
\end{codex}
As a side note, this bijection is mentioned
in \cite{knuth_comb}
with indications that
it might even go back to the early days of the theory of 
recursive functions.

\subsection*{Binomial Coefficients, efficiently}

Binomial coefficients are given by the formula
${n \choose k} = {{n!} \over {k!(n-k)!}} = 
{n(n-1) \ldots (n-(k-1)) \over k!}$.
By performing divisions as early as possible
to avoid generating excessively large intermediate results,
one can derive the {\tt binomialLoop} tail-recursive predicate:
\begin{code}
binomialLoop(_,K,I,P,R) :- I>=K, !, R=P.
binomialLoop(N,K,I,P,R) :- I1 is I+1, P1 is ((N-I)*P) // I1,
   binomialLoop(N,K,I1,P1,R).
\end{code}
The  predicate {\tt binomial(N,K,R)} computes
$N \choose K$ and unifies the result with {\tt R}. 
\begin{code}
binomial(_N,K,R):- K<0,!,R=0.
binomial(N,K,R) :- K>N,!, R=0.
binomial(N,K,R) :- K1 is N-K, K>K1, !, binomialLoop(N,K1,0,1,R).
binomial(N,K,R) :- binomialLoop(N,K,0,1,R).
\end{code}

\begin{codeh}

nv(X):-numbervars(X,0,_).

pn(X):-numbervars(X,0,_),write(X),nl,fail.
pn(_).

lb(B):-b2l(B,T),lamshow(T),nl. 
tb(B):-b2l(B,T),texshow(T),nl. 

lx(X):-c2b(X,B),b2l(B,T),lamshow(T),nl. 
tx(X):-c2b(X,B),b2l(B,T),texshow(T),nl. 

lamshow(T):-
  numbervars(T,0,_),
  lamshow(T,Cs,[]),
  maplist(write,Cs),
  fail.
lamshow(_).

lamshow('$VAR'(I))--> [x],[I].
lamshow(l('$VAR'(I),E))-->[('(')],[('\\')],[x],[I],[('.')],lamshow(E),[(')')].
lamshow(a(X,Y))-->[('(')],lamshow(X),[(' ')],lamshow(Y),[(')')].

texshow(T):-
  numbervars(T,0,_),
  texshow(T,Cs,[]),
  maplist(write,Cs),
  fail.
texshow(_).

texshow('$VAR'(I))--> [x],['_'],[I].
texshow(l('$VAR'(I),E))-->[(' ')],[('~\\lambda ')],[x], ['_'],   [I],[('.')],texshow(E),[(' ')].
texshow(a(X,Y))-->[('(')],texshow(X),[('~')],texshow(Y),[(')')].

sols(Goal, Times) :-
        Counter = counter(0),
        (   Goal,
            arg(1, Counter, N0),
            N is N0 + 1,
            nb_setarg(1, Counter, N),
            fail
        ;   arg(1, Counter, Times)
        ).

count1(F,M,Ks):-findall(K,(between(0,M,L),sols(call(F,L),K)),Ks).

count2(F,M,Ks):-findall(K,(between(0,M,L),sols(call(F,L,_),K)),Ks).

count3(F,M,Ks):-findall(K,(between(0,M,L),sols(call(F,L,_,_),K)),Ks).

\end{codeh}

\end{document}